\documentclass[pdflatex,sn-mathphys-num]{sn-jnl}% Math and Physical Sciences Numbered Reference Style 
\usepackage{mathtools}
\usepackage{lipsum}
\usepackage{mathtools}
\usepackage{multirow}
\usepackage{makecell}
\usepackage{graphicx}%
\usepackage{multirow}%
\usepackage{amsmath,amssymb,amsfonts}%
\usepackage{amsthm}%
\usepackage{mathrsfs}%
\usepackage[title]{appendix}%
\usepackage{xcolor}%
\usepackage{textcomp}%
\usepackage{manyfoot}%
\usepackage{booktabs}%
\usepackage{algorithm}%
\usepackage{algorithmicx}%
\usepackage{algpseudocode}%
\usepackage{listings}%
\usepackage{manyfoot}
\usepackage{xcolor}
\usepackage{lscape}
\usepackage{multicol}%
\usepackage{textcomp}%
\usepackage{manyfoot}%
\usepackage{booktabs}%
\usepackage{array}
\usepackage{enumerate}
\usepackage{hyperref}
\usepackage{ mathrsfs }
\usepackage{autobreak}
\usepackage{amsmath}
\usepackage{mathtools}
\usepackage{upgreek}
\usepackage{multirow}
\usepackage{nicematrix}
\usepackage{tabularray}
\usepackage{adjustbox}
\usepackage{ tipa }
\usepackage{tabularx}
\usepackage{tabulary}
\usepackage{ragged2e}
\usepackage{array}
\usepackage{multicol}
\usepackage{nicematrix}
\usepackage{booktabs}
\usepackage{upgreek}
\usepackage{autobreak}
\usepackage{makecell}
\usepackage{float}
\usepackage{mathrsfs}
\usepackage{amsthm}
\newtheorem{lemma}{Lemma}
\newtheorem{observation}{Observation}
\theoremstyle{thmstyleone}%
\newtheorem{theorem}{Theorem}%  meant for continuous numbers
\newtheorem{corollary}{Corollary}%
\theoremstyle{thmstyletwo}%
\newtheorem{example}{Example}%
\newtheorem{remark}{Remark}%

\theoremstyle{thmstylethree}%

\begin{document}

\title[Article Title]{Construction of cyclic codes with large minimum distance from power functions over odd characteristic finite fields}


\author[1]{\fnm{Mrinal Kanti} \sur{Bose}}\email{21dr0111@mc.iitism.ac.in}

%\author[2]{\fnm{Udaya} \sur{Parampalli}}\email{udaya@unimelb.edu.au}

\author*[1]{\fnm{Abhay Kumar} \sur{Singh}}\email{abhay@iitism.ac.in}

\affil[1]{\orgdiv{Department of Mathematics and Computing }, \orgname{Indian Institute of Technology (ISM)} \orgaddress{\city{Dhanbad}, \postcode{826004}, \state{Jharkhand}, \country{India}}}

%\affil[2]{\orgdiv{School of Computing and Information Systems}, \orgname{The University of Melbourne}, \orgaddress{ \city{Parkville}, \state{Victoria}, \country{Australia}}}


\abstract{Cyclic codes with dimensions exceeding half of the code length and minimum distance greater than the square root of the code length are of significant interest due to their high transmission efficiency and strong error-correcting capability. Such codes are well suited for demanding applications, including communication and storage systems, post-quantum cryptography, radar and sonar systems, wireless sensor networks, and space communications. Motivated by the work of Ding \cite{P3}, this paper extends the binary framework of Ding and Zhou \cite{P2} to a non-binary setting. By employing power functions with known differential uniformity over finite fields of odd characteristic, we present several infinite families of $q$-ary cyclic codes of length $q^m-1$ with dimensions exceeding $(q^m-1)/2$ and the lower bounds on the minimum distances greater than the square root of the code length, thereby achieving a favorable balance between code rate and error-correcting capability. We also determine the exact minimum distance of some of these codes. Furthermore, we partially resolve Open Problem $5.31$ posed by Ding in \cite{P3}.   %    In this article, we utilize permutation polynomials over $\mathbb{F}_{q}$ of small degree, namely of degree four, five and six, to construct cyclic codes through the trace functions and periodic sequences proposed by Ding [].   %their 
%weight distributions are determined using Weil sums. Some of 
%the linear codes obtained are optimal or almost optimal with 
%respect to the Griesmer bound.
}


%\pacs[MSC Classification]{94B05, 11T71, 11T23}



\maketitle

\section{Introduction}\label{sec1}
Let $p$ be a prime and $q=p^{e}$, where $e$ is a positive integer. Let $\mathbb{F}_{q}$ be a field with $q$ elements and $\mathbb{F}_{q}^{*}=\mathbb{F}_{q}\setminus\{0\}$. A $q$-ary linear $[v,k,d]$ code $\mathcal{C}$ of length $v$ is a $k$-dimensional subspace of $\mathbb{F}_{q}^{v}$ equipped with a minimum nonzero Hamming distance $d$. A code $\mathcal{C}$ with minimum Hamming distance $d$ is said to be $t$-error-correcting if it satisfies $d\geq 2t+1$. A $q$-ary $[v,k,d]$ linear code $\mathcal{C}$ is said to be optimal if there is no such $[v,k,d']$ code over $\mathbb{F}_{q}$ with $d'\geq d+1$, and almost optimal if the $q$-ary $[v,k,d+1]$ code is optimal. A linear code $\mathcal{C}$ over $\mathbb{F}_{q}$ is said to be cyclic if a codeword $(a_{0},a_{1},\cdots,a_{v-1})\in\mathcal{C}$ implies that its cyclic shift $(a_{v-1},a_{0},\cdots,a_{v-2})\in\mathcal{C}$. Cyclic codes have a simple representation in terms of ideals in the polynomial algebra $\mathbb{F}_{q}[x]$. Assuming that $\operatorname{gcd}(v,p)=1$, any codeword $(a_{0},a_{1},\cdots,a_{v-1})$ in $\mathcal{C}$ can be identified with an unique polynomial of the ring $\mathcal{R}:=\mathbb{F}_{q}[x]/(x^v-1)$ by the one-to-one correspondence $\sigma:\mathcal{C}\rightarrow\mathcal{R}$ as $(a_{0},a_{1},\cdots,a_{v-1})\mapsto\sum_{i=0}^{v-1}a_{i}x^{i}$. It is well known that every ideal of $\mathcal{R}$ is principal, and the cyclic code $\mathcal{C}$ of length $v$ is an ideal of $\mathcal{R}$. Then, there exists an unique monic polynomial $g(x)$ of least degree such that $\sigma(\mathcal{C})=g(x)\mathcal{R}$ and $g(x)\mid (x^v-1)$. The polynomial $g(x)$ is called the generator polynomial and $h(x):=(x^v-1)/g(x)$ is called the parity-check polynomial of $\mathcal{C}$. The dual code of $\mathcal{C}$ is also cyclic, which is denoted by $\mathcal{C}^{\perp}$, and generated by the reciprocal of the parity-check polynomial $h(x)$. Let $m$ be a positive integer and $C_{j}$ be the $q$-cyclotomic coset modulo $v$ containing $j$, and is defined as $C_{j}=\{j,jq,jq^2,\cdots,jq^{\ell_{j}-1}\}\text{ mod }v$, where $j\in\mathbb{Z}_{v}=\{0,1,2,\cdots,v-1\}$, $\ell_{j}$ is the size of $C_{j}$, and is the smallest positive integer such that $j\equiv jq^{\ell_{j}}\pmod{v}$. Clearly $C_{j}\subset\mathbb{Z}_{v}$ and if $v+1$ is in the set $C_{j}$, we consider $v+1$ as $1$. Let $m=\operatorname{ord}_{v}(q)$ be the multiplicative order of $q$ modulo $v$ and $\alpha$ be a generator of the multiplicative group $\mathbb{F}_{q^m}^{*}$. Then, $\beta=\alpha^{\frac{q^m-1}{v}}$ is a primitive $v$-th root of unity in $\mathbb{F}_{q^m}$. We know that $m_{\beta^{i}}(x)=\prod_{s\in C_{i}}(x-\beta^{s})$ is an irreducible polynomial of degree $|C_{i}|$ over $\mathbb{F}_{q}$, and refer to as the minimal polynomial of $\beta^{i}$ over $\mathbb{F}_{q}$. Let $\mathcal{C}$ be a cyclic code over $\mathbb{F}_{q}$ with generator polynomial $g(x)=\operatorname{lcm}(m_{\beta^{b}}(x), m_{\beta^{b+1}}(x),\cdots,m_{\beta^{b+\delta-2}}(x))$, where $2\leq \delta\leq v$ and $b$ is an integer. Then the code $\mathcal{C}$ is called a Bose-Chaudhuri-Hocquenghem (BCH) code of length $v$ with designed distance $\delta$, is denoted by $\mathcal{C}_{(q,v,\delta,b)}$ and the set $\mathcal{Z}=\{j\in\mathbb{Z}_{v}:g(\beta^{j})=0\}$ is called the defining set of the code $\mathcal{C}$. When $b=1$, the code $\mathcal{C}$ is a narrow-sense BCH code, and if $v=q^m-1$, the code is called a primitive BCH code. Furthermore, for a BCH code $\mathcal{C}$, its minimum distance $d$ satisfies $d\geq\delta$.
\par Cyclic codes can serve as a foundation for the development of various intriguing structures, such as quantum codes \cite{Quantum1} and frequency-hopping sequences \cite{Frequency1}, etc. Thus, the construction of cyclic codes with good parameters has been of great interest for the past few decades. One way of constructing cyclic codes over $\mathbb{F}_{q}$ of length $v$ uses the generator polynomial $\mathcal{G}(x):=\frac{x^v-1}{\operatorname{gcd}(S_{v}(x),x^v-1)}$, where $S_{v}(x)=\sum_{t=0}^{v-1}s_{t}x^t\in\mathbb{F}_{q}[x]$ and $s^{\infty}=(s_{t})_{t=0}^{\infty}$ is a sequence of period $v$ over $\mathbb{F}_{q}$. The cyclic code defined by the sequence $s^{\infty}$ with the generator polynomial $\mathcal{G}(x)$ is denoted by $\mathcal{C}_{s}$, and the sequence $s^{\infty}$ is called the defining sequence of the cyclic code $\mathcal{C}_{s}$. Another popular method of constructing cyclic codes is through the minimal polynomial approach. A $q$-ary cyclic code $\mathcal{C}_{(d_{1},d_{2},\cdots,d_{k})}$ is obtained by defining its generator polynomial as $m_{\alpha^{d_1}}(x)m_{\alpha^{d_2}}(x)\cdots m_{\alpha^{d_k}}(x)$, where $k$ is a positive integer, $0\leq d_1,d_2,\cdots,d_k\leq q^m-2$ and $C_{e_{i}}$ and $ C_{e_j}$ are pairwise disjoint for any distinct $i$ and $j$ in $\mathbb{Z}_{q^m-1}$. Surprisingly, these cyclic code can be used to characterize the cryptographic properties of functions over finite fields \cite{High1,High2,PN1}. The cyclic code with two zeros $\mathcal{C}_{(1,d)}$ and the two classes of subcodes of $\mathcal{C}_{(1,d)}$, namely  $\mathcal{C}_{(0,1,d)}$ and $\mathcal{C}_{(1,d,\frac{p^m-1}{2})}$, had been well studied in the literature \cite{Zero1,Zero2,Zero3,Zero4,Zero5,Zero6,Zero7,Zero8,Zero9}. %Surprisingly, cryptographic functions play an important role in their construction. For more information, the readers can refer to . 
\vskip 1pt
For any function $F:\mathbb{F}_{q}\rightarrow\mathbb{F}_{q}$, the \textit{derivative function} of $F$ at an element $a$ in $\mathbb{F}_{q}$  is defined by
\begin{equation*}
    \mathbb{D}_{a}(F(x))= F(x+a)-F(x),\text{ for all }x\in\mathbb{F}_{q}.
\end{equation*}
For any $a,b\in\mathbb{F}_{q}$, let $\delta_{F}(a,b)$ denote the number of solutions to the differential equation $\mathbb{D}_{a}(F(x))=b$. The \textit{differential uniformity} of $F$ is defined as   
\begin{equation*}
    \Delta_{F}= \operatorname{max}\left\{\delta_{F}(a,b)\text{ }|\text{ }a\in\mathbb{F}_{q}^{*}\text{ and }b\in\mathbb{F}_{q}\right\}.
\end{equation*}
Then, we say that $F$ is a \textit{differentially $\Delta_{F}$-uniform function}. When $F$ is used as an $S$-box in a cryptosystem, low differential uniformity (i.e., a smaller value of $\Delta_{F}$) indicates stronger resistance against differential attacks \cite{Biham1,Biham2}. If $\Delta_{F}=1$, then $F$ is called a \textit{planar or perfect nonlinear (PN) function}. When $\Delta_{F}=2$, $F$ is called an \textit{almost perfect nonlinear (APN) function}. For more information on PN and APN functions, the readers can refer to \cite{PN2,PN3,PN4,PN5}.

\par In 2013, Ding \cite{P3} introduced another useful representation for cyclic codes through the trace functions and an infinite sequence. For any given polynomial $F(x)$ over $\mathbb{F}_{q^m}$ and a fixed primitive element $\alpha$ of $\mathbb{F}_{q^m}$, Ding \cite{P3} defined its associated sequence $s^{\infty}$ as follows 
\begin{equation}\label{Eq3}
    s_{t}=\operatorname{Tr}_{q}^{q^m}(F(\alpha^t+1)),\text{ for all }t\geq 0.
\end{equation}
One can check that $s^{\infty}=(s_{t})_{t=0}^{\infty}$ is a periodic sequence over $\mathbb{F}_{q}$ with period $q^m-1$. Throughout this paper, the cyclic code generated by the minimal polynomial of the sequence $s^{\infty}$ is denoted by $\mathcal{C}_{F}$. 
\par Ding \cite{P3} and Ding and Zhou \cite{P2} initiated a systematic study on the selection of polynomials $F(x)$ over $\mathbb{F}_{q^m}$ for constructing cyclic codes with optimal or near-optimal parameters. In particular, they investigated monomials and trinomials over $\mathbb{F}_{q^m}$, many of which are permutation polynomials or functions with low differential uniformity, and derived several families of cyclic codes. They also determined the minimum weights for some of these codes and established tight lower bounds for others. Subsequently, Tang et al. \cite{Tang} resolved two open problems posed in \cite{P3,P2}. Rajabi and Khashyarmanesh \cite{Rajabi} further extended these results by constructing new cyclic codes and solving additional open problems from \cite{P3}. Li et al. \cite{Li} provided partial solutions to an open problem proposed in \cite{P2}. Moreover, Mesnager et al. \cite{Mesnager} complemented earlier work by studying cyclic codes arising from several known families of monomial functions with low differential uniformity and provided partial answers to three open problems in \cite{P2,P3}. Bose et al. \cite{Bose} constructed a new infinite family of optimal binary cyclic codes and obtained several families of cyclic codes with dimensions exceeding $(2^m-1)/2$ and minimum distances close to the square-root bound, using known classes of permutation monomials and trinomials over $\mathbb{F}_{2^m}$. Tiwari et al. \cite{Kewat} employed reversed Dickson polynomials of the first kind to construct cyclic codes, determine their exact minimum distances, and derive many optimal and near-optimal cyclic codes. They also obtained quantum error-correcting codes from these constructions via the Calderbank-Shor-Steane (CSS) construction. More recently, Tiwari and Kewat \cite{Kewat1} determined the exact minimum distance of three distinct classes of primitive BCH codes and addressed open problems regarding cyclic codes obtained from Dickson polynomials proposed by Ding \cite{Survey1}.  A comprehensive survey of developments over the past decade on sequence-based constructions of cyclic codes over finite fields can be found in \cite{Survey}. However, the construction of non-binary cyclic codes using power functions with known differential uniformity over finite fields of odd characteristic has received comparatively less attention in the literature. A summary of known power functions with their differential spectra is provided in \cite[Table~1]{List}.  % to be the $q$-ary cyclic code generated by the minimal polynomial of the sequence $s^{\infty}$.
\par In this article, we investigate several infinite families of $q$-ary cyclic codes by suitably selecting power functions with known differential uniformity over $\mathbb{F}_{q^m}$ via the trace sequence approach. By considering the minimal polynomials of sequences associated with these functions as generator polynomials, we construct several families of optimal and near-optimal cyclic codes and establish lower bounds on their minimum distances. In particular, some of the constructed cyclic codes of length $v=q^m-1$ have dimensions exceeding $v/2$ and minimum distances greater than $\sqrt{v}$. Furthermore, we partially resolve Open Problem~5.31 posed in \cite{P3} (see Remark~\ref{Remark1}).
\par The rest of this paper is structured as follows. In Section $\ref{sec2}$, we introduce some notations and auxiliary results. Section $\ref{sec3}$ includes several subsections, where we give the constructions of $q$-ary cyclic codes through suitably chosen power functions with known differential uniformity over odd characteristic finite fields. We provide the comparative analysis of our findings in relation to the existing literature in the remarks. Finally, Section $\ref{sec4}$ concludes the paper.  % in the remarks, we provide a comparative analysis of our findings in relation to the existing literature.

\section{Preliminaries}\label{sec2}
In this section, we first fix some notations, then present important results related to $q$-cyclotomic cosets, number theory, bounds on cyclic codes, and highlight the theory of linear feedback shift registers and their associated sequences over $\mathbb{F}_{q}$, which will be useful in the subsequent sections. %then present key results related to $q$-cyclotomic cosets. Additionally, we will outline the theory of linear feedback shift register structures and the associated sequences
\subsection{Notations}
We adhere to the following notations unless stated otherwise:
\begin{itemize}
     \item For any finite set $S$, we denote $|S|$ as the cardinality of $S$. 
    \item Let $q=p^e$ and $v=q^{m}-1$, where $p$ is a prime, and $e$ and $m$ are two positive integers. 
    \item $\mathbb{Z}_{v}=\{0,1,2,\cdots,v-1\}$ denotes the ring of integers modulo $v$ and $\mathbb{Z}_{v}^{*}=\mathbb{Z}_{v}\setminus\{0\}$.
    \item Let $\alpha$ be a generator of $\mathbb{F}_{q^m}^{*}$, and $m_{\alpha^{i}}(x)$ denotes the minimal polynomial of $\alpha^i\in\mathbb{F}_{q^m}^{*}$ over $\mathbb{F}_{q}$.  %For $u\in\mathbb{F}_{r}$ and $v\in\mathbb{F}_{r}^{*}$, we denote $uv^{-1}$ by $\frac{u}{v}$.
    \item Let $\operatorname{Tr}_{q}^{q^m}(x)$ be the trace function from $\mathbb{F}_{q^m}$ to $\mathbb{F}_{q}$, defined as
    $\operatorname{Tr}_{q}^{q^m}(x)=\sum_{i=0}^{m-1}x^{q^i}$.
    \item The map $\mathbb{N}_{p}(\cdot):\mathbb{N}\cup\{0\}\rightarrow\{0,1\}$ is defined by $\mathbb{N}_{p}(i)=0$ if $p\mid i$, and $\mathbb{N}_{p}(i)=1$ otherwise. 
    \item Let $C_{j}=\{jq^{s}:0\leq s\leq \ell_{j}-1\}\text{ (mod $v$)}\subset\mathbb{Z}_{v}$ be the $q$-cyclotomic coset modulo $v$ containing $j$, where $\ell_{j}$ is the size of $C_{j}$ and is the smallest positive integer such that $j\equiv jq^{\ell_{j}}\pmod{v}$. 
    \item  The smallest integer in $C_{j}$ is called the coset leader of $C_{j}$. We denote $\Gamma$ to be the set of all coset leaders. That is $\bigcup_{j\in\Gamma}C_{j}=\mathbb{Z}_{v}$.
    \item Let $j\in\{0\leq i\leq p^t-1:\text{ $i$ is an integer}\}$, where $t$ is a positive integer. If the $p$-adic expansion of $j$ is $j=\sum_{s=0}^{t-1}j_{s}p^{s}$, then define a map $\beta(\cdot):\{0,1,2,\cdots,p^t-1\}\rightarrow\mathbb{Z}_{p}^{*}$ such that $\beta(j)=\prod_{s=0}^{t-1}\binom{p-1}{j_{s}}\text{ (mod $p$)}$.
    \item By the Database, we mean the collection of the best known linear codes maintained by Markus Grassl at \url{https://www.codetables.de/}.
\end{itemize}

\subsection{Essential number-theoretic results and some bounds}
For a non-negative integer $j$ with $0\leq j\leq q^m-1$, the $q$-adic expansion of $j$ is defined as 
\begin{equation*}
    j= a_{0}+a_{1}q+a_{2}q^{2}\cdots+a_{m-1}q^{m-1},
\end{equation*}
where $0\leq a_{0},a_{1},\cdots,a_{m-1}\leq q-1$ and $j$ can be identified by a vector $\bar{a}=(a_{0},a_{1},\cdots,a_{m-1})$. The $q$-weight of $j$ is denoted by $\operatorname{wt}_{q}(j)$ and is defined as %$\operatorname{wt}_{q}(j)=|\{0\leq s\leq m-1: a_{s}\neq 0\}|$.
\begin{equation*}
    \operatorname{wt}_{q}(j)=|\{0\leq s\leq m-1: a_{s}\neq 0\}|.
\end{equation*}
%\vskip 1pt
The following lemmas will be useful in the subsequent sections.
%For any given polynomial $F(x)$ over $\mathbb{F}_{q^m}$ and a fixed generator $\alpha$ of $\mathbb{F}_{q^m}^{*}$, Ding \cite{P3} defined its associated sequence $s^{\infty}$ by 
%\begin{equation}\label{Eq3}
%    s_{t}=\operatorname{Tr}_{q}^{q^m}(F(\alpha^t+1)),\text{ for all }t\geq 0.
%\end{equation}
%One can check that $s^{\infty}=(s_{t})_{t=0}^{\infty}$ is a periodic sequence over $\mathbb{F}_{q}$ with period $q^m-1$. Throughout this paper, we denote $\mathcal{C}_{F}$ to be the $q$-ary cyclic code generated by the minimal polynomial of the sequence $s^{\infty}$.
\vspace{1em}
\begin{lemma}\label{L4}
   Let $\Gamma^{*}=\Gamma\backslash\{0\}$ and $\Delta=\{\sum_{i=0}^{m-1}a_{i}q^{i}:a_{i}\in\{0,1\}\text{ for all }i=0,1,2,\cdots,m-1\}$. Then
   \begin{itemize}
       \item For any coset leader $j$ in $\Gamma^{*}$, we have $q\nmid j$.
       \item For any $j\in\Gamma^{*}\cap\Delta$, we have $1\leq j\leq q^{m-1}-1$, except that $j=\frac{q^m-1}{q-1}$.
   \end{itemize}  
   \begin{proof}
       Suppose $q\mid j$. Note that $j\equiv \frac{j}{q}\cdot q\pmod{v}$. That means $\frac{j}{q}\in C_{j}$. Since $q\leq j\leq q^m-q$, we have $1\leq \frac{j}{q}<j\leq q^m-q$, which contradicts the fact that $j$ is the coset leader of $C_{j}$. 
       \par It is obvious that $\frac{q^m-1}{q-1}$ is the largest integer in $\Delta$. Suppose $q^{m-1}\leq j\leq \frac{q^m-1}{q-1}$ for some $j\in\Gamma^{*}\cap\Delta$. Then $q^m\leq qj\leq\frac{q^{m+1}-q}{q-1}$, which gives $q^m-v\leq qj-v\leq \frac{q^{m+1}-q}{q-1}-v$. Thus, $1\leq qj-v\leq \frac{q^m-1}{q-1}$. However, since $qj-v\in C_{j}$ and $qj-v<j$ for $j<\frac{q^m-1}{q-1}$, which again contradicts the definition of $\Gamma^{*}$. For $q=2$, $j<\frac{q^m-1}{q-1}$ is necessary. Hence, the proof.
   \end{proof}
\end{lemma}
%\vspace{1em}
\begin{lemma}\textnormal{\cite[Lemma 8]{P-1}}\label{L3}
     Let $v$ be a positive integer such that $\operatorname{gcd}(v,q)=1$ and $q^{\lfloor\frac{m}{2}\rfloor}<v\leq q^m-1$, where $m=\operatorname{ord}_{v}(q)$. For $1\leq x\leq\frac{vq^{\lceil\frac{m}{2}\rceil}}{q^m-1}$, the $q$-ary cyclotomic cosets of $x$ modulo $v$, $C_{x}$ is of size $m$. In other words, if $v=q^m-1$, the $q$-cyclotomic coset $C_{x}$ is of size $m$ for all $1\leq x\leq q^{\lceil\frac{m}{2}\rceil}$.
\end{lemma}
%\vspace{0.5em}
%\begin{remark}
 %   Particularly for $N=q^m-1$, since $\operatorname{gcd}(N,q)=1$, it can be derived from Lemma $\ref{L3}$ that the $q$-cyclotomic coset $C_{x}$ is of size $m$, for all $1\leq x\leq q^{\lceil\frac{m}{2}\rceil}$.
%\end{remark}
\vspace{0.5em}
%\begin{lemma}
 %   Let $t$ be a positive integer with $1\leq t\leq\lceil\frac{m}{2}\rceil$ and $\Gamma^{'}=\{\sum_{i=0}^{t-1}a_{i}q^{i}:a_{i}\in\{0,1\}\text{ for all }i=0,1,2,\cdots,t-1\}$. Then
 %   \begin{itemize}
 %       \item For any $j\in\Gamma^{'}$ with $q\nmid j$, $j$ is the coset leader of $C_{j}$;
  %      \item $C_{j_{1}}\cap C_{j_{2}}=\emptyset$ for any pair of distinct $j_{1}$ and $j_{2}$, not divisible by $q$, in $\Gamma^{'}$.
  %  \end{itemize}
%    \begin{proof}
%      Let $j\in\Gamma^{'}$ be such that $q\nmid j$ and $\operatorname{wt}_{q}(j)=r$. Then $j=1+q^{j_{1}}+q^{j_{2}}+\cdots+q^{i_{r-1}}$, where $1\leq j_{1}<j_{2}<\cdots<j_{r-1}\leq t-1$. According to Lemma $\ref{L4}$, the coset leader of $C_{j}$ must not be divisible by $q$. The list of all integers in $C_{j}$ that are not divisible by $q$ are $j, jq^{m-j_{r-1}}\pmod{n},jq^{m-j_{r-2}}\pmod{n},\cdots,jq^{m-j_{1}}\pmod{n}$. These values must be pairwise distinct; otherwise, it would contradict the fact that $|C_{j}|=m$. Hence, the proof of the first assertion follows due to $j<jq^{m-j_{s}}\pmod{n}$ for all $s=1,2,\cdots,r-1$.
 %     \par The second assertion directly follows from the fact that $j_{1}$ and $j_{2}$ are distinct coset leaders of $C_{j_{1}}$ and $C_{j_{2}}$, respectively.
%    \end{proof}
%\end{lemma}
\begin{lemma}\label{L7}
    Let $q$ be the power of a prime $p$ and $t$ be a positive integer with $1\leq t\leq\lceil\frac{m}{2}\rceil$. If $\widehat{\Gamma}=\{a_{0}+\sum_{i=1}^{t-1}a_{i}q^{i}:a_{0}\neq 0\text{ and }a_{i}\in\{0,1,2,\cdots,p-1\}\text{ for all }i=0,1,\cdots,t-1\}$. Then 
     \begin{itemize}
        \item For any $j\in\widehat{\Gamma}$, $j$ is the coset leader of $C_{j}$;
        \item $C_{j_{1}}\cap C_{j_{2}}=\emptyset$ for any pair of distinct $j_{1}$ and $j_{2}$ in $\widehat{\Gamma}$.
    \end{itemize}
    \begin{proof}
        Let $j\in\widehat{\Gamma}$ and $\operatorname{wt}_{q}(j)=r$. Then $j=a_{0}+a_{j_{1}}q^{j_{1}}+a_{j_{2}}q^{j_{2}}+\cdots+a_{j_{r-1}}q^{j_{r-1}}$, where $1\leq j_{1}<j_{2}<\cdots<j_{r-1}\leq t-1$. According to Lemma $\ref{L4}$, the coset leader of $C_{j}$ must not be divisible by $q$. The list of all integers in $C_{j}$ that are not divisible by $q$ is $j, jq^{m-j_{r-1}}\pmod{v},jq^{m-j_{r-2}}\pmod{v},\cdots,jq^{m-j_{1}}\pmod{v}$. These values must be pairwise distinct; otherwise, due to $1\leq j\leq q^{\lceil\frac{m}{2}\rceil}$, it would contradict the fact that $|C_{j}|=m$. Hence, the proof of the first assertion follows due to the fact $j<jq^{m-j_{s}}\pmod{v}$ for all $s=1,2,\cdots,r-1$.
      \par The second assertion directly follows from the fact that $j_{1}$ and $j_{2}$ are distinct coset leaders of $C_{j_{1}}$ and $C_{j_{2}}$, respectively.
    \end{proof}
\end{lemma}
\vspace{0.5em}
\begin{lemma}\textnormal{\cite[Lemma $15$]{Kewat1}}\label{L21}
    Let $p$ be a prime and $t$ be a positive integer. Let $u$ and $v$ be positive integers such that $v<u\leq p^t-1$. Then $\binom{u(p^t-1)}{v(p^t-1)}\equiv 0\textnormal{ (mod $p$)}$.
\end{lemma}
\vspace{0.5em}
\begin{lemma}\textnormal{(BCH bound \cite{P0})}
    Let $\mathcal{C}$ be a cyclic code of length $v$ over $\mathbb{F}_{q}$ with defining set $\mathcal{Z}$ and minimum distance $d$. Assume $\mathcal{Z}$ contains $\delta-1$ consecutive integers for some integer $\delta$. Then $d\geq\delta$.
\end{lemma}
\vspace{0.5em}
\begin{lemma}\textnormal{(Hartmann-Tzeng bound \cite{Hartmann})}
    Let $\mathcal{C}$ be a cyclic code of length $v$ over $\mathbb{F}_{q}$ with defining set $\mathcal{Z}$ and minimum distance $d$. Let $A$ be a set of $\delta-1$ consecutive elements of $\mathcal{Z}$ and $B=\{jb\text{ mod $v$}:0\leq j\leq s\}$, where $\operatorname{gcd}(b,v)<\delta$. If $A+B\subseteq\mathcal{Z}$ for some $b$ and $s$, then $d\geq\delta+s$.
\end{lemma}
%In this section, we present some preliminaries and important results related to exponential sums over finite fields, cyclotomic fields, and weakly regular bent functions, which will be used in the subsequent sections.
 %In this section, we will review some basic notation and results on cyclotomic fields, Gauss
 %sums and weakly regular bent functions, which will be used to prove our main results in the
% sequel. 
\subsection{Linear feedback shift register (LFSR) in $\mathbb{F}_{q}$}
A sequence $s^{\infty}=(s_{i})_{i=0}^{\infty}$ of period $v$ over $\mathbb{F}_{q}$ satisfying the relation $-a_{0}s_{i}=a_{1}s_{i-1}+a_{2}s_{i-2}+\cdots+a_{L}s_{i-L}$ for all $L\leq i\leq v-1$, where $a_{0}=1$, $a_{1},\cdots,a_{L}\in\mathbb{F}_{q}$, and $a_{L}\neq 0$, is called an $L$-th order linear recurring sequence in $\mathbb{F}_{q}$. A linear feedback shift register (LFSR) of length $L$ consists of $L$ delay elements $D_{0},D_{1},\cdots,D_{L-1}$, containing the values $s_{t-1},s_{t-2},\cdots,s_{t-L}$ respectively at time step $t$, where $t\geq L$. The initial state of the LFSR is the vector $(s_{L-1}, s_{L-2}, \dots, s_{0})$. During a clock cycle, a new element $s_t$ is calculated via the recurrence relation. This feedback element $s_t$ becomes the new content of $D_0$, while the previous content of each $D_i$ is shifted into $D_{i+1}$. The value previously in $D_{L-1}$ (which is $s_{t-L}$) is pushed out of the register. The polynomial $\sum_{i=0}^{L}a_{i}x^{i}$ is called a characteristic polynomial (or feedback polynomial) of the LFSR that generates $s^{v}=s_{0}s_{1}\cdots s_{v-1}$. The characteristic polynomial with the smallest degree is the minimal polynomial of the periodic sequence $s^{\infty}$. The degree of the minimal polynomial of the sequence $s^{\infty}$ is referred to as the linear span (or linear complexity) of $s^{\infty}$.
  \par Throughout this paper, we denote the minimal polynomial of $s^{\infty}$ as $\mathfrak{M}_{s}(x)$ and the linear span of $s^{\infty}$ as $\mathcal{L}_{s}$. The cyclic code generated by the minimal polynomial of the periodic sequence $s^{\infty}$ defined in Eq. $(\ref{Eq3})$ through the polynomial $F(x)$ over $\mathbb{F}_{q^m}$ is referred to as $\mathcal{C}_{F}$.
  \par The following well-known Lemma \cite{P-2} provides an efficient way to determine the minimal polynomial $\mathfrak{M}_{s}(x)$ and the linear span $\mathcal{L}_{s}$ of any sequence of period $q^m-1$ over $\mathbb{F}_{q}$.
%   "A linear feedback shift register (LFSR) of length $L$ consists of $L$ delay elements $D_0, \dots, D_{L-1}$. At time $t$, these stages contain the values $s_{t-1}, \dots, s_{t-L}$ respectively. During a clock cycle, a new element $s_t$ is calculated via the recurrence relation. This feedback element $s_t$ becomes the new content of $D_0$, while the previous content of each $D_i$ is shifted into $D_{i+1}$. The value previously in $D_{L-1}$ (which is $s_{t-L}$) is pushed out of the register."                     %The content of stage $D_i$ is $s_{t-i}$, $i\geq L$ at clock step $t$. %During each unit of time, the content of the stage    %The polynomial $p(x)=a_{0}+a_{1}x+\cdots+a_{L}x^{L}$ over $\mathbb{F}_{q}$ is
\vspace{1em}
\begin{lemma}\label{L6}
    Any sequence $s^{\infty}$ over $\mathbb{F}_{q}$ of period $q^m-1$ has a unique expansion of the form
    \begin{equation*}
        s_{t}=\sum_{i=0}^{q^m-2}c_{i}\alpha^{it},\text{ for all }t\geq 0,
    \end{equation*}
    where $c_{i}\in\mathbb{F}_{q^m}$. Let the index set be $I=\{0\leq i\leq q^m-2:c_{i}\neq 0\}$. Then the linear span $\mathcal{L}_{s}$ of $s^{\infty}$ is $|I|$ and the minimal polynomial $\mathfrak{M}_{s}(x)$ of $s^{\infty}$ is given by
    \begin{equation*}
        \mathfrak{M}_{s}(x)=\prod_{i\in I}(1-\alpha^{i}x).
    \end{equation*}
\end{lemma}
\begin{remark}
From the above discussion, we conclude that the generator polynomial of the $q$-ary cyclic code $\mathcal{C}_{F}$ is given by 
 \begin{equation*}
        \mathfrak{M}_{s}(x)=\prod_{i\in I\cap\Gamma}m_{\alpha^{-i}}(x),
    \end{equation*}
    where $\Gamma$ is the set of all coset leaders of the $q$-cyclotomic cosets modulo $v$.
\end{remark}
%\vspace{1em}
%\begin{lemma}
%    Let $i$ and $j$ be positive integers such that $i+j\leq q^e-1$. Then the $q$-ary cyclotomic cosets $C_{i}$ and $C_{i+j}$ are disjoint, i.e., $C_{i}\cap C_{i+j}=\emptyset$, for $q>j+1$.  
%\end{lemma}
%\vspace{1em}
%\begin{lemma}
%    Let $k$ be an integer such that $q>k>1$, then the $q$-ary cyclic code $\mathcal{C}=\langle\prod_{i=1}^{k-1}M_{i}(x)\rangle$ has the minimum distance $k$.
 %   \begin{proof}
 %       From the BCH bound, it is obvious that the minimum distance of $\mathcal{C}$ is at least $k$.
 %   \end{proof}
%\end{lemma}
%\section{Construction of cyclic codes }
\begin{table}[htbp]
    \centering
    \footnotesize % Slightly smaller base font to fit more content
    \renewcommand{\arraystretch}{1.5} % Increase row height
    \setlength{\tabcolsep}{10pt} % Increase column padding
    \begin{adjustbox}{max width=1.2\textwidth, center} % Allow slight width expansion
        \begin{tabular}{|>{\centering\arraybackslash}m{1.5cm}|>{\centering\arraybackslash}m{4.1cm}|>{\centering\arraybackslash}m{7cm}|>{\centering\arraybackslash}m{6cm}|>{\centering\arraybackslash}m{2cm}|}
            \hline
            \textbf{\normalsize $F(x)$} & \textbf{\normalsize Conditions} & \textbf{\normalsize $\dim(\mathcal{C}_{F})$} & \textbf{\normalsize $d$} & \textbf{\normalsize References} \\
            \hline
            $x^{q^{k}+1}$ & $\frac{m}{\operatorname{gcd}(m,k)}$ and $q$ are odd & $q^m-1-2m-\mathbb{N}_{p}(m)$ &  $\begin{cases} 
                d=4,\text{ if $q=3$ and $m\equiv 0\text{ (mod $p$)}$}, &  \\
              %  \text{if } h \text{ is even}, & \\
                4\leq d\leq 5,\text{ if $q=3$ and $m\not\equiv 0\text{ (mod $p$)}$}, & \\
                d=3,\text{ if $q>3$ and $m\equiv 0\text{ (mod $p$)}$}, \\
                3\leq d\leq 4,\text{ if $q>3$ and $m\not\equiv 0\text{ (mod $p$)}$}.
               % \text{if } h \text{ is odd}
            \end{cases}$ & Theorem 5.2 in \cite{P2} \\
            \hline
            $x^{\frac{q^t-1}{q-1}}$ & $1\leq t\leq\begin{cases}
                (m-1)/2,\text{ if $m$ is odd and} \\
                m/2,\text{ if $m$ is even.}
            \end{cases}$  & $\begin{cases}
                q^m-1-\mathcal{L}_{s},\text{ $\mathcal{L}_{s}$ given in Lemma $5.13$ of \cite{P2};} & \\
                \text{when $t\neq 3$,} \\
                q^m-1-4m-\mathbb{N}_{p}(m);\text{ when $t=3$ and $p\neq 3$,} \\
                q^m-1-3m-\mathbb{N}_{p}(m);\text{ when $t=3$ and $p=3$.}
            \end{cases}$   & For $t=3$, $\begin{cases}
                3\leq d\leq 8,\text{ if $p=3$ and $\mathbb{N}_{p}(m)=1$,} \\
                3\leq d\leq 6,\text{ if $p=3$ and $\mathbb{N}_{p}(m)=0$,} \\
                3\leq d\leq 8,\text{ if $p>3$.} \\
            \end{cases}$ &  Corollary $5.15$ in \cite{P2} \\
            \hline
            $x^{\frac{q^t-1}{q-1}}$ & $t$ is any integer  & $\begin{cases}
                q^m-1-\mathcal{L}_{s},\text{ $\mathcal{L}_{s}$ given in Lemma $3$ of \cite{Rajabi};} & \\
                \text{when $t\equiv 0\text{ (mod $m$)}$,} \\
                q^m-1-\mathcal{L}_{s},\text{ $\mathcal{L}_{s}$ given in Lemma $4$ of \cite{Rajabi};} & \\
                \text{when $t\not\equiv 0\text{ (mod $m$)}$.} \\
            \end{cases}$   &  $\begin{aligned}
&\text{For } t=m \text{ and } \mathbb{N}_{p}(m)=1,\\
&\begin{cases}
3\leq d\leq 2^m+1, & \text{if } q\neq 2,\\
4\leq d\leq 2^m, & \text{if } q=2.
\end{cases}\\
&\text{For } t=m-1 \text{ and }\mathbb{N}_{p}(i)=1 \\ & \text{ for } i\in\{1,2,\ldots,m\},\text{ }3\leq d\leq 2^m 
\end{aligned}$   &  Corollary $1$ and $2$ in \cite{Rajabi} \\
            \hline
            $x^{\frac{q^t+1}{2}}$ & $q=3$, $t$ is odd, $\operatorname{gcd}(m,t)=1$, and $3\leq t\leq\begin{cases}
                (m-1)/2,\text{ if $m$ is odd and} \\
                m/2,\text{ if $m$ is even.}
            \end{cases}$  & $\begin{cases}
                q^m-1-\mathcal{L}_{s},\text{ $\mathcal{L}_{s}$ given in Lemma $5.22$ of \cite{P2};} & \\
                 \text{when $t>3$,} \\
                q^m-1-7m-\mathbb{N}_{3}(m); \text{ when $t=3$.} 
            \end{cases}$     & For $t=3$, $\begin{cases}
                5\leq d\leq 16,\text{ if $\mathbb{N}_{3}(m)=1$,} \\
                4\leq d\leq 16,\text{ if $\mathbb{N}_{3}(m)=0$.} \\
              %  3\leq d\leq 8,\text{ if $p>3$.} \\
            \end{cases}$ & Corollary $5.24$ in \cite{P2} \\
            \hline
            $x^{\frac{q^m-3}{2}}$ & $q=3$ & $q^m-1-\mathcal{L}_{s}$, $\mathcal{L}_{s}$ given in Lemma $6$ of \cite{Rajabi} & $--$ & Lemma $6$ in \cite{Rajabi} \\
            \hline
            $x^{\frac{q^t+1}{2}}$ & $t$ is any integer and $q=p$, where $p$ is an odd prime & $\begin{cases}
                q^m-1-\mathcal{L}_{s},\text{ $\mathcal{L}_{s}$ given in Lemma $7$ of \cite{Rajabi};} & \\
                \text{when $q=p$ is arbitrary and $s$ is even,} \\
                q^m-1-\mathcal{L}_{s},\text{ $\mathcal{L}_{s}$ given in Lemma $8$ of \cite{Rajabi};} & \\
                \text{when $q=p$ is arbitrary and $s$ is an odd integer,} \\
                q^m-1-\mathcal{L}_{s},\text{ $\mathcal{L}_{s}$ given in Lemma $9$ of \cite{Rajabi};} & \\
                \text{when $q=5$, $\operatorname{gcd}(2m,t)=1$ and $s\in\{0,2,4\}$.} \\
                q^m-1-\mathcal{L}_{s},\text{ $\mathcal{L}_{s}$ given in Lemma $10$ of \cite{Rajabi};} & \\
                \text{when $q=5$, $\operatorname{gcd}(2m,t)=1$ and $s\in\{1,3\}$.} \\
            \end{cases}$ & $\begin{aligned}
&\text{For } l\in\{0,1\} \text{ and }\mathbb{N}_{p}(m)=1,\\
&\begin{cases}
3\leq d\leq m+2;\text{ if }l=0,\\
\frac{p+5}{2}\leq d\leq \frac{p+1}{2}m+2;\text{ if }l=1.
\end{cases}\\
&\text{For }\mathbb{N}_{p}(m)=1, \\ & \begin{cases}
    d=\frac{p+5}{2};\text{ if }m=1, \\
    \frac{p+5}{2}\leq d\leq p+3;\text{ if }(m,l)=(2,1).
\end{cases} 
\end{aligned}$ & Corollary $3$ and $4$ in \cite{Rajabi} \\
            \hline
            $x^{q^m-2}$ & $q$ is a power of any odd prime $p$ & $\frac{q^m}{p}-1$ & $d\geq\operatorname{max}\{2p-1,\frac{q(p-1)}{p}+1\}$ & Theorem $4$ in \cite{Tang} \\
            \hline 
            $x^{2\cdot q^{\frac{m-1}{2}}+1}$ & $q=3$ and $m$ is odd & $q^m-1-3m-\mathbb{N}_{3}(m)$ & $3\leq d\leq 8$ & Theorem $\ref{Th12}$ \\
            \hline
            $x^{q^{m/2}+2}$ & $m$ is even and $q=p$, where $p$ is an odd prime & $\begin{cases}
                q^m-\frac{5m}{2}-1-\mathbb{N}_{p}(m);\text{ if $p=3$,} \\
                q^m-\frac{7m}{2}-1-\mathbb{N}_{p}(m);\text{ if $p\geq 5$.} 
            \end{cases}$ & $\begin{cases}
           3\leq d \leq 6; \text{ if }p=3, \\  
           5\leq d\leq 8; \text{ if }p\geq 5 \text{ and }p\nmid m, \\
           4\leq d\leq 8; \text{ if }p\geq 5 \text{ and }p\mid m.
        \end{cases}$ & Theorem $\ref{Th7}$ \\
            \hline
     %       \multirow{2}{*}{$x^{q^{m/2}+2}$} & $m$ is even and $q=p$, where p is an odd prime & $2^{m}-2-m$ & 4 & \multirow{2}{*}{Theorem 1 in \cite{SETA8}}\\
         %   \cline{2-4}
     %       & $m \equiv 2 \pmod{4}$ and $\operatorname{gcd}(m,h)=2$ & $2^{m}-1-m$ & 3 & \\
     %       \hline
            $x^{\frac{q^{t}-1}{q-1}}$  & $t=\frac{m+1}{2}$ and $m$ is odd & $q^{m}-1-\mathbb{N}_{p}(m)-\left(\mathbb{N}_{p}(t)+\sum_{k=1}^{t-1}\mathbb{N}_{p}(k)2^{t-k-1}\right)m$ &  $d \geq \begin{cases}
                 2^{t-2}+2, \text{ if $q=2$ and $t>2$;} \\
                 3,\text{ if $p=2$ with $e\geq 2$, and $p\nmid t$;} \\
                  2,\text{ if $p=2$ with $e\geq 2$, and $p\mid t$;} \\
                4,\text{ if $p\neq 2$, $p\nmid t$ and $t>2$;} \\
                3,\text{ if $p\neq 2$, $p\mid t$ and $t>2$.}
            \end{cases}$  & Theorem $\ref{Th10}$ \\
            \hline
            $x^{\frac{q^{t}+1}{2}}$ & $q=3$, $t=\frac{m+1}{2}$ and $m$ is odd & $q^{m}-1-\mathbb{N}_{3}(m)-\left(\mathbb{N}_{3}(t+1)+\sum_{k=1}^{t}\mathbb{N}_{3}(k)2^{t-k}\right)m$ & $d\geq  \begin{cases}
                8;\text{ if }3\nmid m,\text{ }3\nmid t \text{ and }3\nmid (t+1), \\
                7;\text{ if }3\mid m,\text{ }3\nmid t\text{ and }3\nmid(t+1), \\
                5;\text{ if }3\nmid m,\text{ }3\mid t\text{ and }3\nmid (t+1), \\
                4;\text{ if }3\mid m,\text{ }3\mid t\text{ and }3\nmid (t+1), \\
                4;\text{ if }3\nmid t\text{ and }3\mid(t+1), \\
                3;\text{ otherwise.}
            \end{cases}$ & Theorem $\ref{Th13}$ \\
            \hline
            $x^{q^h-1}$  & $h$ is an integer such that $1\leq h\leq\lceil\frac{m}{2}\rceil$ and $q=p$, where $p$ is an odd prime & $\begin{cases}
            q^m-1-\mathbb{N}_{p}(m)-(p-1)m;\text{ for }h=1, \\
            q^m-1-\mathbb{N}_{p}(m)-\Biggl((p-1)\times\mathbb{N}_{p}(h)+ \\ \sum_{k=1}^{h-1}(p-1)(p^{h-k}-p^{h-k-1})\times\mathbb{N}_{p}(k)\Biggr) m;\text{ for }h\geq 2 
        \end{cases} $ &  $d\geq \begin{cases}
            p^{h}+2;\text{ if }p\nmid m\text{ and }h<p, \\
            p^{h}+1;\text{ if }p\mid m\text{ and }h<p, \\
            p^{h-1};\text{ if }h=p, \\
          %  ;\text{ if }p\mid m\text{ and }h=p, \\
           p^{p-1}+1 ;\text{ if }h>p\text{ and }p\nmid h, \\
            p^{p-1};\text{ if }h>p\text{ and }p\mid h,
        \end{cases}$  & Theorem $\ref{Th8}$ \\
            \hline
             $x^{2\cdot q^{h}-1}$ & $h=\frac{m}{2}$, $m$ is even, and $q=p$, where $p$ is an odd prime & $\begin{cases}
            q^2-1-4p+6; \text{ for }h=1,  \\
          q^m-1-\mathbb{N}_{p}(m)-\Biggl(\mathbb{N}_{p}(hp-h+1) + \mathbb{N}_{p}(h)\times(p-2)+ \\ \sum_{t=1}^{h-1}(p-1)p^{t-1}\Biggl(\mathbb{N}_{p}(h+1-t)+ 
           \mathbb{N}_{p}(h-t)\times(p-2)\Biggr) \\ +\left((p-1)p^{h-1}-1\right)\Biggr)\cdot m -\frac{m}{2};\text{ for }h\geq 2,
       \end{cases}$ & $d \geq \begin{cases}
            2p^{h}+2;\text{ if }p\nmid m\text{ and }1<h<p, \\
            2p^{h}+1;\text{ if }p\mid m\text{ and }1<h<p, \\
            p^{h-1}+1;\text{ if }h=p, \\
  %          p;\text{ if }h=1, \\
            p^{p-1}+1;\text{ if }h>p\text{ and }p\nmid(h-1), \\
            p^{p-1};\text{ if }h>p\text{ and }p\mid(h-1). \\
        \end{cases}$ & Theorem $\ref{Th9}$ \\
            \hline
            $x^{q^{2h}-q^{h}+1}$ & $h=\frac{m-1}{2}$, $m$ is odd, and $q=p$, where $p$ is an odd prime   & $q^m-1-\mathcal{L}_{s}$, $\mathcal{L}_{s}$ given in Theorem $\ref{Th6}$  & $d \geq \begin{cases}
            p^{h}+p^{h-1};\text{ if }p\nmid m \text{ and }1<h<p, \\
            p^{h}+p^{h-1}-1;\text{ if }p\mid m \text{ and }1<h<p, \\
            p^{h}+2;\text{ if }h=p, \\
            p^{p}+1;\text{ if }h>p \text{ and }p\nmid(h-1), \\
            p^{p};\text{ if }h>p \text{ and }p\mid(h-1).
        \end{cases}$ & Theorem $\ref{Th11}$  \\
            \hline
        \end{tabular}
    \end{adjustbox}
    \caption{Known $q$-ary cyclic codes $\mathcal{C}_{F}$ through power functions $F(x)$ over $\mathbb{F}_{q^{m}}$ with parameters $[q^{m}-1,\dim(\mathcal{C}_{F}), d]$, where $q$ is odd.}
    \label{Table3}
\end{table}

\section{Construction of $q$-ary cyclic codes $\mathcal{C}_{F}$ defined in Eq. $(\ref{Eq3})$}\label{sec3}
%In this section, we give construction of $q$-ary cyclic codes $\mathcal{C}_{F}$ 
\subsection{Cyclic codes from $x^{2\cdot 3^{\frac{m-1}{2}}+1}$ over $\mathbb{F}_{3^m}$, where $m$ is odd}
In this subsection, we construct ternary cyclic codes $\mathcal{C}_{F_{1}}$ by employing the monomial $F_{1}(x)=x^{2\cdot 3^{\frac{m-1}{2}}+1}$ over $\mathbb{F}_{3^m}$, where $m$ is odd. The differential uniformity $\Delta_{F_{1}}$ of the monomial $F_{1}$ over $\mathbb{F}_{3^m}$ is $4$ \cite{Dobbertin}. First, we need the following lemma.
\vspace{1em}
\begin{lemma}\label{L8}
    Let $m\geq 3$ be odd. Let $s^{\infty}$ be the sequence over $\mathbb{F}_{3}$ defined in $(\ref{Eq3})$ through the monomial $F_{1}(x)=x^{2\cdot 3^{\frac{m-1}{2}}+1}$ over $\mathbb{F}_{3^m}$. Then the linear span $\mathcal{L}_{s}$ of $s^{\infty}$ is equal to $3m+\mathbb{N}_{3}(m)$ and the minimal polynomial $\mathfrak{M}_{s}(x)$ of $s^{\infty}$ is given by 
\begin{flalign*}
  \hspace{30pt}  \mathfrak{M}_{s}(x) &= (x-1)^{\mathbb{N}_{3}(m)}m_{\alpha^{-2}}(x) m_{\alpha^{-(1+3^{\frac{m-1}{2}})}}(x) m_{\alpha^{-(1+2\cdot3^{\frac{m-1}{2}})}}(x). &
\end{flalign*}
    \begin{proof}
        By definition, we have 
        \begin{flalign}\label{Eq5}
      \hspace{30pt}      s_{t} &=\operatorname{Tr}_{3}^{3^m}\left((\alpha^t+1)^{2\cdot 3^{\frac{m-1}{2}}+1}\right) & \nonumber \\
                               &= \operatorname{Tr}_{3}^{3^m}\left((\alpha^t+1)^{2+3^{\frac{m+1}{2}}}\right) \nonumber \\
                               & = \operatorname{Tr}_{3}^{3^m}\left((\alpha^{2t}+2\alpha^t+1)(\alpha^{t\cdot{3^{\frac{m+1}{2}}}}+1)\right) \nonumber \\
                               &=  \operatorname{Tr}_{3}^{3^m}\left((\alpha^{t})^{2+3^{\frac{m+1}{2}}}+2(\alpha^t)^{1+3^{\frac{m+1}{2}}}+(\alpha^{t})^2+1\right) \nonumber \\
                               &= \sum_{i=0}^{m-1}(\alpha^{t})^{(2+3^{\frac{m+1}{2}})3^i}+2\sum_{i=0}^{m-1}(\alpha^{t})^{(1+3^{\frac{m+1}{2}})3^i}+\sum_{i=0}^{m-1}(\alpha^{t})^{2\cdot 3^{i}}+m \text{ (mod $3$)};\text{ for all $t\geq 0$.}
        \end{flalign} 
        For $m$ being odd, we can say $1+3^{\frac{m+1}{2}}\equiv 1+3^{\frac{m-1}{2}}\pmod{v}$ and $2+3^{\frac{m+1}{2}}\equiv 1+2\cdot3^{\frac{m-1}{2}}\pmod{v}$. From Lemma \ref{L7}, it is clear that $2$, $1+3^{\frac{m-1}{2}}$, and $1+2\cdot3^{\frac{m-1}{2}}$ are the coset leaders of the $3$-cyclotomic cosets $C_{2}$, $C_{1+3^{\frac{m+1}{2}}}$, and $C_{2+3^{\frac{m+1}{2}}}$ respectively. Hence, $C_{2}$, $C_{1+3^{\frac{m-1}{2}}}$, and $C_{1+2\cdot3^{\frac{m-1}{2}}}$ are pairwise disjoint. By Lemma \ref{L3}, we conclude that they are of size $m$.
        \vskip 1pt
        The desired conclusions on the linear span $\mathcal{L}_{s}$ and the minimal polynomial $\mathfrak{M}_{s}(x)$ of $s^{\infty}$ follow from Lemma $\ref{L6}$ and Eq. $(\ref{Eq5})$.
    \end{proof}
\end{lemma}

\begin{theorem}\label{Th12}
    Let $m\geq 3$ be odd. The ternary code $\mathcal{C}_{F_{1}}$ defined by the sequence $s^{\infty}$ in $(\ref{Eq3})$ through the monomial $F_{1}(x)=x^{2\cdot 3^{\frac{m-1}{2}}+1}$ over $\mathbb{F}_{3^m}$ has the generator polynomial $\mathfrak{M}_{s}(x)$ of Lemma $\ref{L8}$ and has parameters $[3^m-1,3^m-3m-1-\mathbb{N}_{3}(m),d]$, where $3\leq d\leq 8$.
    \begin{proof}
        The dimension of the code $\mathcal{C}_{F_{1}}$ directly follows from Lemma $\ref{L8}$. We now determine the bounds on the minimum weight of the code $\mathcal{C}_{F_{1}}$.
        \par The upper bound on $d$ follows from the Sphere-packing bound. It is clear from Lemma $\ref{L8}$ that the generator polynomial $\mathfrak{M}_{s}(x)$ of $\mathcal{C}_{F_{1}}$ has roots $\alpha^{-(1+3^{\frac{m+1}{2}})}$ and $\alpha^{-(2+3^{\frac{m+1}{2}})}$. The lower bound then follows from the BCH bound. 
        
    \end{proof}
\end{theorem}

\begin{example}
    Let $m=3$ and $\alpha$ be a root of the primitive polynomial $x^3+2x+1$ over $\mathbb{F}_{3}$. The generator polynomial of the ternary cyclic code $\mathcal{C}_{F_{1}}$ is $\mathfrak{M}_{s}(x)=x^9+x^8+x^6+1$. Then $\mathcal{C}_{F_{1}}$ is a ternary $[26,17,4]$ cyclic code. Its dual code $\mathcal{C}_{F_{1}}^{\perp}$ is a $[26, 9, 9]$ ternary cyclic code.
\end{example}
\vspace{1em}
\begin{example}
    Let $m=5$ and $\alpha$ be a root of the primitive polynomial $x^5+2x+1$ over $\mathbb{F}_{3}$. The generator polynomial of the ternary cyclic code $\mathcal{C}_{F_{1}}$ is $\mathfrak{M}_{s}(x)=x^{16} + 2x^{15} +  2x^{14} + 2x^{12} + 2x^9 + x^8 + 2x^5 + 2x^2 + 2x + 2$. Then $\mathcal{C}_{F_{1}}$ is a ternary $[242, 226, 5]$ cyclic code. Its dual code $\mathcal{C}_{F_{1}}^{\perp}$ is a $[242,16,126]$ ternary cyclic code. The code $\mathcal{C}_{F_{1}}$ is almost optimal according to the Database \textnormal{\cite{Database}}. 
\end{example}

\subsection{Cyclic codes from $x^{p^{m/2}+2}$ over $\mathbb{F}_{p^m}$, where $m$ is even}
In this subsection, we study the $p$-ary cyclic codes $\mathcal{C}_{F_{2}}$ by utilizing the monomial $F_{2}(x)=x^{p^{m/2}+2}$ over $\mathbb{F}_{p^m}$, where $p$ is an odd prime and $m$ is even. The differential uniformity $\Delta_{F_2}$ of the power function $F_{2}$ over $\mathbb{F}_{p^m}$ is $4$ \cite{Uniform1,Uniform2}. First, we need to prove the following lemma.
\vspace{1em}
\begin{lemma}\label{L13}
    Let $p\geq 3$ be a prime and $m\geq 2$ be an even integer. Let $s^{\infty}$ be the sequence over $\mathbb{F}_{p}$ defined in $(\ref{Eq3})$ through the monomial $F_{2}(x)=x^{p^{m/2}+2}$ over $\mathbb{F}_{p^m}$. Then the linear span $\mathcal{L}_{s}$ and the minimal polynomial $\mathfrak{M}_{s}(x)$ of $s^{\infty}$ are given by  

    \begin{align*}
        \mathcal{L}_{s} &= \begin{cases}
            \mathbb{N}_{p}(m) + \frac{5m}{2};\text{ if }p=3, \\
            \mathbb{N}_{p}(m) + \frac{7m}{2};\text{ if }p\geq 5, 
        \end{cases}
    \end{align*}
    and
    \begin{align*}
        \mathfrak{M}_{s}(x) &= \begin{cases}
            (x-1)^{\mathbb{N}_{p}(m)}m_{\alpha^{-2}}(x)m_{\alpha^{-(p^{m/2}+1)}}(x)m_{\alpha^{-(p^{m/2}+2)}}(x);\text{ if }p=3, \\
            (x-1)^{\mathbb{N}_{p}(m)}m_{\alpha^{-1}}(x)m_{\alpha^{-2}}(x)m_{\alpha^{-(p^{m/2}+1)}}(x)m_{\alpha^{-(p^{m/2}+2)}}(x);\text{ if }p\geq 5.
        \end{cases}       
    \end{align*}
    \begin{proof}
        By definition, we have
        \begin{flalign}\label{Eq15}
        
            s_{t} &= \operatorname{Tr}_{p}^{p^m}\left((\alpha^{t}+1)^{p^{m/2}+2}\right) \nonumber \\
            &= \operatorname{Tr}_{p}^{p^m}\left[\left((\alpha^{t})^{2}+2\alpha^t+1\right)\left((\alpha^{t})^{p^{m/2}}+1\right)\right] \nonumber \\
            &= \operatorname{Tr}_{p}^{p^m}\left[(\alpha^{t})^{p^{m/2}+2}+2(\alpha^{t})^{p^{m/2}+1}+(\alpha^{t})^{p^{m/2}}+(\alpha^{t})^{2}+2\alpha^{t}+1\right] \nonumber \\
            &= \operatorname{Tr}_{p}^{p^m}(1)+ 3\operatorname{Tr}_{p}^{p^m}\left(\alpha^{t}\right)+\operatorname{Tr}_{p}^{p^m}\left((\alpha^{t})^{2}\right)+2\operatorname{Tr}_{p}^{p^m}\left((\alpha^{t})^{p^{m/2}+1}\right)+ \operatorname{Tr}_{p}^{p^m}\left((\alpha^{t})^{p^{m/2}+2}\right)  \nonumber \\
            &=  m\text{ (mod $p$)} + 3\sum_{i=0}^{m-1}\left(\alpha^{t}\right)^{p^i}+\sum_{i=0}^{m-1}\left(\alpha^{t}\right)^{2\cdot p^i} + 2\sum_{i=0}^{m-1}\left(\alpha^{t}\right)^{(p^{m/2}+1)\cdot p^{i}}+\sum_{i=0}^{m-1}\left(\alpha^{t}\right)^{(p^{m/2}+2)\cdot p^{i}};\text{ for all }t\geq 0. 
        \end{flalign}
        By using Lemma $\ref{L7}$ and $\ref{L3}$, it is easy to check that $C_{1}\cap C_{2}=\emptyset$ and $|C_{1}|=|C_{2}|=m$. Since $(p^{m/2}+1)\cdot p^{m/2}\equiv p^{m/2}+1 \pmod{v}$, we have $|C_{p^{m/2}+1}|=\frac{m}{2}$. It is clear that $C_{1}\cap C_{p^{m/2}+1}=\emptyset$ and $C_{2}\cap C_{p^{m/2}+1}=\emptyset$. Note that $p^{m/2}+2$ and $2\cdot p^{m/2}+1$ are the only integers in the $p$-cyclotomic coset $C_{p^{m/2}+2}$ that are not divisible by $p$, which implies $C_{p^{m/2}+2}$ must be of size $m$. Hence, we conclude $C_{p^{m/2}+1} \cap C_{p^{m/2}+2}=\emptyset$. Recall that $C_{j_{1}}=C_{j_{2}}$ would imply $\operatorname{wt}_{p}(j_{1})=\operatorname{wt}_{p}(j_{2})$. Since $\operatorname{wt}_{p}(1)=\operatorname{wt}_{p}(2)=1$, whereas $\operatorname{wt}_{p}(p^{m/2}+2)=2$; therefore $C_{1}\cap C_{p^{m/2}+2}=\emptyset$ and $C_{2}\cap C_{p^{m/2}+2}=\emptyset$.
        \vskip 1pt
        The desired conclusions on the linear span $\mathcal{L}_{s}$ and the minimal polynomial $\mathfrak{M}_{s}(x)$ of $s^{\infty}$ follow from Lemma $\ref{L6}$ and Eq. $(\ref{Eq15})$.
    \end{proof}
\end{lemma}

\begin{theorem}\label{Th7}
    Let $m\geq 2$ be an even integer, and the code $\mathcal{C}_{F_{2}}$ be defined by the sequence $s^{\infty}$ in $(\ref{Eq3})$ through the monomial $F_{2}(x)$ over $\mathbb{F}_{p^m}$ of Lemma $\ref{L13}$. Then $\mathcal{C}_{F_{2}}$ has parameters $[p^{m}-1,p^{m}-1-\mathcal{L}_{s},d]$ over $\mathbb{F}_{p}$ with the generator polynomial $\mathfrak{M}_{s}(x)$, where $\mathcal{L}_{s}$ and $\mathfrak{M}_{s}(x)$ are given in Lemma $\ref{L13}$. In addition, 
    \begin{flalign*}
   \hspace{15pt}    & \begin{cases}
           3\leq d \leq 6; \text{ if }p=3, \\  
           5\leq d\leq 8; \text{ if }p\geq 5 \text{ and }p\nmid m, \\
           4\leq d\leq 8; \text{ if }p\geq 5 \text{ and }p\mid m.
        \end{cases}&
    \end{flalign*}
    \begin{proof}
        The dimension of the code $\mathcal{C}_{F_{2}}$ directly follows from Lemma $\ref{L13}$. We now determine the bounds on the minimum weight of the code $\mathcal{C}_{F_{2}}$.
        \par When $p=3$, it is clear from Lemma $\ref{L13}$ that the generator polynomial $\mathfrak{M}_{s}(x)$ of $\mathcal{C}_{F_{2}}$ has roots $\alpha^{-(p^{m/2}+1)}$ and $\alpha^{-(p^{m/2}+2)}$. The lower and upper bounds on the minimum distance $d$ then follow from the BCH bound and the Sphere-packing bound, respectively.
        \par When $p\geq 5$ and $p\nmid m$, note that the reciprocal of the generator polynomial $\mathfrak{M}_{s}(x)$ defined in Lemma $\ref{L13}$ has roots $\alpha^{j}$ for all $j\in A+B$, where $A=\{0,1,2\}$ and $B=\{0,p^{m/2}\}$. As we know, the code with generator polynomial $\mathfrak{M}_{s}(x)$ and the code generated by the reciprocal of $\mathfrak{M}_{s}(x)$ have identical weight distributions. Since $\operatorname{gcd}(p^{m/2},v)=1<4$ and $A+B$ is a subset of the defining set of the reciprocal of $\mathfrak{M}_{s}(x)$, the Hartmann-Tzeng bound gives $d\geq 5$. The upper bound on the minimum distance $d$ then follows from the Sphere-packing bound. In the case of $p\mid m$, one can prove similarly.
    \end{proof}
\end{theorem}

\begin{corollary}\label{cor1}
    Let $p=5$ and $m$ be an even positive integer such that $p\nmid m$. Then the quinary cyclic code $\mathcal{C}_{F_{2}}$ defined in Theorem $\ref{Th7}$ has parameters $[5^m-1,5^{m}-2-\frac{7m}{2},5]$.
    \begin{proof}
        For $p=5$ and $p\nmid m$, $\mathbb{N}_{p}(m)=1$. Thus, the dimension of the code $\mathcal{C}_{F_{2}}$ follows from the linear span $\mathcal{L}_{s}$ of the sequence $s^{\infty}$ in Lemma $\ref{L13}$. From Theorem $\ref{Th7}$, the minimum distance $d$ of the code $\mathcal{C}_{F_{2}}$ satisfies $d\geq 5$. To show that $d=5$, we construct a codeword of Hamming weight $5$ in $\mathcal{C}_{F_{2}}$.
        \par Consider a codeword of the form $c(x)=1+x^{t_{1}}+x^{t_{2}}+x^{t_{3}}+x^{t_{4}}$ in $\mathbb{F}_{5}[x]$. It is clear that $c(1)=0$. For $c(x)$ to be valid, it must satisfy $c(\alpha^{-s})=0$ for all $s\in\{1,2,5^{m/2}+1,5^{m/2}+2\}$. Let $\beta_{i}=\alpha^{-t_{i}}$ for $i=1,2,3,4$. Then, it is equivalent to finding four distinct non-zero elements $\beta_{i}\in\mathbb{F}_{5^m}$ such that 
        \begin{equation}\label{Eq29}
            \sum_{i=1}^{4}\beta_{i}^{s}=-1\text{ for all $s\in\{1,2,5^{m/2}+1,5^{m/2}+2\}$.}
        \end{equation}
        Note that the existence of such distinct non-zero elements $\beta_{1},\beta_{2},\beta_{3}$ and $\beta_{4}$ in $\mathbb{F}_{5^m}$ will guaranty the existence of the non-zero distinct integers $t_{1},t_{2},t_{3}$ and $t_{4}$, where $1\leq t_{1},t_{2},t_{3},t_{4}\leq 5^m-2$. Suppose we choose $\beta_{i}\in\mathbb{F}_{5^{m/2}}^{*}$ for $i=1,2,3,4$. Then for any $\beta_{i}\in\mathbb{F}_{5^{m/2}}$, we have $\beta_{i}^{5^{m/2}}=\beta_{i}$, $\beta_{i}^{5^{m/2}+1}=\beta_{i}^{2}$, and $\beta_{i}^{5^{m/2}+2}=\beta_{i}^{3}$. Thus, Eq. $(\ref{Eq29})$ reduces to the following system  
        \begin{equation}\label{Eq30}
            \begin{cases}
                1+\sum_{i=1}^{4}\beta_{i}=0, \\
                1+\sum_{i=1}^{4}\beta_{i}^{2}=0, \\
                1+\sum_{i=1}^{4}\beta_{i}^{3}=0.
            \end{cases}
        \end{equation}
        Consequently, it suffices to find four distinct non-zero elements $\beta_i \in \mathbb{F}_{5^{m/2}}$ satisfying the system $(\ref{Eq30})$. The existence of such a set follows from the fact that the primitive BCH code of length $5^{m/2}-1$ over $\mathbb{F}_{5}$ with designed distance $5$ has the minimum distance equal to $5$. 
       \par We now show that the primitive BCH code $\mathcal{C}_{(5,5^{m/2}-1,5,0)}$ has the minimum Hamming distance equal to $5$. From the BCH bound, the minimum distance of the code $\mathcal{C}_{(5,5^{m/2}-1,5,0)}$ is at least $5$. Consider a polynomial $p(x)=1+x^{r_{1}}+x^{r_{2}}+x^{r_{3}}+x^{r_{4}}\in\mathbb{F}_{5}[x]$, where $1\leq r_{1},r_{2},r_{3},r_{4}\leq 5^{m/2}-2$. Let $\gamma$ be a primitive element of $\mathbb{F}_{5^{m/2}}$. Now, for any given $r_{1}$, we can choose $r_{2}$, $r_{3}$, and $r_{4}$ such that $\gamma^{r_{2}}=2\gamma^{r_{1}}+4$, $\gamma^{r_{3}}=3\gamma^{r_{1}}+3$, and $\gamma^{r_{4}}=4\gamma^{r_{1}}+2$. Note that for any $1\leq i\neq j\leq 4$, $r_{i}=r_{j}$ would imply $\gamma^{r_{i}}=\gamma^{r_{j}}$, which gives $\gamma^{r_{1}}=1$ due of the fact $\mathbb{F}_{5^{m/2}}$ is of characteristic $5$. But since $\gamma$ is a generator of $\mathbb{F}_{5^{m/2}}^{*}$, there is no $r_{1}$ such that $\gamma^{r_{1}}=1$. Thus, $r_{1},r_{2},r_{3}$, and $r_{4}$ must be pairwise distinct. Furthermore, one can easily verify that $p(1)=p(\gamma)=p(\gamma^2)=p(\gamma^3)=0$, which implies that $p(x)$ is a codeword of weight $5$ in $\mathcal{C}_{(5,5^{m/2}-1,5,0)}$. Therefore, $\beta_{i}=\gamma^{r_{i}}$, $i=1,2,3,4$ satisfies the system $(\ref{Eq30})$. This completes the proof.
        %the generator polynomial $m_{\gamma^{0}}(x)m_{\gamma^{1}}(x)m_{\gamma^{2}}(x)m_{\gamma^{3}}(x)$ 
        % where $1\leq t_{1},t_{2},t_{3},t_{4}\leq 5^m-2$. Let $\beta=\alpha^{-1}$. Now, for any given $t_{1}$, we can choose $t_{2}$, $t_{3}$, and $t_{4}$ such that $\beta^{t_{2}}=2\beta^{t_{1}}+4$, $\beta^{t_{3}}=3\beta^{t_{1}}+3$, and $\beta^{t_{4}}=4\beta^{t_{1}}+2$. Note that $\beta^{t_{i}}=\beta^{t_{j}}$ for any $1\leq i\neq j\leq 4$ would imply $\beta^{t_{1}}=1$ because $\mathbb{F}_{5^m}$ is of characteristic $5$. Since $\beta$ is also a generator of $\mathbb{F}_{5^m}^{*}$, there is no such $t_{1}$ satisfying $\beta^{t_{1}}=1$.
    \end{proof}
\end{corollary}

\begin{corollary}\label{cor2} 
    Let $p\geq 5$ be a prime and $m$ be an even positive integer such that $p\mid m$. Then the $p$-ary cyclic code $\mathcal{C}_{F_{2}}$ defined in Theorem $\ref{Th7}$ has parameters $[p^{m}-1,p^m-1-\frac{7m}{2},4]$. 
    \begin{proof}
        For $p\geq5$ with $p\mid m$, we have $\mathbb{N}_{p}(m)=0$. Hence, the dimension of the code $\mathcal{C}_{F_{2}}$ directly follows from the linear span $\mathcal{L}_{s}$ of the sequence $s^{\infty}$ in Lemma $\ref{L13}$. From Theorem $\ref{Th7}$, the minimum distance $d$ of the code $\mathcal{C}_{F_{2}}$ satisfies $d\geq 4$. In order to show that $d=4$, we construct a codeword of Hamming weight $4$ in $\mathcal{C}_{F_{2}}$. 
        \par Consider a codeword of the form $c(x)=c_{1}x^{t_{1}}+c_{2}x^{t_{2}}+c_{3}x^{t_{3}}+c_{4}x^{t_{4}}$ with the coefficients $c_{1},c_{2},c_{3},c_{4}\in\mathbb{F}_{p}^{*}$. For $c(x)$ to be valid, it must satisfy $c(\alpha^{-s})=0$ for all $s\in\{1,2,p^{m/2}+1,p^{m/2}+2\}$. Let $\beta_{i}=\alpha^{-t_{i}}$ for $i=1,2,3,4$. Then, it is equivalent to finding four distinct non-zero elements $\beta_{i}\in\mathbb{F}_{p^m}$ and their corresponding coefficients $c_{i}\in\mathbb{F}_{p}^{*}$ such that
        \begin{equation}\label{Eq27}
            \sum_{i=1}^{4}c_{i}\beta_{i}^{s}=0\text{ for all $s\in\{1,2,p^{m/2}+1,p^{m/2}+2\}$.}
        \end{equation}
        Suppose we choose $\beta_{i}\in\mathbb{F}_{p^{m/2}}^{*}$ for all $i=1,2,3,4$. Then $\beta_{i}^{p^{m/2}}=\beta_{i}$. Thus, Eq. $(\ref{Eq27})$ reduces to the following system  
        \begin{equation}\label{Eq28}
            \begin{cases}
                \sum_{i=1}^{4}c_{i}\beta_{i}=0, \\
                \sum_{i=1}^{4}c_{i}\beta_{i}^{2}=0, \\
                \sum_{i=1}^{4}c_{i}\beta_{i}^{3}=0.
            \end{cases}
        \end{equation}
        Consequently, it suffices to find four distinct elements $\beta_i \in \mathbb{F}_{p^{m/2}}^*$ and nonzero coefficients $c_i \in \mathbb{F}_p^*$ satisfying the system $(\ref{Eq28})$. The existence of such a set follows from the fact that the narrow-sense primitive BCH code of length $p^{m/2}-1$ over $\mathbb{F}_p$ with designed distance $4$ has the minimum distance $d=4$.% The existence of such a set follows from classical narrow-sense primitive BCH codes of length $p^{m/2}-1$ over $\mathbb{F}_{p}$ with designed distance $4$.
        \par We now show that the narrow-sense primitive BCH code $\mathcal{C}_{(p,p^{m/2}-1,4,1)}$ has the minimum Hamming distance equal to $4$. From the BCH bound, the minimum distance of the code $\mathcal{C}_{(p,p^{m/2}-1,4,1)}$ is at least $4$. Consider a polynomial 
        \begin{equation*}
            p(x)=1-\sum_{i=1}^{3}\left(\prod_{\substack{j=1,j\neq i}}^{3}\frac{1-\gamma^{a_{j}}}{(\gamma^{a_{i}}-\gamma^{a_{j}})}\right)\frac{x^{a_{i}}}{\gamma^{a_{i}}}, 
        \end{equation*}
        where $a_{i}=i\frac{p^{m/2}-1}{p-1}$ for $i\in\{1,2,3\}$ and $\gamma$ is a primitive element of $\mathbb{F}_{p^{m/2}}$. Clearly, $\gamma^{a_{i}}\in\mathbb{F}_{p}^{*}$, and since the order of $\gamma$ is $p^{m/2}-1$, $\gamma^{a_{i}}\neq 1$ and $\gamma^{a_{i}}\neq\gamma^{a_{j}}$ for any $1\leq i\neq j\leq 3$. Therefore, $p(x)$ is well-defined and also belongs to $\mathbb{F}_{p}[x]$. One can verify that $p(\gamma)=p(\gamma^2)=p(\gamma^{3})=0$. Hence, $p(x)$ is a codeword of weight $4$. With the help of $p(x)$ above, the codeword $c(x)$ in $\mathcal{C}_{F_{2}}$ can be easily constructed. This completes the proof.
    \end{proof}
\end{corollary}
\vspace{1em}
\begin{example}
     Let $m=4$ and $p=3$. Let $\alpha$ be a root of the primitive polynomial $x^4 + 2x^3 + 2$ over $\mathbb{F}_{3}$. The generator polynomial of the cyclic code $\mathcal{C}_{F_{2}}$ is $\mathfrak{M}_{s}(x)=x^{11} + 2x^{10} +  2x^8 + x^2 + x + 2$. Then $\mathcal{C}_{F_{2}}$ is a $[80, 69, 4]$ ternary cyclic code. Its dual code  $\mathcal{C}_{F_{2}}^{\perp}$ is a $[80, 11, 35]$ ternary cyclic code.
\end{example}
\vspace{1em}
\begin{example}
    Let $m=2$ and $p=5$. Let $\alpha$ be a root of the primitive polynomial $x^2 + 4x + 2$ over $\mathbb{F}_{5}$. The generator polynomial of the cyclic code $\mathcal{C}_{F_{2}}$ is $\mathfrak{M}_{s}(x)=x^8 + x^7 + 2x^4 + 2x^3 + 3x^2 + 4x + 2$. Then $\mathcal{C}_{F_{2}}$ is a quinary $[24,16,5]$ cyclic code. Its dual code  $\mathcal{C}_{F_{2}}^{\perp}$ is a $[24,8,13]$ quinary cyclic code. According to the Database \textnormal{\cite{Database}}, the code $\mathcal{C}_{F_{2}}$ is almost optimal, and the dual code $\mathcal{C}_{F_{2}}^{\perp}$ is optimal.
\end{example}
\vspace{1em}
\begin{example}
    Let $m=2$ and $p=7$. Let $\alpha$ be a root of the primitive polynomial $x^2 + 6x + 3$ over $\mathbb{F}_{7}$. The generator polynomial of the cyclic code $\mathcal{C}_{F_{2}}$ is $\mathfrak{M}_{s}(x)=x^8 + 5x^7 + 4x^6 + 3x^5 + 6x^4 + 5x^3 + 6x + 5$. Then $\mathcal{C}_{F_{2}}$ is a $[48, 40, 5]$ cyclic code over $\mathbb{F}_{7}$. Its dual code  $\mathcal{C}_{F_{2}}^{\perp}$ is a $[48, 8, 33]$ cyclic code over $\mathbb{F}_{7}$. According to the Database \textnormal{\cite{Database}}, the code $\mathcal{C}_{F_{2}}$ is almost optimal, and the dual code $\mathcal{C}_{F_{2}}^{\perp}$ is optimal.
\end{example}
\vspace{1em}
\begin{example}
    Let $m=4$ and $p=5$. Let $\alpha$ be a root of the primitive polynomial $x^4 + 4x^2 + 4x + 2$ over $\mathbb{F}_{5}$. The generator polynomial of the cyclic code $\mathcal{C}_{F_{2}}$ is $\mathfrak{M}_{s}(x)=x^{15} + 4x^{14} + x^{12} + 3x^9 + 4x^7 + 3x^6 + 4x^5 + 2x^4 + 4x^2 + x + 3$. Then $\mathcal{C}_{F_{2}}$ is a quinary $[624, 609, 5]$ cyclic code. Its dual code  $\mathcal{C}_{F_{2}}^{\perp}$ is a $[624,15,399]$ quinary cyclic code. Due to the large code length, it is not possible to verify the optimality of $\mathcal{C}_{F_{2}}$ and $\mathcal{C}_{F_{2}}^{\perp}$ from the Database \textnormal{\cite{Database}}.  
\end{example}
\vspace{1em}
\begin{remark}
    When $p\geq 5$ and $p\nmid m$, Theorem $\ref{Th7}$ shows that the difference between the lower and upper bounds on the minimum distance $d$ is relatively small. This suggests a higher possibility that, in this scenario, the parameters of $\mathcal{C}_{F_{2}}$ can be optimal or near-optimal.   
\end{remark}

\subsection{Cyclic codes from $x^{\frac{q^{(m+1)/2}-1}{q-1}}$ over $\mathbb{F}_{q^m}$, where $m$ is odd}
In this subsection, we construct $q$-ary cyclic codes $\mathcal{C}_{F_{3}}$ defined by the sequence $s^{\infty}$ of $(\ref{Eq3})$ through the monomial $F_{3}(x)=x^{\frac{q^{(m+1)/2}-1}{q-1}}$ over $\mathbb{F}_{q^m}$, where $m$ is odd.
\par First, we shall generalize some techniques demonstrated in \cite{P2} and present some Lemmas, which will be utilized to determine the generator polynomial of the code $\mathcal{C}_{F_{3}}$.
\par
Let $t$ be a positive integer and $q=p^e$. We first consider a set $D_{t}(q)$ depending upon $t$ and $q$ as follows
\begin{equation}\label{Eq1}
    D_{t}(q)=\{\sum_{i=0}^{t-1}a_{i}q^{i}:a_{i}\text{'s, not all zero, is either $0$ or $1$ for all $i=0,1,2,\cdots,t-1$}\}.
\end{equation}
Clearly, $0\notin D_{t}(q)$ and any integer in $D_{t}(q)$ must lie between $1$ and $\frac{q^t-1}{q-1}$. Let $j\in D_{t}(q)$ be such that $q\nmid j$. we define the following notations:
\begin{align}
    \epsilon_{j}^{(t)}& =\begin{cases}
        1,\text{ if }j=\frac{q^t-1}{q-1} \\
        \Big\lceil\operatorname{log}_{q}\left(\frac{q^t-1}{j(q-1)}\right)\Big\rceil,\text{ if }1\leq j<\frac{q^t-1}{q-1}
    \end{cases}
\end{align}
and 
\begin{align}
    \kappa_{j}^{(t)} &=\epsilon_{j}^{(t)}\pmod{p}.
\end{align}
Let $B_{j}^{(t)}=\{q^{i}j:i=0,1,2,\cdots,\epsilon_{j}^{(t)}-1\}$. One can verify that 
\begin{flalign*}
   \hskip 15pt \bigcup_{\substack{j\in D_{t}(q),\text{ } q\nmid j}} B_{j}^{(t)}&=D_{t}(q)
 \text{ and } B_{i}^{(t)}\bigcap B_{j}^{(t)}=\emptyset & 
\end{flalign*}
for any pair of distinct $i$ and $j$, not divisible by $q$, in $D_{t}(q)$.
\vskip 1pt
%The following Lemma follows directly from the definition of $\epsilon_{j}^{(t)}$ and $B_{j}^{(t)}$.
\hspace{1em}
\begin{lemma}\label{L1}
    Let $t\in\mathbb{N}$ and $j\in D_{t+1}(q)$ be an integer not divisible by $q$, where $D_{t}(q)$ is defined as in $(\ref{Eq1})$. Then 
    \begin{enumerate}
        \item[$1.$] For $j\in D_{t}(q)$, $B_{j}^{(t+1)}=B_{j}^{(t)}\cup\{jq^{\epsilon_{j}^{(t)}}\}$ and $\epsilon_{j}^{(t+1)}=\epsilon_{j}^{(t)}+1$. 
        \item[$2.$] For $j\in D_{t+1}(q)\backslash D_{t}(q)$, $B_{j}^{(t)}=\{j\}$ and $\epsilon_{j}^{(t+1)}=1$. 
    \end{enumerate}
    \begin{proof}
      From the definition of $\epsilon_{j}^{(t+1)}$, we know that 
      \begin{align*}
          \epsilon_{j}^{(t+1)} &=\begin{cases}
              1, \text{ if }j=\frac{q^{t+1}-1}{q-1} \\
              \Big\lceil\operatorname{log}_{q}\left(\frac{q^{t+1}-1}{j(q-1)}\right)\Big\rceil,\text{ if }1\leq j<\frac{q^{t+1}-1}{q-1}
          \end{cases} 
      \end{align*}
      If $j\in D_{t}(q)$ and $q\nmid j$, then $\frac{q^{t+1}-1}{(q-1)j}=\frac{q(q^t-1)}{(q-1)j}+\frac{1}{j}$ gives $\frac{q(q^t-1)}{(q-1)j}<\frac{q^{t+1}-1}{(q-1)j}\leq \frac{q(q^t-1)}{(q-1)j}+1$, implying that $1+\operatorname{log}_{q}\left(\frac{q^t-1}{(q-1)j}\right)<\operatorname{log}_{q}\left(\frac{q^{t+1}-1}{(q-1)j}\right)<2+\operatorname{log}_{q}\left(\frac{q^t-1}{(q-1)j}\right)$. Since $\operatorname{log}_{q}\left(\frac{q^t-1}{(q-1)j}\right)$ can not be a nonzero integer, we conclude that $\epsilon_{j}^{(t+1)}=\epsilon_{j}^{(t)}+1$ for all $j\in D_{t}(q)$.
      \newline
      If $j\in D_{t+1}(q)\backslash D_{t}(q)$, $q\nmid j$, note that $1+q^t\leq j\leq\frac{q^{t+1}-1}{q-1}$. For $j=\frac{q^{t+1}-1}{q-1}$, we know that $\epsilon_{j}^{(t+1)}=1$ hold by the definition. For $j\neq \frac{q^{t+1}-1}{q-1}$, $j<\frac{q^{t+1}-1}{q-1}<qj$ would give $0<\operatorname{log}_{q}\left(\frac{q^{t+1}-1}{j(q-1)}\right)<1$. Therefore, $\epsilon_{j}^{(t+1)}=1$ for all $j\in D_{t+1}(q)\backslash D_{t}(q)$.
      \vskip 1pt
      Thus, the proof of the lemma follows from the the definition of $B_{j}^{(t)}$.
    \end{proof}
\end{lemma}

\begin{lemma}\label{L2}
    Let $t\in\mathbb{N}$ and $j\in D_{t}(q)$ be such that $q\nmid j$, where $D_{t}(q)$ is defined as in $(\ref{Eq1})$. Then
    \begin{align*}
        \epsilon_{j}^{(t)} &=\begin{cases}
            t,\text{ if }j=1, \\
            t-k,\text{ if }j\in D_{k+1}(q)\backslash D_{k}(q),\text{ where }k\in\{1,2,3,\cdots,t-1\}.
        \end{cases}
    \end{align*}
    \begin{proof}
        For $t=1$, we have $D_{t}(q)=\{1\}$. Observe that $\epsilon_{1}^{(1)}=1$ due to the definition. For all $t\geq 2$, since $q^{t-1}<\frac{q^t-1}{q-1}<q^t$, we obtain $\epsilon_{1}^{(t)}=\Big\lceil\operatorname{log}_{q}\left(\frac{q^t-1}{q-1}\right)\Big\rceil=t$.
        %\newline
       \par Take $t=2$ and $j\in D_{2}(q)$ with $q\nmid j$. Note that $D_{1}(q)=\{1\}$ and $D_{2}(q)=\{1,q,1+q\}$. Then, from Lemma \ref{L1}, we obtain
       \begin{align*}
           \epsilon_{j}^{(t)}=\epsilon_{j}^{(2)} &=\begin{cases}
               \epsilon_{j}^{(1)}+1=2;\text{ for }j=1, \\
               1;\text{ for }j\in D_{2}(q)\backslash D_{1}(q)\text{ and }q\nmid j.
           \end{cases}
       \end{align*}
      % $\epsilon_{1}^{(2)}=\epsilon_{1}^{(1)}+1=2$ and $\epsilon_{1+q}^{(2)}=1$. 
       \par Take $t=3$ and $j\in D_{3}(q)$ with $q\nmid j$. Note that $D_{2}(q)=\{1,q,1+q\}$ and $D_{3}(q)=\{1,q,q^2,1+q,1+q^2,q+q^2,1+q+q^2\}$. Then from Lemma \ref{L1}, we obtain
        \begin{align*}
           \epsilon_{j}^{(t)}=\epsilon_{j}^{(3)} &=\begin{cases}
               \epsilon_{j}^{(2)}+1;\text{ for }j\in D_{2}(q)\text{ and }q\nmid j, \\
               1;\text{ for }j\in D_{3}(q)\backslash D_{2}(q)\text{ and }q\nmid j
           \end{cases}
           =\begin{cases}
               3;\text{ for }j=1, \\
               2;\text{ for }j\in D_{2}(q)\backslash D_{1}(q)\text{ and }q\nmid j, \\
               1;\text{ for }j\in D_{3}(q)\backslash D_{2}(q)\text{ and }q\nmid j.
           \end{cases}
       \end{align*}
      % $\epsilon_{1}^{(3)}=\epsilon_{1}^{(2)}+1=3$, $\epsilon_{1+q}^{(3)}=\epsilon_{1+q}^{(2)}+1=2$ and $\epsilon_{1+q^2}^{(3)}=\epsilon_{1+q+q^2}^{(3)}=1$.
       \vskip 1pt
     \par  By continuing this reasoning for all values of $t$, we obtain the desired result.
    \end{proof}
\end{lemma}

\begin{lemma}\label{L5}
    Let $m$ and $t$ be positive integers and $D_{t}(q)$ be defined as in $(\ref{Eq1})$. Define $\Gamma_{(t)}=\{j\in D_{t}(q):q\nmid j\}$. Then, for any $x\in\mathbb{F}_{q^m}$, we have
    \[\begin{multlined}
        \operatorname{Tr}_{q}^{q^{m}}\left((x+1)^{\frac{q^t-1}{q-1}}\right)=\left( m+t\operatorname{Tr}_{q}^{q^m}(x)+\sum_{j\in\Gamma_{(2)}\backslash\Gamma_{(1)}}(t-1)\operatorname{Tr}_{q}^{q^m}(x^j)+\sum_{j\in\Gamma_{(3)}\backslash\Gamma_{(2)}}(t-2)\operatorname{Tr}_{q}^{q^m}(x^j)+\cdots\right.+\\ \left. \sum_{j\in\Gamma_{(t-1)}\backslash\Gamma_{(t-2)}}2\operatorname{Tr}_{q}^{q^m}(x^j)+\sum_{j\in\Gamma_{(t)}\backslash\Gamma_{(t-1)}}\operatorname{Tr}_{q}^{q^m}(x^j) \right) \pmod{p}
    \end{multlined}\]
        %\end{align*}
    \begin{proof}
   Note that $q=p^e$ and $1+p^{e}+p^{2e}+\cdots+p^{(t-1)e}$ is the $p$-adic representation of $\frac{q^t-1}{q-1}$. For a non-negative integer $j$ with $0\leq j\leq \frac{q^t-1}{q-1}$, let the $p$-adic representation of $j$ be defined as 
    \begin{equation*}
        j=j_{0}+j_{1}p+j_{2}p^{2}+\cdots+j_{(t-1)e}p^{(t-1)e}, \text{ where $0\leq j_{0},j_{1},j_{2},\cdots,j_{(t-1)e}\leq p-1$.}
    \end{equation*}
    From Lucas' Theorem \cite{P1}, one can derive that $\binom{\frac{q^t-1}{q-1}}{j}\equiv\binom{1}{j_{0}}\binom{0}{j_{1}}\cdots\binom{1}{j_{e}}\binom{0}{j_{e+1}}\cdots\binom{1}{j_{(t-1)e}}\not\equiv 0\pmod{p}$ if and only if $j_{s}\in\{0,1\}$ for all $e\mid s$ and $j_{s}=0$ for all $e\nmid s$, where $s\in\{0,1,2,\cdots,(t-1)e\}$. Hence, $\binom{\frac{q^t-1}{q-1}}{j}\equiv 1\pmod{p}$ for all $j\in D_{t}(q)\cup\{0\}$, and $p\large\mid\binom{\frac{q^t-1}{q-1}}{j}$, otherwise. Note that for each $j\in\Gamma_{(t)}$, there exists an unique set $B_{j}^{(t)}=\{j,jq,jq^2,\cdots,jq^{\epsilon_{j}^{(t)}-1}\}$. Since $\operatorname{char}(\mathbb{F}_{q^m})=p$, for any $x\in\mathbb{F}_{q^m}$, we can write
        \begin{flalign}\label{Eq2}
      \hskip 30pt  \operatorname{Tr}_{q}^{q^{m}}\left((x+1)^{\frac{q^t-1}{q-1}}\right)& =\operatorname{Tr}_{q}^{q^m}\left(1+\sum_{j\in D_{t}(q)}x^{j}\right) & \nonumber \\
        &=\operatorname{Tr}_{q}^{q^m}\left(1+\sum_{j\in \Gamma_{(t)}}\sum_{i\in B_{j}^{(t)}} x^{i}\right) & \nonumber \\
        &=\operatorname{Tr}_{q}^{q^m}\left(1\right)+\sum_{j\in \Gamma_{(t)}}\kappa_{j}^{(t)}\operatorname{Tr}_{q}^{q^m}\left(x^{j}\right) & \nonumber \\
        &=\operatorname{Tr}_{q}^{q^m}\left(1\right)+\operatorname{Tr}_{q}^{q^m}\left(\kappa_{1}^{(t)}x\right)+\sum_{k=1}^{t-1}\operatorname{Tr}_{q}^{q^m}\left(\sum_{j\in\Gamma_{(t-k+1)}\backslash\Gamma_{(t-k)}}\kappa_{j}^{(t)}x^{j}\right) 
    \end{flalign}
    According to Lemma \ref{L2}, $\epsilon_{1}^{(t)}=t$ and $\epsilon_{j}^{(t)}=t-(t-k)=k$ for all $j\in\Gamma_{(t-k+1)}\backslash\Gamma_{(t-k)}$. It is clear from the right-hand side of $(\ref{Eq2})$ that for $k=1,2,\cdots,t-1$ and $j\in\Gamma_{(t-k+1)}\backslash\Gamma_{(t-k)}$, $\kappa_{j}^{(t)}=0$ if and only if $k$ is divisible by $p$, which completes the proof.
    \end{proof}
\end{lemma}

\begin{theorem}\label{Th1}
    Let $m=2t-1\geq 3$ be odd. Define $\Gamma_{(t)}=\{j\in D_{t}(q):q\nmid j\}$, where $D_{t}(q)$ is defined as in $(\ref{Eq1})$. Let $s^{\infty}$ be the sequence defined as in $(\ref{Eq3})$ through the monomial $F_{3}(x)=x^{\frac{q^{t}-1}{q-1}}$ over $\mathbb{F}_{q^m}$. Then the linear span $\mathcal{L}_{s}$ and the minimal polynomial $\mathfrak{M}_{s}(x)$ of $s^{\infty}$ are given by 
\begin{flalign*}
    \mathcal{L}_{s} &= \mathbb{N}_{p}(m)+\left(\mathbb{N}_{p}(t)+\sum_{k=1}^{t-1}\mathbb{N}_{p}(k)2^{t-k-1}\right)m
\end{flalign*}
and
    \begin{flalign*}
        \mathfrak{M}_{s}(x) &= (x-1)^{\mathbb{N}_{p}(m)}\left(m_{\alpha^{-1}}(x)\right)^{\mathbb{N}_{p}(t)}\prod_{\substack{k=1 \\\mathbb{N}_{p}(k)=1}}^{t-1}\left(\prod_{j\in\Gamma_{(t-k+1)}\setminus\Gamma_{(t-k)}}m_{\alpha^{-j}}(x) \right)
    \end{flalign*}
    \begin{proof}
       By the help of Lemma \ref{L5}, the sequence $s^{\infty}$ defined in $(\ref{Eq3})$ through the monomial $F_{3}(x)=x^{\frac{q^{t}-1}{q-1}}$ over $\mathbb{F}_{q^m}$ is given by  
       \begin{flalign}\label{Eq4}
      \hskip 15pt     s_{i} &= \operatorname{Tr}_{q}^{q^{m}}\left((\alpha^i+1)^{\frac{q^t-1}{q-1}}\right) \nonumber & \\
           &= \left( m+t\operatorname{Tr}_{q}^{q^m}(\alpha^i)+\sum_{j\in\Gamma_{(2)}\backslash\Gamma_{(1)}}(t-1)\operatorname{Tr}_{q}^{q^m}\left((\alpha^i)^j\right)+\sum_{j\in\Gamma_{(3)}\backslash\Gamma_{(2)}}(t-2)\operatorname{Tr}_{q}^{q^m}\left((\alpha^i)^j\right)+\cdots\right.+ \nonumber \\ & \left. \sum_{j\in\Gamma_{(t-1)}\backslash\Gamma_{(t-2)}}2\operatorname{Tr}_{q}^{q^m}\left((\alpha^i)^j\right)+\sum_{j\in\Gamma_{(t)}\backslash\Gamma_{(t-1)}}\operatorname{Tr}_{q}^{q^m}\left((\alpha^i)^j\right) \right) \pmod{p}; \text{ for all $i\geq 0$.}
       \end{flalign}
       From Lemma \ref{L7}, it can be verified that for an odd integer $m$ and $t=\frac{m+1}{2}=\lceil\frac{m}{2}\rceil$, $\Gamma_{(t)}$ is a subset of $\widehat{\Gamma}$. Since $\Gamma_{(k)}\subset\Gamma_{(k+1)}$ for all $k=1,2,\cdots,t-1$, the $q$-cyclotomic cosets $C_{j_1}$ and $C_{j_2}$ are pairwise disjoint for any pair of distinct $j_{1}$ and $j_{2}$ in $\Gamma_{(t)}$. According to Lemma \ref{L3}, $|C_{j}|=m$ for every $j\in\Gamma_{(t)}$. 
       \vskip 1pt
       Using the fact $|\Gamma_{(t-k+1)}\setminus\Gamma_{(t-k)}|=2^{t-k}-2^{t-k-1}=2^{t-k-1}$, we conclude the desired conclusions on the linear span $\mathcal{L}_{s}$ and the minimal polynomial $\mathfrak{M}_{s}(x)$ of $s^{\infty}$ from Lemma \ref{L6} and Eq. $(\ref{Eq4})$. 
    \end{proof}
\end{theorem}
    \begin{theorem}\label{Th10}
       Let $q=p^e$ and $m=2t-1\geq 3$ be odd. Suppose the code $\mathcal{C}_{F_{3}}$ be defined by the sequence $s^{\infty}$ through the monomial $F_{3}(x)$ over $\mathbb{F}_{q^m}$ of Theorem $\ref{Th1}$. Then $\mathcal{C}_{F_{3}}$ has parameters $[q^m-1,q^m-1-\mathcal{L}_{s},d]$ over $\mathbb{F}_{q}$ with the generator polynomial $\mathfrak{M}_{s}(x)$, where $\mathcal{L}_{s}$ and $\mathfrak{M}_{s}(x)$ are given in Theorem $\ref{Th1}$. In addition, 
        \begin{flalign*}
       \hspace{15pt}   d & \geq \begin{cases}
                 2^{t-2}+2, \text{ if $q=2$ and $t>2$;} \\
                 3,\text{ if $p=2$ with $e\geq 2$, and $p\nmid t$;} \\
                  2,\text{ if $p=2$ with $e\geq 2$, and $p\mid t$;} \\
                4,\text{ if $p\neq 2$, $p\nmid t$ and $t>2$;} \\
                3,\text{ if $p\neq 2$, $p\mid t$ and $t>2$.}
               
             %   4,\text{ if $p\neq 2$ and $p\nmid t$;} \\
             %   2,\text{ if $p\neq 2$ and $p\mid t$.} \\
            \end{cases} &
        \end{flalign*}
    \begin{proof}
        The dimension of the code $\mathcal{C}_{F_{3}}$ directly follows from Theorem $\ref{Th1}$. We now determine the lower bounds on the minimum weight of the code $\mathcal{C}_{F_{3}}$. 
        \par For $p=2$ with $e\geq 2$, $\mathbb{N}_{p}(m)=1$, since $m$ is odd. It can be verified that $m_{\alpha^{-1}}(x)$ divides the generator polynomial $\mathfrak{M}_{s}(x)$ of $\mathcal{C}_{F_{3}}$ given in Theorem $\ref{Th1}$ if and only if $p\nmid t$. Thus, $\alpha^{-q^{t-1}}$, $\alpha^{-(1+q^{t-1})}$ are the roots of $\mathfrak{M}_{s}(x)$ if $p\nmid t$. The BCH bound implies $d\geq 3$ for $p\nmid t$; and $d\geq 2$, otherwise. 

      \par  For $p\neq 2$, due to the fact $\mathbb{N}_{p}(2)=1$, we can say $m_{\alpha^{-(1+q^{t-2})}}(x)$ is a factor of the generator polynomial $\mathfrak{M}_{s}(x)$ in Theorem $\ref{Th1}$. Since $m_{\alpha^{-(1+q^{t-2})}}(x)=\prod_{s\in C_{1+q^{t-2}}}(1-\alpha^{s}x)$, we know that $\alpha^{-(q+q^{t-1})}$ is also a root of $m_{\alpha^{-(1+q^{t-2})}}(x)$ and hence a root of $\mathfrak{M}_{s}(x)$ given in Theorem $\ref{Th1}$. When $p\nmid t$, since $m_{\alpha^{-1}}(x)$ divides $\mathfrak{M}_{s}(x)$, one can verify that the set $A=\{q^{t-1},1+q^{t-1}\}$ is a subset of the defining set of the reciprocal of $\mathfrak{M}_{s}(x)$. If we take $B=\{qj:j=0,1\}$, then $A+B$ is also contained in the defining set of the reciprocal of $\mathfrak{M}_{s}(x)$. Since $\operatorname{gcd}(q,v)<3$ and because the code generated by the reciprocal of $\mathfrak{M}_{s}(x)$ and the code $\mathcal{C}_{F_{3}}$ both have same weight distribution, applying the Hartmann-Tzeng bound, we have $d\geq 4$. When $p\mid t$, since $m_{\alpha^{-1}}(x)$ is not a factor of $\mathfrak{M}_{s}(x)$, similarly it can be shown that $d\geq 3$ whenever $t>2$.
        
        \par For $q=2$, the minimum distance $d\geq 2^{t-2}+2$ can be achieved similarly from the Hartmann-Tzeng bound, considering the sets $A=\{1+2^{t-1}\}$ and $B=\{2j:0\leq j\leq 2^{t-2}-1\}$ and combining the fact that $\mathcal{C}_{F_{3}}$ is an even-weight code.
        \vskip 1pt
        Hence, the proof is completed.
    \end{proof}
\end{theorem}
\begin{example}
    Let $q=m=3$, then $\mathbb{N}_{p}(m)=0$ and $t=\frac{m+1}{2}=2$. Let $\alpha$ be a root of the primitive polynomial $x^3 + 2x + 1$ over $\mathbb{F}_{3}$. The generator polynomial of the cyclic code $\mathcal{C}_{F_{3}}$ is $\mathfrak{M}_{s}(x)=x^6 + 2x^5 + 2x^4 + x^3 + x^2 + 2x + 2$. Then $\mathcal{C}_{F_{3}}$ is an optimal ternary $[26, 20, 4]$ cyclic code, and its dual code $\mathcal{C}_{F_{3}}^{\perp}$ is an optimal $[26, 6, 15]$ ternary cyclic code. Both the optimal ternary linear codes with parameters $[26, 20, 4]$ and $[26, 6, 15]$ in the Database $\textnormal{\cite{Database}}$ are not cyclic.  
\end{example}
\vspace{1em}
\begin{example}
    Let $q=5$ and $m=3$, then $\mathbb{N}_{p}(m)=1$ and $t=\frac{m+1}{2}=2$. Let $\alpha$ be a root of the primitive polynomial $x^3 + 3x + 3$ over $\mathbb{F}_{5}$. The generator polynomial of the cyclic code $\mathcal{C}_{F_{3}}$ is $\mathfrak{M}_{s}(x)=x^7 + x^5 +  3x^4 + x^3 + 3x^2 + 3x + 3$. Then $\mathcal{C}_{F_{3}}$ is an optimal quinary $[124, 117, 4]$ cyclic code. Its dual code $\mathcal{C}_{F_{3}}^{\perp}$ is an optimal quinary $[124,7,94]$ cyclic code. Both the optimal quinary linear codes with parameters $[124, 117, 4]$ and $[124,7,94]$ in the Database $\textnormal{\cite{Database}}$ are not cyclic. % According to the Database $\textnormal{\cite{Database}}$, both codes $\mathcal{C}_{F}$ and $\mathcal{C}_{F}^{\perp}$ are optimal.  
\end{example}
\vspace{1em}
\begin{example}
    Let $q=4$ and $m=3$, then $\mathbb{N}_{p}(m)=1$ and $t=\frac{m+1}{2}=2$. Let $\alpha$ be a root of the primitive polynomial $x^3 + x^2 + x + \omega$ over $\mathbb{F}_{4}=\mathbb{F}_{2}(\omega)$. The generator polynomial of the cyclic code $\mathcal{C}_{F_{3}}$ is $\mathfrak{M}_{s}(x)=x^4 + \omega x^3 + x^2 + \omega$. Then $\mathcal{C}_{F_{3}}$ is an optimal quaternary $[63, 59, 3]$ cyclic code, and its dual code $\mathcal{C}_{F_{3}}^{\perp}$ is an optimal quaternary $[63,4,47]$ cyclic code. Both the optimal quaternary linear codes with parameters $[63, 59, 3]$ and $[63,4,47]$ in the Database $\textnormal{\cite{Database}}$ are not cyclic.
\end{example}
\vspace{1em}
\begin{example}
    Let $q=3$ and $m=5$, then $\mathbb{N}_{p}(m)=1$, $t=\frac{m+1}{2}=3$, and $\mathbb{N}_{p}(t)=0$. Let $\alpha$ be a root of the primitive polynomial $x^5 + 2x + 1$ over $\mathbb{F}_{3}$. The generator polynomial of the cyclic code $\mathcal{C}_{F_{3}}$ is $\mathfrak{M}_{s}(x)=x^{16} + 2x^{14} +
    2x^{12} + 2x^{11} + x^{10} + x^9 + x^6 + x^3 + 2x^2 + 2$. Then $\mathcal{C}_{F_{3}}$ is a ternary $[242, 226, 5]$ cyclic code, and its dual code $\mathcal{C}_{F_{3}}^{\perp}$ is a ternary $[242, 16, 131]$ cyclic code. According to the Database $\textnormal{\cite{Database}}$, the parameter of the ternary cyclic code $\mathcal{C}_{F_{3}}$ is almost optimal.  %  The optimal ternary linear code with parameter $[242, 227, 5]$ in the Database $\textnormal{\cite{Database}}$ is not cyclic.
\end{example}
\vspace{1em}
\begin{remark}\label{Remark1}
 In $\textnormal{\cite[Section $5.3$]{P3}}$, Ding utilized the monomial $F_{3}(x)=x^{\frac{q^{t}-1}{q-1}}$ over $\mathbb{F}_{q^m}$ under the restrictions on $t$ as follows:
 \begin{equation*}
     1\leq t\leq \begin{cases}
     (m-1)/2;\text{ if $m$ is odd,} \\
     m/2;\text{ if $m$ is even.}
     \end{cases}
 \end{equation*}
 and investigated the cyclic code $\mathcal{C}_{F_{3}}$ to determine its generator polynomial and the bounds of its minimum weight. In this section, we chose $t=\frac{m+1}{2}$ and $m$ to be odd for studying the $q$-ary cyclic code $\mathcal{C}_{F_{3}}$. Particularly for $q=3$, Theorem $\ref{Th10}$ determines the dimension and the generator polynomial of the ternary code $\mathcal{C}_{F_{3}}$, which partially solves the open problem $5.31$ proposed in $\textnormal{\cite{P3}}$.
\end{remark}

%\subsection{Cyclic codes from $x^{2^{m+1}+3}$ over $\mathbb{F}_{2^{2m}}$, $m\geq 5$ is odd}
%In this subsection,

%\subsection{Cyclic codes from $x^{2^{m}+2^{(m+1)/2}+1}$ over $\mathbb{F}_{2^{2m}}$, $m\geq 5$ is odd}
%In this subsection,

%\subsection{Cyclic codes from $x^{p^{2h}-p^h+1}$ over $\mathbb{F}_{p^m}$, where $p=3$, $m$ is odd, and $\operatorname{gcd}(m,h)=1$}
%In this subsection,

\subsection{Cyclic codes from $x^{\frac{3^{(m+1)/2}+1}{2}}$ over $\mathbb{F}_{3^{m}}$, where $m$ is odd}
In this subsection, we deal with the ternary cyclic code $\mathcal{C}_{F_{4}}$ defined by the sequence $s^{\infty}$ of $(\ref{Eq3})$ through the monomial $F_{4}(x)=x^{\frac{3^{(m+1)/2}+1}{2}}$ over $\mathbb{F}_{3^{m}}$, where $m$ is an odd integer. The monomial $F_{4}$ over $\mathbb{F}_{3^m}$ is a perfect nonlinear (PN) function if $\frac{m+1}{2}$ is odd \cite[Theorem $4.1$]{Uniform3}, and a $4$-uniform function if $\frac{m+1}{2}$ is even \cite[Theorem $2$]{Uniform4}.
\par Before proving the main results of this section, we need to observe some general structures and present an important lemma. 
\vskip 1pt
We consider the following two sets
\begin{align}\label{Eq12}
    D_{t}(q)&=\{\sum_{i=0}^{t-1}a_{i}q^{i}:a_{i}\text{'s, not all zero, is either $0$ or $1$ for all $i=0,1\cdots,t-1$}\} \nonumber \\
   & \text{ and }  \\
    \Gamma_{(t)}&=\{j\in D_{t}(q):q\nmid j\},\text{ where $t$ is a positive integer.} \nonumber
\end{align}
%; and let $\Gamma_{(t)}=\{j\in D_{t}(q):q\nmid j\}$. 
%\vspace{1em}
\begin{observation}\label{Ob1}
The above defined sets $D_{t}(q)$ and $\Gamma_{(t)}$ can be observed as follows: 
    \begin{itemize}
    \item $D_{t}(q)=\bigcup_{k=0}^{t-1}q^{k}\Gamma_{(t-k)}$.
    \item $1+D_{t}(q)=\{1+j:j\in D_{t}(q)\}=\left(\Gamma_{(t)}\backslash\{1\}\right)\cup \left(1+\Gamma_{(t)}\right)$.
    \item $\Gamma_{(t)}\backslash\{1\}=\bigcup_{k=1}^{t-1}\left(\Gamma_{(t-k+1)}\backslash\Gamma_{(t-k)}\right)$.
   % \item $1+\Gamma_{(t)}=\{1+j:j\in\Gamma_{(t)}\}=$
\end{itemize}
\end{observation}

\begin{lemma}\label{L9}
   Let $m$ and $t$ be two positive integers. Suppose $D_{t}(q)$ and $\Gamma_{(t)}$ are defined as in $(\ref{Eq12})$. Then, for any $x\in\mathbb{F}_{q^m}$, we have 
   \[\begin{multlined}
       \operatorname{Tr}_{q}^{q^m}\left((x+1)^{\frac{q^t-1}{q-1}+1}\right) = \left( m+(t+1)\operatorname{Tr}_{q}^{q^m}(x)+\sum_{j\in\Gamma_{(2)}\backslash\Gamma_{(1)}}t\operatorname{Tr}_{q}^{q^m}(x^j)+\sum_{j\in\Gamma_{(3)}\backslash\Gamma_{(2)}}(t-1)\operatorname{Tr}_{q}^{q^m}(x^j)+\cdots\right.+\\ \left. \sum_{j\in\Gamma_{(t-1)}\backslash\Gamma_{(t-2)}}3\operatorname{Tr}_{q}^{q^m}(x^j)+\sum_{j\in\Gamma_{(t)}\backslash\Gamma_{(t-1)}}2\operatorname{Tr}_{q}^{q^m}(x^j)+\sum_{j\in\Gamma_{(t)}}\operatorname{Tr}_{q}^{q^m}(x^{1+j}) \right) \pmod{p}
  \end{multlined}\]
   \begin{proof}
       Note that $q=p^e$ and $1+p^e+p^{2e}+\cdots+p^{(t-1)e}$ is the $p$-adic representation of $\frac{q^e-1}{q-1}$. Let $j$ be a non-negative integer such that $0\leq j\leq \frac{q^e-1}{q-1}$, with its $p$-adic representation $j=j_{0}+j_{1}p+j_{2}p^2+\cdots+j_{(t-1)e}p^{(t-1)e}$, where $0\leq j_{0}, j_{1},j_{2},\cdots,j_{(t-1)e}\leq p-1$. Then, from Lucas' Theorem \cite{P1}, one can derive that $\binom{\frac{q^t-1}{q-1}}{j}\equiv\binom{1}{j_{0}}\binom{0}{j_{1}}\cdots\binom{1}{j_{e}}\binom{0}{j_{e+1}}\cdots\binom{1}{j_{(t-1)e}}\not\equiv 0\pmod{p}$ if and only if $j_{s}\in\{0,1\}$ for all $e\mid s$ and $j_{s}=0$ for all $e\nmid s$, where $s\in\{0,1,2,\cdots,(t-1)e\}$. Hence, $\binom{\frac{q^t-1}{q-1}}{j}\equiv 1\pmod{p}$ for all $j\in D_{t}(q)\cup\{0\}$, and $p\large\mid\binom{\frac{q^t-1}{q-1}}{j}$, otherwise. Since $\operatorname{char}(\mathbb{F}_{q^m})=p$, for any $x\in\mathbb{F}_{q^m}$, we can write
       \begin{flalign}\label{Eq6}
   \hspace{30pt}    \operatorname{Tr}_{q}^{q^m}\left((x+1)^{\frac{q^t-1}{q-1}+1}\right) &= \operatorname{Tr}_{q}^{q^m}\left[\left(x+1\right)\left(1+\sum_{j\in D_{t}(q)}x^{j}\right)\right] & \nonumber \\
   &= \operatorname{Tr}_{q}^{q^m}\left(x+\sum_{j\in D_{t}(q)}x^{1+j}\right)+\operatorname{Tr}_{q}^{q^m}\left(1+\sum_{j\in D_{t}(q)}x^j\right)
   \end{flalign}
   From Eq. $(\ref{Eq2})$ and Lemma $\ref{L5}$, we have
   \begin{align}\label{Eq7}
       \operatorname{Tr}_{q}^{q^m}\left(1+\sum_{j\in D_{t}(q)}x^j\right) &= \operatorname{Tr}_{q}^{q^m}\left(1\right)+\kappa_{1}^{(t)}\operatorname{Tr}_{q}^{q^m}\left(x\right)+\sum_{k=1}^{t-1}\left(\sum_{j\in\Gamma_{(t-k+1)}\backslash\Gamma_{(t-k)}}\kappa_{j}^{(t)}\operatorname{Tr}_{q}^{q^m}(x^j)\right) \nonumber \\
       &= \left( m+t\operatorname{Tr}_{q}^{q^m}(x)+\sum_{j\in\Gamma_{(2)}\backslash\Gamma_{(1)}}(t-1)\operatorname{Tr}_{q}^{q^m}(x^j)+\sum_{j\in\Gamma_{(3)}\backslash\Gamma_{(2)}}(t-2)\operatorname{Tr}_{q}^{q^m}(x^j)+\cdots\right.+ \nonumber \\ &\left. \sum_{j\in\Gamma_{(t-1)}\backslash\Gamma_{(t-2)}}2\operatorname{Tr}_{q}^{q^m}(x^j)+\sum_{j\in\Gamma_{(t)}\backslash\Gamma_{(t-1)}}\operatorname{Tr}_{q}^{q^m}(x^j) \right) \pmod{p}
   \end{align}
   Due to the facts in observation $\ref{Ob1}$, we can write
   \begin{align}\label{Eq8}
      \operatorname{Tr}_{q}^{q^m}\left(\sum_{j\in D_{t}(q)}x^{1+j}\right) &= \operatorname{Tr}_{q}^{q^m}\left(\sum_{j\in\Gamma_{(t)}}x^{1+j}\right)+ \operatorname{Tr}_{q}^{q^m}\left(\sum_{j\in \Gamma_{(t)}\backslash\{1\}}x^{j}\right) \nonumber \\
      &= \operatorname{Tr}_{q}^{q^m}\left(\sum_{j\in\Gamma_{(t)}}x^{1+j}+\sum_{j\in \Gamma_{(t)}\backslash\Gamma_{(t-1)}}x^{j}+\sum_{j\in \Gamma_{(t-1)}\backslash\Gamma_{(t-2)}}x^{j}+\cdots+\sum_{j\in \Gamma_{(2)}\backslash\Gamma_{(1)}}x^{j}\right)
   \end{align}
   Hence, the desired conclusion follows by substituting the values of $\operatorname{Tr}_{q}^{q^m}\left(1+\sum_{j\in D_{t}(q)}x^j\right)$ and $\operatorname{Tr}_{q}^{q^m}\left(\sum_{j\in D_{t}(q)}x^{1+j}\right)$ from their respective equations $(\ref{Eq7})$ and $(\ref{Eq8})$ into Eq. $(\ref{Eq6})$. 
   \end{proof}
\end{lemma}

\begin{theorem}\label{Th2}
    Let $m=2t-1\geq 3$ be odd. Suppose $D_{t}(3)$ and its corresponding subset $\Gamma_{(t)}$ are defined as in Eq. $(\ref{Eq12})$, when $q=3$. Let $s^{\infty}$ be the sequence defined in $(\ref{Eq3})$ through the monomial $F_{4}(x)=x^{\frac{3^{t}+1}{2}}$ over $\mathbb{F}_{3^m}$. Then the linear span $\mathcal{L}_{s}$ and the minimal polynomial $\mathfrak{M}_{s}(x)$ of $s^{\infty}$ are given by
\begin{flalign*}
    \mathcal{L}_{s} &= \mathbb{N}_{3}(m)+\left(\mathbb{N}_{3}(t+1)+\sum_{k=1}^{t}\mathbb{N}_{3}(k)2^{t-k}\right)m
\end{flalign*}
and
    \begin{flalign*}
        \mathfrak{M}_{s}(x) &= (x-1)^{\mathbb{N}_{3}(m)}\left(m_{\alpha^{-1}}(x)\right)^{\mathbb{N}_{3}(t+1)}\times\prod_{\substack{k=2 \\\mathbb{N}_{3}(k)=1}}^{t}\left(\prod_{j\in\Gamma_{(t-k+2)}\setminus\Gamma_{(t-k+1)}}m_{\alpha^{-j}}(x) \right)\times\prod_{j\in\Gamma_{(t)}}m_{\alpha^{-1-j}}(x).
    \end{flalign*}
    \begin{proof}
        By the help of Lemma $\ref{L9}$, the sequence $s^{\infty}$ defined in $(\ref{Eq3})$ through the monomial $F_{4}(x)=x^{\frac{3^{t}+1}{2}}$ over $\mathbb{F}_{3^m}$ is given by
        
        \begin{flalign}\label{Eq9}
        \hspace{30pt}    s_{i} &= \operatorname{Tr}_{q}^{q^m}\left((\alpha^i+1)^{\frac{3^t-1}{3-1}+1}\right) & \nonumber \\
        &=  \left( m+(t+1)\operatorname{Tr}_{q}^{q^m}(\alpha^i)+\sum_{j\in\Gamma_{(2)}\backslash\Gamma_{(1)}}t\operatorname{Tr}_{q}^{q^m}((\alpha^i)^j)+\sum_{j\in\Gamma_{(3)}\backslash\Gamma_{(2)}}(t-1)\operatorname{Tr}_{q}^{q^m}((\alpha^i)^j)+\cdots\right.+ & \nonumber \\ &\left. \sum_{j\in\Gamma_{(t)}\backslash\Gamma_{(t-1)}}2\operatorname{Tr}_{q}^{q^m}((\alpha^i)^j)+\sum_{j\in\Gamma_{(t)}}\operatorname{Tr}_{q}^{q^m}((\alpha^i)^{1+j}) \right) \pmod{3};\text{ for all }i\geq 0. 
        \end{flalign}
        Particularly when $q=3$ and $t=\frac{m+1}{2}$, according to the definition of $\widehat{\Gamma}$ in Lemma $\ref{L7}$, note that $1+j\in\widehat{\Gamma}$ for every $j\in\Gamma_{(t)}$ and $\Gamma_{(t)}=\{1\}\cup\left(\bigcup_{k=1}^{t-1}\left(\Gamma_{(t-k+1)}\backslash\Gamma_{(t-k)}\right)\right)\subset\widehat{\Gamma}$, which implies $C_{1+j_1}\cap C_{1+j_2}=\emptyset$ and $C_{j_1}\cap C_{j_2}=\emptyset$ for any distinct pair of $j_1,j_2\in\Gamma_{(t)}$. Since $\Gamma_{(t)}\cap(1+ \Gamma_{(t)})=\emptyset$, we conclude that $C_{j}\cap C_{1+j'}=\emptyset$ for any $j,j'\in\Gamma_{(t)}$. According to Lemma $\ref{L3}$, $|C_{j}|=m$ for every $j\in\Gamma_{(t)}\cup (1+\Gamma_{(t)})$.
        \vskip 1pt
Using the facts $|\Gamma_{(t)}|=2^{t-1}$ and $|\Gamma_{(t-k+2)}\setminus\Gamma_{(t-k+1)}|=2^{t-k+1}-2^{t-k}=2^{t-k}$ for $k=2,3,\cdots,t$, we conclude the desired conclusions on the linear span $\mathcal{L}_{s}$ and the minimal polynomial $\mathfrak{M}_{s}(x)$ of $s^{\infty}$ from Lemma \ref{L6} and Eq. $(\ref{Eq9})$.
    \end{proof}
\end{theorem}

\begin{theorem}\label{Th13}
        Let $m=2t-1\geq 3$ be odd. Suppose the code $\mathcal{C}_{F_{4}}$ be defined by the sequence $s^{\infty}$ through the monomial $F_{4}(x)$ over $\mathbb{F}_{3^m}$ of Theorem $\ref{Th2}$. Then $\mathcal{C}_{F_{4}}$ has parameters $[3^m-1,3^m-1-\mathcal{L}_{s},d]$ over $\mathbb{F}_{3}$ with the generator polynomial $\mathfrak{M}_{s}(x)$, where $\mathcal{L}_{s}$ and $\mathfrak{M}_{s}(x)$ are given in Theorem $\ref{Th2}$. In addition, 
        \begin{flalign*}
     \hspace{15pt}       d\geq & \begin{cases}
                8,\text{ if }3\nmid m,\text{ }3\nmid t \text{ and }3\nmid (t+1); \\
                7,\text{ if }3\mid m,\text{ }3\nmid t\text{ and }3\nmid(t+1); \\
                5,\text{ if }3\nmid m,\text{ }3\mid t\text{ and }3\nmid (t+1); \\
                4,\text{ if }3\mid m,\text{ }3\mid t\text{ and }3\nmid (t+1); \\
                4,\text{ if }3\nmid t\text{ and }3\mid(t+1);\\
                3,\text{ otherwise.}
            \end{cases} &
        \end{flalign*}
        \begin{proof}
        The dimension of the code $\mathcal{C}_{F_{4}}$ directly follows from Theorem $\ref{Th2}$. We now determine the lower bounds on the minimum weight of the code $\mathcal{C}_{F_{4}}$.
            \vskip 1pt
            When $3\nmid m$, $3\nmid t$ and $3\nmid (t+1)$, one can note that $\mathbb{N}_{3}(m)=\mathbb{N}_{3}(t)=\mathbb{N}_{3}(t+1)=1$, $2,5\in 1+\Gamma_{(t)}$, and $4\in\Gamma_{(2)}\backslash\Gamma_{(1)}$. Then it can be verified from the generator polynomial $\mathfrak{M}_{s}(x)$ of $\mathcal{C}_{F_{4}}$ in Theorem $\ref{Th2}$ that $1$, $\alpha^{-1}$, $\alpha^{-2}$, $\alpha^{-3}$, $\alpha^{-4}$, $\alpha^{-5}$, and $\alpha^{-6}$ are the roots of $\mathfrak{M}_{s}(x)$. Hence, the BCH bound gives $d\geq 8$.

            \par When $3\mid m$, $3\nmid t$ and $3\nmid (t+1)$, since $\mathbb{N}_{3}(m)=0$, $\mathbb{N}_{3}(t)=\mathbb{N}_{3}(t+1)=1$, $2,5\in 1+\Gamma_{(t)}$ and $4\in\Gamma_{(2)}\backslash\Gamma_{(1)}$, it can be verified from the generator polynomial $\mathfrak{M}_{s}(x)$ of $\mathcal{C}_{F_{4}}$ in Theorem $\ref{Th2}$ that $\alpha^{-1}$, $\alpha^{-2}$, $\alpha^{-3}$, $\alpha^{-4}$, $\alpha^{-5}$ and $\alpha^{-6}$ are the roots of $\mathfrak{M}_{s}(x)$. Hence, the BCH bound implies $d\geq 7$.
            
            \par When $3\nmid m$, $3\mid t$ and $3\nmid (t+1)$, since $\mathbb{N}_{3}(t)=0$, $\mathbb{N}_{3}(m)=\mathbb{N}_{3}(t+1)=1$, and $2\in 1+\Gamma_{(t)}$, from the the generator polynomial $\mathfrak{M}_{s}(x)$ of $\mathcal{C}_{F_{4}}$ in Theorem $\ref{Th2}$, it can be checked that $1$, $\alpha^{-1}$, $\alpha^{-2}$, and $\alpha^{-3}$ are the roots of $\mathfrak{M}_{s}(x)$. Hence, the BCH bound gives $d\geq 5$.
            
            \par When $3\mid m$, $3\mid t$ and $3\nmid (t+1)$, it can be checked that $\alpha^{-1}$, $\alpha^{-2}$ and $\alpha^{-3}$ are the roots of $\mathfrak{M}_{s}(x)$ in Theorem $\ref{Th2}$. Hence, the BCH bound gives $d\geq 4$.
            
            \par When $3\nmid t$ and $3\mid(t+1)$, since $4\in\Gamma_{(2)}\backslash\Gamma_{(1)}$ and $2,5\in 1+\Gamma_{(t)}$, it can be checked that $\alpha^{-4}$, $\alpha^{-5}$ and $\alpha^{-6}$ are the roots of $\mathfrak{M}_{s}(x)$ in Theorem $\ref{Th2}$. Hence, the BCH bound gives $d\geq 4$.
            
            \par When $3\mid t$ and $3\mid(t+1)$, since $\alpha^{-(1+3^{t-1})}$ and $\alpha^{-(2+3^{t-1})}$ are the roots of $\mathfrak{M}_{s}(x)$ in Theorem $\ref{Th2}$. Hence, the BCH bound implies $d\geq 3$.
        \end{proof}
    \end{theorem}
\begin{example}
    Let $m=3$, then $t=\frac{m+1}{2}=2$ and $\mathbb{N}_{3}(m)=\mathbb{N}_{3}(t+1)=0$. Let $\alpha$ be a root of the primitive polynomial $x^3 + 2x + 1$ over $\mathbb{F}_{3}$. The generator polynomial of the cyclic code $\mathcal{C}_{F_{4}}$ is $\mathfrak{M}_{s}(x)=x^9 + 2x^7 +  2x^6 + 2x^5 + x^2 + x + 1$. Then $\mathcal{C}_{F_{4}}$ is a ternary $[26, 17, 4]$ cyclic code, and its dual code $\mathcal{C}_{F_{4}}^{\perp}$ is a $[26, 9, 9]$ ternary cyclic code. % Both the optimal ternary linear codes with parameters $[26, 20, 4]$ and $[26, 6, 15]$ in the Database $\textnormal{\cite{Database}}$ are not cyclic.
\end{example}
\vspace{1em}
\begin{example}
    Let $m=5$, then $t=\frac{m+1}{2}=3$ and $\mathbb{N}_{3}(m)=\mathbb{N}_{3}(t+1)=1$. Let $\alpha$ be a root of the primitive polynomial $x^5 + 2x + 1$ over $\mathbb{F}_{3}$. The generator polynomial of the cyclic code $\mathcal{C}_{F_{4}}$ is $\mathfrak{M}_{s}(x)=x^{36} + x^{35} +
    2x^{34} + x^{32} + 2x^{30} + x^{29} + x^{28} + 2x^{26} + x^{24} + x^{23} + x^{22} + x^{20} +
    2x^{19} + x^{18} + 2x^{16} + x^{15} + 2x^{14} + x^{13} + x^{12} + 2x^{11} + 2x^{10} +
    2x^8 + 2x^7 + 2x^6 + 2x^3 + x + 1$. Then $\mathcal{C}_{F_{4}}$ is a ternary $[242,206,d]$ cyclic code, where $8\leq d\leq 10$, and its dual code $\mathcal{C}_{F_{4}}^{\perp}$ is a $[242,36,81]$ ternary cyclic code.
\end{example}
\begin{remark}
    In \textnormal{\cite[Section $5.4$]{P3}}, Ding employed the monomial $F_{4}(x)=x^{\frac{3^{t}+1}{2}}$ over $\mathbb{F}_{3^m}$ under the restrictions on $t$ as follows:
    \begin{align}\label{Eq26}
        &\begin{cases}
            \text{$t$ is odd}, \\
            \operatorname{gcd}(m,t)=1, \\
            3\leq t\leq \begin{cases}
                (m-1)/2\text{ if $m$ is odd and} \\
                m/2\text{ if $m$ is even.}
            \end{cases}
        \end{cases}
    \end{align}
    and determined the generator polynomial of the cyclic code $\mathcal{C}_{F_{4}}$. However, the bounds on the minimum distances of $\mathcal{C}_{F_{4}}$ are not determined due to the complex expression of its generator polynomial. In this subsection, we chose $t=\frac{m+1}{2}$ and $m$ to be odd to study the ternary cyclic code $\mathcal{C}_{F_{4}}$. We explicitly determine the generator polynomial and the lower bounds on the minimum distance of $\mathcal{C}_{F_{4}}$. By the help of Lemma $\ref{L9}$, one can determine the lower bounds on the minimum distance of the cyclic code $\mathcal{C}_{F_{4}}$ when $t$ satisfies $(\ref{Eq26})$.  
\end{remark}
    \subsection{Cyclic codes from $x^{p^h-1}$ over $\mathbb{F}_{p^{m}}$, where $p\geq 3$}
    In this section, we construct $p$-ary cyclic codes $\mathcal{C}_{F_{5}}$ defined by the sequence $s^{\infty}$ of $(\ref{Eq3})$ through the monomial $F_{5}(x)=x^{p^h-1}$ over $\mathbb{F}_{p^m}$, where $p$ is an odd prime and $h$ is a positive integer with $1\leq h\leq\lceil\frac{m}{2}\rceil$. The differential uniformity $\Delta_{F_{5}}$ of the power function $F_{5}$ over $\mathbb{F}_{p^m}$ is $p^{m/2}-2$ if $h=\frac{m}{2}$ and $m$ is even \cite{Uniform5}.
    \vskip 1pt
    Before proving the main results of this section, we first need some preparations. Let $t$ be a positive integer. Consider the following two sets
     \begin{align}\label{Eq10}
         \widehat{D_{t}(p)} &= \{j\in\mathbb{N}:1\leq j\leq p^{t}-1\}              
        \text{ and } 
         \widehat{\Gamma_{(t)}}= \{j\in \widehat{D_{t}(p)}:p\nmid j\}.  
     \end{align}
     %Clearly, $\overline{D_{t}(3)}=\{\sum_{i=0}^{t-1}a_{i}3^{i}:a_{i}\text{'s, not all zero, is either $0$, $1$ or $2$ for all }i=0,1,\cdots,t-1\}$. 
     For $j\in\widehat{\Gamma_{(t)}}$, define the following notations
     \begin{align}\label{Eq11}
         \hat{\epsilon}_{j}^{(t)} &= \begin{cases}
             1,\text{ if }j=p^t-1 \\
             \Big\lceil\operatorname{log}_{p}\left(\frac{p^t-1}{j}\right)\Big\rceil,\text{ if }1\leq j<p^t-1 \\
         \end{cases} 
     \end{align}
     and 
     \begin{align*}
         \hat{\kappa}_{j}^{(t)}= \hat{\epsilon}_{j}^{(t)}\pmod{p}
     \end{align*}
     \begin{lemma}\label{L10}
         Let $\hat{B}_{j}^{(t)}=\{jp^{i}:i=0,1,2,\cdots,\hat{\epsilon}_{j}^{(t)}-1\}$ for some $j\in\widehat{\Gamma_{(t)}}$, where $\widehat{\Gamma_{(t)}}$ and $\hat{\epsilon}_{j}^{(t)}$ are defined as in $(\ref{Eq10})$ and $(\ref{Eq11})$ respectively. Then the following statements are true.
         \begin{itemize}
             \item $\hat{B}_{j_{1}}^{(t)}\bigcap \hat{B}_{j_{2}}^{(t)}=\emptyset$ for any pair of distinct $j_{1}$ and $j_{2}$ in $\widehat{\Gamma_{(t)}}$.
             \item $\bigcup_{j\in\widehat{\Gamma_{(t)}}}\hat{B}_{j}^{(t)}=\{1,2,\cdots,p^t-1\}$.
           %  \item For $j\in\widehat{\Gamma_{(t)}}$,
           %  $\hat{B}_{j}^{(t+1)}=\hat{B}_{j}^{(t)}\cup\{j3^{\hat{\epsilon}_{j}^{(t)}}\}$ and $\hat{\epsilon}_{j}^{(t+1)}=\hat{\epsilon}_{j}^{(t)}+1$.
           %  \item For $j\in\widehat{\Gamma_{(t+1)}}\backslash\widehat{\Gamma_{(t)}}$, $\hat{B}_{j}^{(t)}=\{j\}$ and $\hat{\epsilon}_{j}^{(t+1)}=1$.
             \item If $j\in\widehat{\Gamma_{(t+1)}}$, then 
             \begin{enumerate}
                 \item[] $\hat{B}_{j}^{(t+1)}=\hat{B}_{j}^{(t)}\cup\{jp^{\hat{\epsilon}_{j}^{(t)}}\}$ and $\hat{\epsilon}_{j}^{(t+1)}=\hat{\epsilon}_{j}^{(t)}+1$ for $j\in\widehat{\Gamma_{(t)}}$. 
                 \item[] $\hat{B}_{j}^{(t)}=\{j\}$ and $\hat{\epsilon}_{j}^{(t+1)}=1$ for $j\in\widehat{\Gamma_{(t+1)}}\backslash\widehat{\Gamma_{(t)}}$.
             \end{enumerate}
         \end{itemize}
         \begin{proof}
             Suppose $\hat{B}_{j_{1}}^{(t)}\bigcap \hat{B}_{j_{2}}^{(t)}\neq\emptyset$ for some pair of distinct $j_{1}$ and $j_{2}$ in $\widehat{\Gamma_{(t)}}$. Then $j_{1}p^{i_{1}}=j_{2}p^{i_{2}}$ for some distinct $i_{1}$ and $i_{2}$ in 
            $\{0,1,2,\cdots,\hat{\epsilon}_{j}^{(t)}-1\}$. Without any loss of generality, assume that $i_{2}>i_{1}$, which implies $j_{1}=j_{2}p^{i_{2}-i_{1}}$. This contradicts the definition of $\widehat{\Gamma_{(t)}}$ that $p\nmid j_{1}$. This proofs the first assertion.
            \par From the definition of $\hat{\epsilon}_{j}^{(t)}$, one can check that $\hat{B}_{p^t-1}^{(t)}=\{p^t-1\}$ and $jp^{\hat{\epsilon}_{j}^{(t)}-1}<p^t-1$ for any $j\in\widehat{\Gamma_{(t)}}\backslash\{p^t-1\}$. Hence, 
            $\hat{B}_{j}^{(t)}=\{jp^{i}:i=0,1,2,\cdots,\hat{\epsilon}_{j}^{(t)}-1\}\subset\{1,2,\cdots,p^t-1\}$ for every $j\in\widehat{\Gamma_{(t)}}$. Thus, the second assertion follows by taking the disjoint union of $\hat{B}_{j}^{(t)}$ over all $j$'s in $\widehat{\Gamma_{(t)}}$.   
\par From the definition of $\hat{\epsilon}_{j}^{(t+1)}$, we know that 
\begin{align*}
         \hat{\epsilon}_{j}^{(t+1)} &= \begin{cases}
             1,\text{ if }j=p^{t+1}-1 \\
             \Big\lceil\operatorname{log}_{p}\left(\frac{p^{t+1}-1}{j}\right)\Big\rceil,\text{ if }1\leq j<p^{t+1}-1 \\
         \end{cases} 
     \end{align*}
     \par
 When $j\in\widehat{\Gamma_{(t)}}$, since $0<\frac{p-1}{j}\leq\frac{p^t-1}{j}$, we can say  $\frac{p(p^t-1)}{j}<\frac{p^{t+1}-1}{j}=\frac{p(p^t-1)}{j}+\frac{p-1}{j}\leq\frac{(p+1)(p^t-1)}{j}$. This gives  $1+\operatorname{log}_{p}\left(\frac{p^t-1}{j}\right)<\operatorname{log}_{p}\left(\frac{p^{t+1}-1}{j}\right)\leq\operatorname{log}_{p}\left(p+1\right)+\operatorname{log}_{p}\left(\frac{p^t-1}{j}\right)$. Since $\operatorname{log}_{p}\left(\frac{p^t-1}{j}\right)$ can not be a nonzero integer, we conclude that $\hat{\epsilon}_{j}^{(t+1)}=\hat{\epsilon}_{j}^{(t)}+1$ for all $j\in\widehat{\Gamma_{(t)}}$.
\par When $j\in\widehat{\Gamma_{(t+1)}}\backslash\widehat{\Gamma_{(t)}}$, note that $1+p^{t}\leq j\leq p^{t+1}-1$. For $j=p^{t+1}-1$, we know that $\hat{\epsilon}_{j}^{(t+1)} =1$ due to the definition. For $j\neq p^{t+1}-1$, $j<p^{t+1}-1<pj$ would give $0<\operatorname{log}_{p}\left(\frac{p^{t+1}-1}{j}\right)<1$. Therefore, $\hat{\epsilon}_{j}^{(t+1)} =1$ for all $j\in\widehat{\Gamma_{(t+1)}}\backslash\widehat{\Gamma_{(t)}}$.
\newline
Thus, the proof of the third assertion follows from the definition of $\hat{B}_{j}^{(t)}$.
         \end{proof}
     \end{lemma}
 \begin{lemma}\label{L11}
     Let $\widehat{\Gamma_{(t)}}$ be defined as in $(\ref{Eq10})$ for some positive integer $t$. Then for $j\in\widehat{\Gamma_{(t)}}$, we have
     \begin{align*}
         \hat{\epsilon}_{j}^{(t)} &= \begin{cases}
             t,\text{ if }j\in \widehat{\Gamma_{(1)}}, \\
             t-k,\text{ if }j\in\widehat{\Gamma_{(k+1)}}\backslash\widehat{\Gamma_{(k)}},\text{ where }k\in\{1,2,3,\cdots,t-1\}.
         \end{cases}
     \end{align*}
     \begin{proof}
         Note that $\widehat{\Gamma_{(1)}}=\{1,2,\cdots,p-1\}$. For $t=1$, $\hat{\epsilon}_{p-1}^{(1)}=1$ due to the definition of $\hat{\epsilon}_{j}^{(t)}$, and since $j<p-1<pj$ for $j\in\{1,2,\cdots,p-2\}$, we obtain $\hat{\epsilon}_{j}^{(1)}=\Big\lceil\operatorname{log}_{p}\left(\frac{p-1}{j}\right)\Big\rceil=1$ for every $j\in\widehat{\Gamma_{(1)}}\backslash\{p-1\}$. Hence, $\hat{\epsilon}_{j}^{(1)}=1$ for all $j\in\widehat{\Gamma_{(1)}}$.
          For $t\geq 2$, since $p^{t-1}j\leq p^{t-1}(p-1)<p^{t}-1<p^{t}\leq p^{t}j$ for each $j\in\widehat{\Gamma_{(1)}}$, this gives $t-1<\operatorname{log}_{p}\left(\frac{p^t-1}{j}\right)<t$ for $j\in\widehat{\Gamma_{(1)}}$. Hence $\hat{\epsilon}_{j}^{(t)}=t$ for all $j\in\widehat{\Gamma_{(1)}}$. 
          \par
          Take $t=2$. Note that $\widehat{\Gamma_{(2)}}=\{j\in\mathbb{N}:1\leq j\leq p^{2}-1\text{ and }p\nmid j\}$ and $\widehat{\Gamma_{(2)}}\backslash\widehat{\Gamma_{(1)}}=\{j\in\mathbb{N}:p+1\leq j\leq p^{2}-1\text{ and }p\nmid j\}$. Since $\hat{\epsilon}_{j}^{(1)}=1$ for $j\in\widehat{\Gamma_{(1)}}$, from Lemma $\ref{L10}$, we obtain 
          \begin{align*}
              \hat{\epsilon}_{j}^{(t)}=\hat{\epsilon}_{j}^{(2)} & =\begin{cases}
              \hat{\epsilon}_{j}^{(1)}+1=2;\text{ for }j\in\widehat{\Gamma_{(1)}} ,\\
              1;\text{ for }j\in\widehat{\Gamma_{(2)}}\backslash\widehat{\Gamma_{(1)}}.
          \end{cases}
          \end{align*}
         \par Taking $t=3$, similarly, from Lemma $\ref{L10}$, we obtain 
          \begin{align*}
              \hat{\epsilon}_{j}^{(t)}=\hat{\epsilon}_{j}^{(3)} &=\begin{cases}
              \hat{\epsilon}_{j}^{(2)}+1;\text{ for }j\in\widehat{\Gamma_{(2)}}, \\
              1;\text{ for }j\in\widehat{\Gamma_{(3)}}\backslash\widehat{\Gamma_{(2)}}
          \end{cases}
           =\begin{cases}
              3;\text{ for }j\in\widehat{\Gamma_{(1)}}, \\
              2;\text{ for }j\in\widehat{\Gamma_{(2)}}\backslash\widehat{\Gamma_{(1)}}, \\
              1;\text{ for }j\in\widehat{\Gamma_{(3)}}\backslash\widehat{\Gamma_{(2)}}.
          \end{cases}
          \end{align*}
          By continuing this reasoning for all values of $t$, we achieve the desired conclusion.
     \end{proof}
 \end{lemma}    

\begin{lemma}\label{L12}
    Let $m$ and $h$ be two positive integers such that $1\leq h\leq\lceil\frac{m}{2}\rceil$. Suppose $\widehat{D_{t}(3)}$ and its corresponding subset $\widehat{\Gamma_{(t)}}$ are defined as in Eq. $(\ref{Eq10})$, when $p=3$, and $t$ is any fixed positive integer. Then for any $x\in\mathbb{F}_{3^m}$, we have
    \begin{flalign*}
\operatorname{Tr}_{3}^{3^m}\left((x+1)^{3^{h}-1}\right)
&=
\Biggl(
    m
    + \sum_{j\in\widehat{\Gamma_{(1)}}}
        h\,\beta(j)\operatorname{Tr}_{3}^{3^m}\left(x^{j}\right) 
    + \sum_{j\in\widehat{\Gamma_{(2)}}\setminus\widehat{\Gamma_{(1)}}}
        (h-1)\,\beta(j)\operatorname{Tr}_{3}^{3^m}\left(x^{j}\right)
    + \cdots \\
&\quad
    + \sum_{j\in\widehat{\Gamma_{(h-1)}}\setminus\widehat{\Gamma_{(h-2)}}}
        2\,\beta(j)\operatorname{Tr}_{3}^{3^m}\left(x^{j}\right) 
    + \sum_{j\in\widehat{\Gamma_{(h)}}\setminus\widehat{\Gamma_{(h-1)}}}
        \beta(j)\operatorname{Tr}_{3}^{3^m}\left(x^{j}\right)
\Biggr)
\pmod{3},
\end{flalign*}
where the map $\beta(\cdot):\widehat{D_{h}(3)}\rightarrow\{1,2\}$ is defined by 
\begin{align*}
\beta(j)=\begin{cases}
           1,\text{ if the cardinality of }\{0\leq s\leq h-1:j_{s}=1\}\text{ is even,} \\
            2,\text{ otherwise.}
       \end{cases}
\end{align*}       
    \begin{proof}
        Note that $3^{h}-1=2+2\cdot 3+2\cdot 3^{2}+\cdots+2\cdot 3^{h-1}$. For a non-negative integer $j$ satisfying $0\leq j\leq 3^{h}-1$, let the $3$-adic expansion of $j$ be $j=j_{0}+j_{1}3+\cdots+j_{h-1}3^{h-1}$, where $j_{0},j_{1},\cdots,j_{h-1}\in\{0,1,2\}$. Then, from Lucas’ Theorem \cite{P1}, one can verify that
        \begin{align*}
           \binom{3^h-1}{j}= \binom{2+2\cdot 3+2\cdot 3^{2}+\cdots+2\cdot 3^{h-1}}{j_{0}+j_{1}3+\cdots+j_{h-1}3^{h-1}} & \equiv\binom{2}{j_{0}}\binom{2}{j_{1}}\cdots\binom{2}{j_{h-1}}\pmod{3}
            =\beta(j),
        \end{align*}
         where $\beta(j)=\begin{cases}
           1,\text{ if the cardinality of }\{0\leq s\leq h-1:j_{s}=1\}\text{ is even,} \\
            2,\text{ otherwise.}
       \end{cases}$
       \vskip 1pt
      Note that for each $j\in\widehat{\Gamma_{(h)}}$, there exists an unique set $\hat{B}_{j}^{(h)}=\{j,3j,\cdots,3^{\hat{\epsilon}_{j}^{(h)}-1}j\}$. Since $\bigcup_{j\in\widehat{\Gamma_{(h)}}}\hat{B}_{j}^{(h)}=\widehat{D_{h}(3)}$ and $\beta(i)=\beta(j)$ for every $i\in\hat{B}_{j}^{(h)}$, then for any $x\in\mathbb{F}_{3^m}$, we can write
       \begin{align}\label{Eq13}
           \operatorname{Tr}_{3}^{3^m}\left((x+1)^{3^{h}-1}\right) &= \operatorname{Tr}_{3}^{3^m}\left(1+\sum_{j\in\widehat{D_{h}(3)}}\beta(j)x^{j}\right) \nonumber \\
           &= \operatorname{Tr}_{3}^{3^m}\left(1+\sum_{j\in\widehat{\Gamma_{(h)}}}\sum_{i\in\hat{B}_{j}^{(h)}}\beta(j)x^{i}\right) \nonumber \\
           &= \operatorname{Tr}_{3}^{3^m}\left(1\right)+\sum_{j\in\widehat{\Gamma_{(h)}}} \hat{\kappa}_{j}^{(h)}\beta(j)\operatorname{Tr}_{3}^{3^m}\left(x^{j}\right)  \nonumber \\
           &= \operatorname{Tr}_{3}^{3^m}\left(1\right) + \operatorname{Tr}_{3}^{3^m}\left(\sum_{j\in\widehat{\Gamma_{(1)}}}\hat{\kappa}_{j}^{(h)}\beta(j)x^{j}\right)+\sum_{k=1}^{h-1}\operatorname{Tr}_{3}^{3^m}\left(\sum_{j\in\widehat{\Gamma_{(h-k+1)}}\backslash\widehat{\Gamma_{(h-k)}}}\hat{\kappa}_{j}^{(h)}\beta(j)x^{j}\right)
       \end{align}
      According to Lemma $\ref{L11}$, $\hat{\epsilon}_{j}^{(h)}=h$ for $j\in\widehat{\Gamma_{(1)}}$ and $\hat{\epsilon}_{j}^{(h)}=h-(h-k)=k$ for all $j\in\widehat{\Gamma_{(h-k+1)}}\backslash\widehat{\Gamma_{(h-k)}}$. It is clear from the right-hand side of $(\ref{Eq13})$ that for $k=1,2,\cdots,h-1$ and $j\in\widehat{\Gamma_{(h-k+1)}}\backslash\widehat{\Gamma_{(h-k)}}$, $\hat{\kappa}_{j}^{(h)}\beta(j)=0$ if and only if $k$ is divisible by $3$, which completes the proof.
\end{proof}
\end{lemma}
\begin{theorem}\label{Th3}
    Let $s^{\infty}$ be the sequence defined in $(\ref{Eq3})$ through the monomial $F_{5}(x)=x^{3^h-1}$, $1\leq h\leq\lceil\frac{m}{2}\rceil$ over $\mathbb{F}_{3^m}$, and $\widehat{\Gamma_{(t)}}$ be defined as in Eq. $(\ref{Eq10})$, when $p=3$. Then the linear span $\mathcal{L}_{s}$ and the minimal polynomial $\mathfrak{M}_{s}(x)$ of $s^{\infty}$ are given by
    \begin{align*}
        \mathcal{L}_{s} &=\begin{cases}
            \mathbb{N}_{3}(m)+2m;\text{ for }h=1, \\
            \mathbb{N}_{3}(m)+\left(2\times\mathbb{N}_{3}(h)+\sum_{k=1}^{h-1}(3^{h-k}+3^{h-k-1})\times\mathbb{N}_{3}(k)\right) m;\text{ for }h\geq 2
        \end{cases}    
    \end{align*}
    and
    \begin{align*}
        \mathfrak{M}_{s}(x) &= \begin{cases}
            (x-1)^{\mathbb{N}_{3}(m)}\prod_{j\in\widehat{\Gamma_{(1)}}}m_{\alpha^{-j}}(x);\text{ for }h=1, \\
            (x-1)^{\mathbb{N}_{3}(m)}\prod_{\substack{j\in\widehat{\Gamma_{(1)}}\\\mathbb{N}_{3}(h)=1}}m_{\alpha^{-j}}(x)\times \prod_{\substack{k=1,\text{ }\mathbb{N}_{3}(k)=1}}^{h-1}\left(\prod_{j\in\widehat{\Gamma_{(h-k+1)}}\backslash\widehat{\Gamma_{(h-k)}}}m_{\alpha^{-j}}(x)\right);\text{ for }h\geq 2. \\
        \end{cases}
    \end{align*}
    \begin{proof}
        With the aid of Lemma $\ref{L12}$, the sequence $s^{\infty}$ defined in $(\ref{Eq3})$ through the monomial $F_{5}(x)=x^{3^h-1}$, $1\leq h\leq\lceil\frac{m}{2}\rceil$ over $\mathbb{F}_{3^m}$ is given by 
        \begin{align}\label{Eq14}
           s_{t}&= \operatorname{Tr}_{3}^{3^m}\left((\alpha^{t}+1)^{3^{h}-1}\right) \nonumber \\ 
&=
\Biggl(
    m
    + \sum_{j\in\widehat{\Gamma_{(1)}}}
        h\,\beta(j)\operatorname{Tr}_{3}^{3^m}\left((\alpha^{t})^{j}\right) 
    + \sum_{j\in\widehat{\Gamma_{(2)}}\setminus\widehat{\Gamma_{(1)}}}
        (h-1)\,\beta(j)\operatorname{Tr}_{3}^{3^m}\left((\alpha^{t})^{j}\right)
    + \cdots \nonumber \\
&\quad
    + \sum_{j\in\widehat{\Gamma_{(h-1)}}\setminus\widehat{\Gamma_{(h-2)}}}
        2\,\beta(j)\operatorname{Tr}_{3}^{3^m}\left((\alpha^{t})^{j}\right) 
    + \sum_{j\in\widehat{\Gamma_{(h)}}\setminus\widehat{\Gamma_{(h-1)}}}
        \beta(j)\operatorname{Tr}_{3}^{3^m}\left((\alpha^{t})^{j}\right) 
\Biggr)
\pmod{3}; \text{ for all }t\geq 0.
        \end{align}
        Note that $\widehat{\Gamma_{(h)}}=\widehat{\Gamma_{(1)}}\cup\left(\bigcup_{k=1}^{h-1}\widehat{\Gamma_{(h-k+1)}}\backslash\widehat{\Gamma_{(h-k)}}\right)$. According to the definition of $\hat{\Gamma}$ in Lemma $\ref{L7}$, one can verify that $\hat{\Gamma}=\widehat{\Gamma_{(h)}}$, particularly when $q=p=3$ with $1\leq h\leq\lceil\frac{m}{2}\rceil$. Therefore, it follows from Lemma \ref{L7} that the $3$-cyclotomic cosets $C_{j_1}$ and $C_{j_2}$ are pairwise disjoint for any distinct pair of $j_1$ and $j_2$ in $\widehat{\Gamma_{(h)}}$. From Lemma $\ref{L3}$, we conclude that $|C_{j}|=m$ for every $j\in\widehat{\Gamma_{(h)}}$.
        \par The coefficients $k\beta(j)$ for $j\in\widehat{\Gamma_{(h-k+1)}}\backslash\widehat{\Gamma_{(h-k)}}$, $k=1,2,\cdots,h-1$ under the summation in Eq. $(\ref{Eq14})$ vanish if and only if $3\mid k$. By using the facts $|\widehat{\Gamma_{(h)}}|=2\cdot3^{h-1}$ and $|\widehat{\Gamma_{(h-k+1)}}\backslash\widehat{\Gamma_{(h-k)}}|=2\cdot(3^{h-k}-3^{h-k-1})=3^{h-k}+3^{h-k-1}$ for $h>k$, we conclude the desired conclusions on the linear span $\mathcal{L}_{s}$ and the minimal polynomial $\mathfrak{M}_{s}(x)$ of $s^{\infty}$ from Lemma $\ref{L6}$ and Eq. $(\ref{Eq14})$.
    \end{proof}
\end{theorem}
By analyzing Lemma $\ref{L12}$ and Theorem $\ref{Th3}$, it can be concluded that the linear span $\mathcal{L}_{s}$ and the minimal polynomial $\mathfrak{M}_{s}(x)$ of $s^{\infty}$, defined in Theorem $\ref{Th3}$, are independent of the image values of $\beta(j)$, defined in Lemma $\ref{L12}$, for any $j\in\widehat{D_{h}(3)}$. Taking advantage of the fact that $\binom{p-1}{j_{0}}\binom{p-1}{j_{1}}\cdots\binom{p-1}{j_{h-1}}\not\equiv 0\pmod{p}$ for any $j_{0},j_{1},\cdots,j_{h-1}\in\{0,1,\cdots,p-1\}$, we generalize Theorem $\ref{Th3}$ to an arbitrary odd prime $p$ in the following Theorem. %under the map $\beta(\cdot)$ defined in Lemma $\ref{L12}$.

\hspace{0.1mm}
\begin{theorem}\label{Th4} 
    Let $s^{\infty}$ be the sequence defined in $(\ref{Eq3})$ through the monomial $F_{5}(x)=x^{p^h-1}$, $1\leq h\leq\lceil\frac{m}{2}\rceil$ over $\mathbb{F}_{p^m}$, and $\widehat{\Gamma_{(t)}}$ be defined as in Eq. $(\ref{Eq10})$. Then the linear span $\mathcal{L}_{s}$ and the minimal polynomial $\mathfrak{M}_{s}(x)$ of $s^{\infty}$ are given by 
     \begin{align*}
        \mathcal{L}_{s} &=\begin{cases}
            \mathbb{N}_{p}(m)+(p-1)m;\text{ for }h=1, \\
            \mathbb{N}_{p}(m)+\left((p-1)\times\mathbb{N}_{p}(h)+\sum_{k=1}^{h-1}(p-1)(p^{h-k}-p^{h-k-1})\times\mathbb{N}_{p}(k)\right) m;\text{ for }h\geq 2
        \end{cases}    
    \end{align*}
    and
    \begin{align*}
        \mathfrak{M}_{s}(x) &= \begin{cases}
            (x-1)^{\mathbb{N}_{p}(m)}\prod_{j\in\widehat{\Gamma_{(1)}}}m_{\alpha^{-j}}(x);\text{ for }h=1, \\
            (x-1)^{\mathbb{N}_{p}(m)}\prod_{\substack{j\in\widehat{\Gamma_{(1)}}\\\mathbb{N}_{p}(h)=1}}m_{\alpha^{-j}}(x)\times \prod_{\substack{k=1,\text{ }\mathbb{N}_{p}(k)=1}}^{h-1}\left(\prod_{j\in\widehat{\Gamma_{(h-k+1)}}\backslash\widehat{\Gamma_{(h-k)}}}m_{\alpha^{-j}}(x)\right);\text{ for }h\geq 2. \\
        \end{cases}
    \end{align*}
    \begin{proof}
        The proof of this Theorem is similar to that of Theorem $\ref{Th3}$. 
    \end{proof}
\end{theorem}

\begin{theorem}\label{Th8}
    Let the code $\mathcal{C}_{F_{5}}$ be defined by the sequence $s^{\infty}$ through the monomial $F_{5}(x)$ over $\mathbb{F}_{p^m}$ of Theorem $\ref{Th4}$. Then $\mathcal{C}_{F_{5}}$ has parameters $[p^m-1,p^m-1-\mathcal{L}_{s},d]$ over $\mathbb{F}_{p}$ with the generator polynomial $\mathfrak{M}_{s}(x)$, where $\mathcal{L}_{s}$ and $\mathfrak{M}_{s}(x)$ are given in Theorem $\ref{Th4}$. In addition,
    \begin{flalign*}
   \hspace{15pt}    d &\geq \begin{cases}
            p^{h}+2;\text{ if }p\nmid m\text{ and }h<p, \\
            p^{h}+1;\text{ if }p\mid m\text{ and }h<p, \\
            p^{h-1};\text{ if }h=p, \\
          %  ;\text{ if }p\mid m\text{ and }h=p, \\
           p^{p-1}+1 ;\text{ if }h>p\text{ and }p\nmid h, \\
            p^{p-1};\text{ if }h>p\text{ and }p\mid h,
        \end{cases} &
    \end{flalign*}
    \begin{proof}
        The dimension of the code $\mathcal{C}_{F_{5}}$ follows from the linear span $\mathcal{L}_{s}$ determined in Theorem $\ref{Th4}$. We now determine the lower bounds on the minimum weight of the code $\mathcal{C}_{F_{5}}$.
        \par When $p\nmid m$ and $h<p$, one can note that $\mathbb{N}_{p}(m)=1$ and $\mathbb{N}_{p}(t)=1$ for all $t=1,2,\cdots,h$. Then it can be verified that $\alpha^{-j}$ for all $j\in \{0\}\cup\widehat{D_{h}(p)}\cup \{p^h\}$ are the zeros of the generator polynomial $\mathfrak{M}_{s}(x)$ of $\mathcal{C}_{F_{5}}$ in Theorem $\ref{Th4}$. Hence, the BCH bound gives $d\geq p^h+2$. In the case of $p\mid m$ and $h<p$, the lower bounds on the minimum weight $d$ can be achieved in a similar manner.
 \par When $h=p$, it can be noted that $\mathbb{N}_{p}(h)=0$ and $\mathbb{N}_{p}(t)=1$ for all $t=1,2,\cdots,h-1$. Then, one can verify that $\alpha^{-j}$ for all $j\in\{p^{h-1}+1,p^{h-1}+2,\cdots,2p^{h-1}-1\}$ are the zeros of the generator polynomial $\mathfrak{M}_{s}(x)$ of $\mathcal{C}_{F_{5}}$ in Theorem $\ref{Th4}$. Hence, the BCH bound gives $d\geq p^{h-1}$.    
 \par When $h>p$ and $p\nmid h$, we have $\mathbb{N}_{p}(h)=1$. Then, it can be verified that $\alpha^{-j}$ for all $j\in\widehat{\Gamma_{(1)}}\cup\left(\bigcup_{t=1}^{p-1}\widehat{\Gamma_{(h-t+1)}}\backslash\widehat{\Gamma_{(h-t)}}\right)$ are the zeros of the generator polynomial $\mathfrak{M}_{s}(x)$ of $\mathcal{C}_{F_{5}}$ in Theorem $\ref{Th4}$. Taking advantage of the facts $p^{h-1}+kp^{t-1}\in C_{p^{h-t}+k}$ and $p^{h-t}+k\in\widehat{\Gamma_{(h-t+1)}}\backslash\widehat{\Gamma_{(h-t)}}$ for any pair of $t$ and $k$ in $\{1,2,\cdots,p-1\}$, we constructed a set $T$ with $p^{p-1}$ consecutive elements, where $T=\{p^{h-1},p^{h-1}+1,p^{h-1}+2,\cdots,p^{h-1}+p^{p-1}-1\}$. Hence, from the BCH bound, we have $d\geq p^{p-1}+1$. The lower bounds on the minimum distance $d$ for the case $h>p$ and $p\mid h$ can be achieved in a similar manner.    \end{proof}
\end{theorem}
%\vspace{1em}
\begin{corollary}
    Let $p\geq 3$ be a prime. Let $m>1$ and $h$ be two integers such that $1\leq h\leq \lceil\frac{m}{2}\rceil$. Then the $p$-ary cyclic code $\mathcal{C}_{F_{5}}$ defined in Theorem $\ref{Th8}$ has parameters $[p^m-1,p^m-1-(p^h-p^{h-1})m,p^h+1]$, when $1\leq h<p$, $p\mid m$, and $h\mid m$.
    \begin{proof}
       For $p\mid m$, $\mathbb{N}_{p}(m)=0$. Thus, the dimension of the code $\mathcal{C}_{F_{5}}$ follows from the linear span $\mathcal{L}_{s}$ of the sequence $s^{\infty}$ in Theorem $\ref{Th4}$. We now determine the exact minimum distance of the code $\mathcal{C}_{F_{5}}$. 
       \par When $h\mid m$, $\mathbb{F}_{p^h}^{*}$ is a cyclic subgroup of $\mathbb{F}_{p^m}^{*}$. By Theorem $\ref{Th4}$, the generator polynomial of the code $\mathcal{C}_{F_{5}}$ is given by 
        \begin{equation*}
            \mathfrak{M}_{s}(x)=\displaystyle\prod_{i=1}^{p^{h}-1}m_{\alpha^{-i}}(x)=\operatorname{lcm}\left(m_{\alpha^{-1}}(x),m_{\alpha^{-2}}(x),\cdots,m_{\alpha^{-p^{h}}}(x)\right),
        \end{equation*}
         which is equivalent to a narrow-sense primitive BCH code $\mathcal{C}_{(p,p^m-1,p^h+1,1)}$ over $\mathbb{F}_{p}$ with designed distance $p^h+1$. Theorem $\ref{Th8}$ demonstrates that the minimum distance $d$ of the code $\mathcal{C}_{F_{5}}$ satisfies $d\geq p^{h}+1$. To show that $d=p^{h}+1$, we need to construct a codeword in $\mathcal{C}_{F_{5}}$ of Hamming weight equal to $p^{h}+1$.
         \par Let $\beta=\alpha^{-1}$. Consider a codeword of the form $c(x)=1+x^{t_{1}}+x^{t_{2}}+\cdots+x^{t_{p^{h}}}\in\mathbb{F}_{p}[x]$, where $1\leq t_{i}\leq p^{m}-2$ for $1\leq i\leq p^{h}$. For $c(x)$ to be valid, it must satisfy $c(\beta^{s})=0$ for all $s\in\{1,2,\cdots,p^{h}-1\}$. If $\gamma$ is a generator of $\mathbb{F}_{p^h}^{*}$. Then we can choose $t_{i}$'s such that $\beta^{t_{i}}=(\beta+\gamma^{i})^{p^{h}-1}$ for $i=1,2,\cdots,p^h-1$ and $t_{p^h}=p^h-1$. Therefore, 
         \begin{equation*}
             c(\beta^s)=1+\beta^{st_{1}}+\beta^{st_{2}}+\cdots+\beta^{st_{p^h}}, \text{ for all }s=1,2,\cdots,p^h-1.
         \end{equation*}
         By substituting the values of $\beta^{t_{i}}$ for $1\leq i\leq p^h$ in the above equation, we obtain
         \begin{align*}
             c(\beta^s) &= 1+ \sum_{i=1}^{p^h-1}(\beta+\gamma^{i})^{s(p^h-1)}+\beta^{s(p^h-1)} \\
             &= 1+ \sum_{i=1}^{p^h-1}\sum_{j=0}^{s(p^h-1)}\binom{s(p^h-1)}{j}\beta^{s(p^h-1)-j}\left(\gamma^i\right)^{j} +\beta^{s(p^h-1)} \\
             &= 1 + \beta^{s(p^h-1)} + \sum_{j=0}^{s(p^h-1)}\binom{s(p^h-1)}{j}\beta^{s(p^h-1)-j}\sum_{i=1}^{p^h-1}\left(\gamma^i\right)^{j}.
         \end{align*}
         Since $\gamma$ generates the multiplicative group $\mathbb{F}_{p^h}^{*}$, we know that $\sum_{i=1}^{p^h-1}\left(\gamma^i\right)^{j}=\begin{cases}
             0, \text{ if }j\nmid(p^h-1), \\
             p^{h}-1, \text{ otherwise.}
         \end{cases}$
         \par Therefore, 
         \begin{align*}
             c(\beta^s) &= 1 + \beta^{s(p^h-1)} + (p^h-1)\sum_{k=0}^{s}\binom{s(p^h-1)}{k(p^h-1)}\beta^{(s-k)(p^h-1)} \\
             &= 1 + \beta^{s(p^h-1)} + (p^h-1)\left(\binom{s(p^h-1)}{0}\beta^{s(p^h-1)}+\binom{s(p^h-1)}{s(p^h-1)}\beta^{0}\right) \\
             &= 1 + \beta^{s(p^h-1)} + (p^h-1)(\beta^{s(p^h-1)}+1) \\
             &= 0.
         \end{align*}
         The above equality holds, since from Lemma $\ref{L21}$, we have $\binom{s(p^h-1)}{k(p^h-1)}=0$ in $\mathbb{F}_{p^h}$ for all $k=1,2,\cdots,s-1$.
         \par To ensure that $c(x)$ is of Hamming weight $p^h+1$, it is sufficient to show that $t_{1},t_{2},\cdots,t_{p^h}$ are all distinct. 
         \par If $\beta^{t_{i}}=\beta^{t_{p^h}}$ for some $1\leq i\leq p^h-1$, then $(\beta+\gamma^i)^{p^h-1}=\beta^{p^h-1}$, which gives $\frac{\beta+\gamma^i}{\beta}=\lambda$ for some $\lambda\in\mathbb{F}_{p^h}^{*}$. Note that $\lambda$ cannot be $1$, otherwise $\gamma^i=0$, which is not possible. For $\lambda\neq 1$, $\beta=\frac{\gamma^i}{\lambda-1}\in\mathbb{F}_{p^h}^{*}$, which contradicts the fact that $\beta$ generates the multiplicative group $\mathbb{F}_{p^m}^{*}$.  %So, there exists $\mu\in$ This contradicts the fact that $\beta$ generates the multiplicative group $\mathbb{F}_{p^m}^{*}$.
         \par If $\beta^{t_{i}}=\beta^{t_{j}}$ for some $1\leq i\neq j\leq p^h-1$, then $(\beta+\gamma^i)^{p^h-1}=(\beta+\gamma^j)^{p^h-1}$. Similarly, it also leads to a contradiction. 
         \vskip 1pt This completes the proof.
    \end{proof}
\end{corollary}
\begin{example}
    Let $p=3$, $m=2$, and $h=1$; then $\mathbb{N}_{p}(m)=1$. Let $\alpha$ be a root of the primitive polynomial $x^2 + 2x + 2$ over $\mathbb{F}_{3}$. The generator polynomial of the cyclic code $\mathcal{C}_{F_{5}}$ is $\mathfrak{M}_{s}(x)=x^5 + 2x^3 + x^2 + x + 1$. Then $\mathcal{C}_{F_{5}}$ is an optimal ternary $[8, 3, 5]$ cyclic code, and its dual $\mathcal{C}_{F_{5}}^{\perp}$ is an optimal $[8, 5, 3]$ ternary cyclic code. Both the optimal ternary linear codes with parameters $[8, 3, 5]$ and $[8, 5, 3]$ in the Database $\textnormal{\cite{Database}}$ are not cyclic.
\end{example}
\vspace{1em}
\begin{example}
  Let $p=5$, $m=2$, and $h=1$; then $\mathbb{N}_{p}(m)=1$. Let $\alpha$ be a root of the primitive polynomial $x^2 + 4x + 2$ over $\mathbb{F}_{5}$. The generator polynomial of the cyclic code $\mathcal{C}_{F_{5}}$ is $\mathfrak{M}_{s}(x)=x^9 + 2x^8 + 2x^7 + 3x^6 + 2x^5 + x^4 + 4x^3 + 4x^2 + 1$. Then $\mathcal{C}_{F_{5}}$ is an optimal quinary $[24, 15, 7]$ cyclic code, and its dual $\mathcal{C}_{F_{5}}^{\perp}$ is an almost optimal $[24, 9, 11]$ quinary cyclic code according to the Database $\textnormal{\cite{Database}}$. % The optimal ternary linear codes with parameters $[8, 3, 5]$ and $[8, 5, 3]$ in the Database $\textnormal{\cite{Database}}$ are not cyclic. 
\end{example}
\vspace{1em}
\begin{example}
    Let $p=m=3$ and $h=2$; then $\mathbb{N}_{p}(m)=0$ and $\mathbb{N}_{p}(h)=1$. Let $\alpha$ be a root of the primitive polynomial $x^3 + 2x + 1$ over $\mathbb{F}_{3}$. The generator polynomial of the cyclic code $\mathcal{C}_{F_{5}}$ is $\mathfrak{M}_{s}(x)=x^{18} + 2x^{17} +  2x^{14} + x^{13} + x^{12} + x^{11} + 2x^9 + x^7 + 2x^6 + x^5 + x^4 + 2x + 2$. Then $\mathcal{C}_{F_{5}}$ is an optimal ternary $[26, 8, 13]$ cyclic code, and its dual $\mathcal{C}_{F_{5}}^{\perp}$ is an optimal $[26, 18, 6]$ ternary cyclic code. Both the optimal ternary linear codes with parameters $[26, 8, 13]$ and $[26, 18, 6]$ in the Database $\textnormal{\cite{Database}}$ are not cyclic.
\end{example}
\vspace{1em}
\begin{example}
    Let $p=5$, $m=3$, and $h=1$; then $\mathbb{N}_{p}(m)=\mathbb{N}_{p}(h)=1$. Let $\alpha$ be a root of the primitive polynomial $x^3 + 3x + 3$ over $\mathbb{F}_{5}$. The generator polynomial of the cyclic code $\mathcal{C}_{F_{5}}$ is $\mathfrak{M}_{s}(x)=x^{13} + 2x^{12} +  4x^{11} + 2x^{10} + 4x^9 + x^8 + 4x^7 + 2x^5 + x^3 + 2x^2 + x + 1$. Then $\mathcal{C}_{F_{5}}$ is a quinary $[124, 111, 7]$ cyclic code, and its dual $\mathcal{C}_{F_{5}}^{\perp}$ is a quinary $[124, 13, 82]$ cyclic code. Both the quinary cyclic codes $\mathcal{C}_{F_{5}}$ and $\mathcal{C}_{F_{5}}^{\perp}$ are optimal according to the Database $\textnormal{\cite{Database}}$.
\end{example}
\vspace{1em}
\begin{example}
    Let $p=7$, $m=2$, and $h=1$; then $\mathbb{N}_{p}(m)=\mathbb{N}_{p}(h)=1$. Let $\alpha$ be a root of the primitive polynomial $x^2 + 6x + 3$ over $\mathbb{F}_{7}$. The generator polynomial of the cyclic code $\mathcal{C}_{F_{5}}$ is $\mathfrak{M}_{s}(x)=x^{13} + x^{12} + x^{11} + 2x^{10} + x^9 + x^7 + 5x^6 + 6x^5 + x^4 + 3x^3 + 5x + 1$. Then $\mathcal{C}_{F_{5}}$ is an optimal $[48, 35, 9]$ cyclic code over $\mathbb{F}_{7}$, and its dual $\mathcal{C}_{F_{5}}^{\perp}$ is a $[48, 13, 23]$ cyclic code over $\mathbb{F}_{7}$. The optimal $7$-ary linear code with parameter $[48,35,9]$ in the Database $\textnormal{\cite{Database}}$ is not cyclic.
\end{example}
\vspace{1em}
\begin{example}
    Let $p=5$, $m=3$, and $h=2$; then $\mathbb{N}_{p}(m)=\mathbb{N}_{p}(h)=1$. Let $\alpha$ be a root of the primitive polynomial $x^3 + 3x + 3$ over $\mathbb{F}_{5}$. The generator polynomial of the cyclic code $\mathcal{C}_{F_{5}}$ is $\mathfrak{M}_{s}(x)=x^{61} + 2x^{60} +
    x^{58} + 2x^{54} + 3x^{53} + 2x^{52} + x^{51} + x^{50} + 3x^{49} + x^{48} + 3x^{47} +
    4x^{46} + 3x^{45} + 3x^{44} + 4x^{43} + x^{42} + 4x^{41} + 4x^{39} + 3x^{38} + x^{36} +
    4x^{35} + 3x^{34} + 2x^{32} + 3x^{31} + 4x^{30} + 3x^{29} + 2x^{28} + 3x^{27} +
    3x^{25} + 3x^{24} + x^{23} + 4x^{22} + 4x^{21} + 4x^{20} + 4x^{19} + x^{18} + 3x^{17} +
    2x^{15} + x^{14} + 3x^{13} + 2x^{12} + x^{11} + x^{10} + 3x^9 + x^8 + 2x^7 + 3x^6
    + 4x^5 + 2x^4 + x^2 + 1$. Then $\mathcal{C}_{F_{5}}$ is an optimal quinary $[124, 63, 32]$ cyclic code, and its dual $\mathcal{C}_{F_{5}}^{\perp}$ is a quinary $[124, 61, d^{\perp}]$ cyclic code, where $15\leq d^{\perp}\leq 33$. The optimal quinary linear code with parameter $[124, 63, 32]$ in the Database $\textnormal{\cite{Database}}$ is not cyclic.
\end{example}
\vspace{1em}
\begin{remark}
    In $\textnormal{\cite[Section $5$]{P2}}$, Ding and Zhou employed the monomial $F_{5}(x)=x^{2^h-1}$ over $\mathbb{F}_{2^m}$ with $1\leq h\leq\lceil\frac{m}{2}\rceil$ and studied the cyclic code $\mathcal{C}_{F_{5}}$ to determine its generator polynomial and the lower bounds of its minimum weight. In this subsection, we have investigated a more general scenario by utilizing the monomial $F_{5}(x)=x^{p^h-1}$ over $\mathbb{F}_{p^m}$ for $p\geq 3$ and $1\leq h\leq\lceil\frac{m}{2}\rceil$, and studied the cyclic code $\mathcal{C}_{F_{5}}$. It is noteworthy that Theorem $\ref{Th8}$ demonstrates that the lower bounds on the minimum distance $d$ of the cyclic code $\mathcal{C}_{F_{5}}$ in some cases are better than the square root of its code length.  
\end{remark}

\subsection{Cyclic codes from $x^{2\cdot p^{m/2}-1}$ over $\mathbb{F}_{p^m}$, where $p\geq 3$ and $m$ is even}\label{sec8}
In this subsection, we construct $p$-ary cyclic codes $\mathcal{C}_{F_{6}}$ defined by the sequence $s^{\infty}$ of $(\ref{Eq3})$ through the monomial $F_{6}(x)=x^{2\cdot p^{m/2}-1}$ over $\mathbb{F}_{p^m}$, where $p$ is an odd prime and $m$ is an even integer. The differential uniformity $\Delta_{F_{6}}$ of  the power function $F_{6}$ over $\mathbb{F}_{p^m}$ is $p^{m/2}$ \cite{Uniform6}.    %Throughout this section, we assume $h=\frac{m}{2}$.

\par Recall that for any given positive integer $t$ and an integer $j$ satisfying $0\leq j\leq p^{t}-1$, if the $p$-adic expansion of $j$ is $j=j_{0}+j_{1}p+\cdots+j_{t-1}p^{t-1}$, where $j_{0},j_{1},\cdots,j_{t-1}\in\{0,1,\cdots,p-1\}$, then $j$ can be uniquely expressed as a vector $(j_{0},j_{1},\cdots,j_{t-1})$. Define a map $\beta(\cdot):\{0,1,2,\cdots,p^{t}-1\}\rightarrow\{1,2,\cdots,p-1\}$ as follows
\begin{equation}\label{Eq16}
    \beta(j) = \prod_{s=0}^{t-1}\binom{p-1}{j_s} \pmod{p}
\end{equation}
Below, we mention some important observations that will be useful in expanding the sequence $s^{\infty}$. 
\vspace{1mm}
\begin{enumerate}
    \item $\operatorname{Im}(\beta)$ never takes the value zero due to the fact that $p\nmid\binom{p-1}{k}$ for any integer $k\in\{0,1,\cdots,p-1\}$.
    \item From Fermat’s Little Theorem, we know that $\beta(j)^{p}\equiv\beta(j)\pmod{p}$ for any $j\in\{0,1,2,\cdots,p^{t}-1\}$.
    \item Let $\widehat{D_{t}(p)}=\{1,2,\cdots,p^{t}-1\}$ and $\widehat{\Gamma_{(t)}}=\{j\in\widehat{D_{t}(p)}:p\nmid j\}$, where $t$ is any fixed positive integer, then $\widehat{D_{t}(p)}=\widehat{\Gamma_{(t)}}\cup p\widehat{\Gamma_{(t-1)}}\cup\cdots\cup p^{t-1}\widehat{\Gamma_{(1)}}$ and $\widehat{\Gamma_{(t)}}\setminus\widehat{\Gamma_{(t-1)}}=\bigcup_{k=1}^{p-1}\{i+kp^{t-1}:i\in\widehat{\Gamma_{(t-1)}}\}$. 
\end{enumerate}
\vspace{1mm}
We need the following lemmas before presenting the main results of this section.
\vspace{1mm}
\begin{lemma}\label{L14}
    Let $m$ and $t$ be two positive integers. Suppose $\widehat{D_{t}(p)}$ and $\widehat{\Gamma_{(t)}}$ are defined as in Eq. $(\ref{Eq10})$. Then for any $x\in\mathbb{F}_{p^m}$, we have
    \begin{align*}
        \operatorname{Tr}_{p}^{p^m}\left((x+1)^{p^{t}-1}\right) &= \operatorname{Tr}_{p}^{p^m}\left(1+\sum_{j\in \widehat{D_{t}(p)}}\beta(j)x^{j}\right) \\
        &= \Biggl(
    m
    + \sum_{j\in\widehat{\Gamma_{(1)}}}
        t\,\beta(j)\operatorname{Tr}_{p}^{p^m}\left(x^{j}\right) 
    + \sum_{j\in\widehat{\Gamma_{(2)}}\setminus\widehat{\Gamma_{(1)}}}
        (t-1)\,\beta(j)\operatorname{Tr}_{p}^{p^m}\left(x^{j}\right)
    + \cdots \\
&\quad
    + \sum_{j\in\widehat{\Gamma_{(t-1)}}\setminus\widehat{\Gamma_{(t-2)}}}
        2\,\beta(j)\operatorname{Tr}_{p}^{p^m}\left(x^{j}\right) 
    + \sum_{j\in\widehat{\Gamma_{(t)}}\setminus\widehat{\Gamma_{(t-1)}}}
        \beta(j)\operatorname{Tr}_{p}^{p^m}\left(x^{j}\right)
\Biggr)
\pmod{p},
    \end{align*}
    where $\beta(j)$ is defined as in Eq. $(\ref{Eq16})$.
    \begin{proof}
        The proof of Lemma $\ref{L12}$ can be easily generalized into a proof of this Lemma. 
    \end{proof}
\end{lemma}

\begin{lemma}\label{L15}
    Let $m\geq 2$ be an even integer and $h=\frac{m}{2}$. Suppose $\widehat{D_{t}(p)}$ and its corresponding subset $\widehat{\Gamma_{(t)}}$ are defined as in Eq. $(\ref{Eq10})$, where $t$ is any fixed positive integer. Then, for any $x\in\mathbb{F}_{p^m}$, we have
    \begin{flalign*}
      \hspace{15pt}  \operatorname{Tr}_{p}^{p^m}\left((x+1)^{2\cdot p^{h}-1}\right) &= \Biggl(m + (hp-h+1)\operatorname{Tr}_{p}^{p^m}\left(x\right)+\sum_{j\in\widehat{\Gamma_{(1)}}\setminus\{1\}}h\beta(j)\operatorname{Tr}_{p}^{p^m}\left(x^{j}\right)+ \nonumber \\
         &\quad \sum_{j\in\widehat{\Gamma_{(1)}}}\left(h\beta(j)\operatorname{Tr}_{p}^{p^m}\left(x^{j+p}\right)+\sum_{k=2}^{p-1}(h-1)\beta(j)\operatorname{Tr}_{p}^{p^m}\left(x^{j+kp}\right)\right)+\cdots \nonumber \\
         &\quad + \sum_{j\in\widehat{\Gamma_{(h-2)}}}\left(3\beta(j)\operatorname{Tr}_{p}^{p^m}\left(x^{j+p^{h-2}}\right)+\sum_{k=2}^{p-1}2\beta(j)\operatorname{Tr}_{p}^{p^m}\left(x^{j+kp^{h-2}}\right)\right)+ \nonumber \\ &\quad \sum_{j\in\widehat{\Gamma_{(h-1)}}}\left(2\beta(j)\operatorname{Tr}_{p}^{p^m}\left(x^{j+p^{h-1}}\right)+\sum_{k=2}^{p-1}\beta(j)\operatorname{Tr}_{p}^{p^m}\left(x^{j+kp^{h-1}}\right)\right) + \nonumber \\ &\quad \sum_{j\in\widehat{\Gamma_{(h)}}}\beta(j)\operatorname{Tr}_{p}^{p^m}\left(x^{j+p^{h}}\right)\Biggr)\textnormal{ (mod $p$),} &
    \end{flalign*}
    where the map $\beta(\cdot)$ is defined as in Eq. $(\ref{Eq16})$.
    \begin{proof}
    Proceeding as in Lemma $\ref{L12}$, we can write 
        \begin{flalign}\label{Eq17}
         \hspace{15pt}   \operatorname{Tr}_{p}^{p^m}\left((x+1)^{2\cdot p^{h}-1}\right) &= \operatorname{Tr}_{p}^{p^m}\left((x^{p^h}+1)(x+1)^{ p^{h}-1}\right) \nonumber \\
            &= \operatorname{Tr}_{p}^{p^m}\left[(x^{p^h}+1)\left(1+\sum_{j\in\widehat{D_{h}(p)}}\beta(j)x^{j}\right)\right]  \nonumber \\
            &= \operatorname{Tr}_{p}^{p^m}\left(1+\sum_{j\in\widehat{D_{h}(p)}}\beta(j)x^{j}\right)+\operatorname{Tr}_{p}^{p^m}\left(x+\sum_{j\in\widehat{D_{h}(p)}}\beta(j)x^{j+p^{h}}\right)    &          
        \end{flalign}
        Note that $\operatorname{Tr}_{p}^{p^m}(x^{p})=\operatorname{Tr}_{p}^{p^m}(x)$ for any $x\in\mathbb{F}_{p^m}$, and by definition, we know that $\widehat{D_{h}(p)}=\widehat{\Gamma_{(h)}}\cup p\widehat{\Gamma_{(h-1)}}\cup\cdots\cup p^{h-1}\widehat{\Gamma_{(1)}}$. Combining these two facts, we have
        \begin{flalign}\label{Eq18}
          \operatorname{Tr}_{p}^{p^m}\left(x+\sum_{j\in\widehat{D_{h}(p)}}\beta(j)x^{j+p^{h}}\right) &= \operatorname{Tr}_{p}^{p^m}\Biggl(x+\sum_{j\in\widehat{\Gamma_{(h)}}}\beta(j)x^{j+p^{h}}+\sum_{j\in\widehat{\Gamma_{(h-1)}}}\beta(j)x^{j+p^{h-1}}+ \cdots \nonumber \\
&\quad + \sum_{j\in\widehat{\Gamma_{(2)}}}\beta(j)x^{j+p^{2}} +\sum_{j\in\widehat{\Gamma_{(1)}}}\beta(j)x^{j+p}\Biggr) &
        \end{flalign}
        For a positive integer $t=2,3,\cdots,h$, using the fact $\widehat{\Gamma_{(t)}}\setminus\widehat{\Gamma_{(t-1)}}=\bigcup_{k=1}^{p-1}\{j+kp^{t-1}:j\in\widehat{\Gamma_{(t-1)}}\}$, and with the help of Lemma $\ref{L14}$, Eq. $(\ref{Eq17})$ and $(\ref{Eq18})$, we obtain
        \begin{flalign}\label{Eq23}
            \hspace{15pt}   \operatorname{Tr}_{p}^{p^m}\left((x+1)^{2\cdot p^{h}-1}\right) &= \Biggl(
    m
    + \sum_{j\in\widehat{\Gamma_{(1)}}}
        h\,\beta(j)\operatorname{Tr}_{p}^{p^m}\left(x^{j}\right) 
    + \sum_{j\in\widehat{\Gamma_{(2)}}\setminus\widehat{\Gamma_{(1)}}}
        (h-1)\,\beta(j)\operatorname{Tr}_{p}^{p^m}\left(x^{j}\right)
    + \cdots \nonumber \\
&\quad
    + \sum_{j\in\widehat{\Gamma_{(h-1)}}\setminus\widehat{\Gamma_{(h-2)}}}
        2\,\beta(j)\operatorname{Tr}_{p}^{p^m}\left(x^{j}\right) 
    + \sum_{j\in\widehat{\Gamma_{(h)}}\setminus\widehat{\Gamma_{(h-1)}}}
        \beta(j)\operatorname{Tr}_{p}^{p^m}\left(x^{j}\right)
\Biggr)\text{ (mod $p$)} \nonumber \\ + &\quad \Biggl(\operatorname{Tr}_{p}^{p^m}\left(x\right)+\sum_{j\in\widehat{\Gamma_{(h)}}}\beta(j)\operatorname{Tr}_{p}^{p^m}\left(x^{j+p^{h}}\right)+\sum_{j\in\widehat{\Gamma_{(h-1)}}}\beta(j)\operatorname{Tr}_{p}^{p^m}\left(x^{j+p^{h-1}}\right)+\cdots \nonumber \\ &\quad +\sum_{j\in\widehat{\Gamma_{(2)}}}\beta(j)\operatorname{Tr}_{p}^{p^m}\left(x^{j+p^{2}}\right)+\sum_{j\in\widehat{\Gamma_{(1)}}}\beta(j)\operatorname{Tr}_{p}^{p^m}\left(x^{j+p}\right) \Biggr)  
 %        &= \Biggl(m + (h\beta(1)+1)\operatorname{Tr}_{p}^{p^m}\left(x\right)+\sum_{j\in\widehat{\Gamma_{(1)}}\setminus\{1\}}h\beta(j)\operatorname{Tr}_{p}^{p^m}\left(x^{j}\right)+ \nonumber \\
  %       &\quad \sum_{j\in\widehat{\Gamma_{(1)}}}\left(h\beta(j)\operatorname{Tr}_{p}^{p^m}\left(x^{j+p}\right)+\sum_{k=2}^{p-1}(h-1)\beta(j)\operatorname{Tr}_{p}^{p^m}\left(x^{j+kp}\right)\right)+\cdots \nonumber \\
 %        &\quad + \sum_{j\in\widehat{\Gamma_{(h-2)}}}\left(3\beta(j)\operatorname{Tr}_{p}^{p^m}\left(x^{j+p^{h-2}}\right)+\sum_{k=2}^{p-1}2\beta(j)\operatorname{Tr}_{p}^{p^m}\left(x^{j+kp^{h-2}}\right)\right)+ \nonumber \\ &\quad \sum_{j\in\widehat{\Gamma_{(h-1)}}}\left(2\beta(j)\operatorname{Tr}_{p}^{p^m}\left(x^{j+p^{h-1}}\right)+\sum_{k=2}^{p-1}\beta(j)\operatorname{Tr}_{p}^{p^m}\left(x^{j+kp^{h-1}}\right)\right) + \nonumber \\ &\quad \sum_{j\in\widehat{\Gamma_{(h)}}}\beta(j)\operatorname{Tr}_{p}^{p^m}\left(x^{j+p^{h}}\right)\Biggr)\text{ (mod $p$)}
        \end{flalign}
       Note that $\beta(1)=\binom{p-1}{1}=p-1$, so the coefficient of $\operatorname{Tr}_{p}^{p^m}(x)$ on the right-hand side of Eq. $(\ref{Eq23})$ becomes $(hp-h+1)$, and for every $t=2,3,\cdots,h$, we can write
      \[
\begin{aligned}
&\sum_{j\in\widehat{\Gamma_{(t)}}\setminus\widehat{\Gamma_{(t-1)}}}
   (h-t+1)\beta(j)\operatorname{Tr}_{p}^{p^m}\left(x^{j}\right)
 + \sum_{j\in\widehat{\Gamma_{(t-1)}}}
   \beta(j)\operatorname{Tr}_{p}^{p^{m}}\left(x^{j+p^{t-1}}\right) \\
&\qquad =
\sum_{j\in\widehat{\Gamma_{(t-1)}}}
\left(
   (h-t+2)\beta(j)\operatorname{Tr}_{p}^{p^{m}}\left(x^{j+p^{t-1}}\right)
   + \sum_{k=2}^{p-1}
     (h-t+1)\beta(j)\operatorname{Tr}_{p}^{p^{m}}\left(x^{j+kp^{t-1}}\right)
\right).
\end{aligned}
\]
        Hence, the conclusion follows by adding the terms on the right-hand side of Eq. $(\ref{Eq23})$ under modulo $p$.   
    \end{proof}
\end{lemma}
\begin{lemma}\label{L17}
    Let $m\geq 2$ be an even integer. Then for any $j\in\{1,2,3,\cdots,p^{m/2}-1\}$, we have
     $|C_{j+p^{m/2}}|=m$ for $j>1$ and $|C_{p^{m/2}+1}|=\frac{m}{2}$.   
    \begin{proof}
        If $j\in\{1,2,3,\cdots,p^{m/2}-1\}$ and $p\mid j$. Then there is an integer $j_{1}$, $p\nmid j_{1}$ such that $j=p^{k}j_{1}$, where $k\in\{1,2,\cdots,\frac{m}{2}-1\}$. One can note that $j+p^{m/2}$ and $j_{1}+p^{\frac{m}{2}-k}$ are in the same $p$-cyclotomic coset $C_{j_{1}+p^{\frac{m}{2}-k}}$. Since $j_{1}+p^{\frac{m}{2}-k}<p^{m/2}$, according to Lemma $\ref{L3}$, we get  $|C_{j+p^{m/2}}|=m$. 
        \par If $j\in\{1,2,3,\cdots,p^{m/2}-1\}$ and $p\nmid j$. Note that $(j+p^{m/2})\cdot p^{\ell}<p^{m}-1$ for any $0\leq\ell\leq\frac{m}{2}-1$. That means $|C_{j+p^{m/2}}|\geq\frac{m}{2}$. We know that the size of $C_{j+p^{m/2}}$ is divisible by $m$. But $(j+p^{m/2})\cdot p^{m/2}\text{ (mod $v$)}= j\cdot p^{m/2}+1>j+p^{m/2}$ for $j>1$; hence, we conclude that $C_{j+p^{m/2}}$ must be of size $m$ when $j>1$, and it is obvious that $|C_{p^{m/2}+1}|=\frac{m}{2}$. Hence, the proof.
    \end{proof}
\end{lemma}

\begin{lemma}\label{L16}
    Let $m\geq 2$ be an even integer and $h=\frac{m}{2}$. If $\widehat{\Gamma_{(h)}}=\{j\in\mathbb{N}:1\leq j\leq p^{h}-1\text{ and }p\nmid j\}$, then
    \begin{enumerate}
        \item[$1.$] $j+p^{h}$ is the coset leader of $C_{j+p^{h}}$ for any $j\in\widehat{\Gamma_{(h)}}$.
        \item[$2.$] $C_{j_{1}+p^{h}}\cap C_{j_{2}+p^{h}}=\emptyset$ for any pair of distinct $j_{1}$ and $j_{2}$ in $\widehat{\Gamma_{(h)}}$.
        \item[$3.$] $C_{j_{1}}\cap C_{j_{2}+p^{h}}=\emptyset$ for any $j_{1}$ and $j_{2}$ in $\widehat{\Gamma_{(h)}}$.
    \end{enumerate}
    \begin{proof}
        Let $j\in\widehat{\Gamma_{(h)}}$ with $\operatorname{wt}_{p}(j)=r\geq 1$. Then $j=a_{0}+a_{1}p^{j_{1}}+a_{2}p^{j_{2}}+\cdots+a_{r-1}p^{j_{r-1}}$ for some $a_{0},a_{1},\cdots,a_{r-1}\in\{1,2,\cdots,p-1\}$ and $1\leq j_{1}<j_{2}<\cdots<j_{r-1}\leq h-1$. According to Lemma $\ref{L4}$, the coset leader of $C_{j+p^{h}}$ must not be divisible by $p$, which means that the coset leader of $C_{j+p^{h}}$ must be one of $(j+p^{h})\cdot p^{m-j_{t}}\text{ (mod $v$)}$ for some $t\in\{1,2,\cdots,r-1\}$, or $1+j\cdot p^{h}$, or $j+p^{h}$ itself. However, since $j+p^{h}\leq 1+j\cdot p^{h}$ for any $j\in\widehat{\Gamma_{(h)}}$ and $h+1\leq m-j_{k}\leq m-1$ for each $k\in\{1,2,\cdots,r-1\}$, it is not difficult to conclude that $j+p^{h}$ is the smallest integer in $C_{j+p^{h}}$ that is not divisible by $p$. Hence, the proof of the first statement. The second assertion directly follows from the first statement. Note that every $j\in\widehat{\Gamma_{(h)}}$ is the coset leader of $C_{j}$, by which the third statement similarly follows.
    \end{proof}
\end{lemma}
\begin{theorem}\label{Th5}
    Let $s^{\infty}$ be defined as in Eq. $(\ref{Eq3})$ through the monomial $F_{6}(x)=x^{2\cdot p^{h}-1}$, $h=\frac{m}{2}$ over $\mathbb{F}_{p^m}$, and $\widehat{\Gamma_{(t)}}$ be defined as in Eq. $(\ref{Eq10})$. Then the linear span $\mathcal{L}_{s}$ and the minimal polynomial $\mathfrak{M}_{s}(x)$ of $s^{\infty}$ are given by
    \begin{flalign*}
       \mathcal{L}_{s} &=\begin{cases}
             4p-6; \text{ for }h=1,  \\
           \mathbb{N}_{p}(m)+\Biggl(\mathbb{N}_{p}(hp-h+1) + \mathbb{N}_{p}(h)\times(p-2)+\sum_{t=1}^{h-1}(p-1)p^{t-1}\Biggl(\mathbb{N}_{p}(h+1-t)+  \\  
           \hspace{2.8cm}\mathbb{N}_{p}(h-t)\times(p-2)\Biggr)+\left((p-1)p^{h-1}-1\right)\Biggr)\cdot m +\frac{m}{2};\text{ for }h\geq 2,
       \end{cases}   &
    \end{flalign*}
    and 
  \begin{flalign*}
\mathfrak{M}_{s}(x) &=
\begin{cases}
 (x-1)
 \displaystyle\prod_{j\in\widehat{\Gamma_{(1)}}\setminus\{1\}}
   m_{\alpha^{-j}}(x)
 \displaystyle\prod_{j\in\widehat{\Gamma_{(1)}}}
   m_{\alpha^{-j-p}}(x),
  \text{ for } h=1, \\[6pt]
 (x-1)^{\mathbb{N}_{p}(m)}
 \left(m_{\alpha^{-1}}(x)\right)^{\mathbb{N}_{p}(hp-h+1)}
 \displaystyle\prod_{\substack{j\in\widehat{\Gamma_{(1)}}\setminus\{1\} \\ \mathbb{N}_{p}(h)=1}}
   m_{\alpha^{-j}}(x)\times
 \displaystyle\prod_{t=1}^{h-1} \Biggl(\prod_{j\in\widehat{\Gamma_{(t)}}}\Big(\big(m_{\alpha^{-j-p^{t}}}(x)\big)^{\mathbb{N}_{p}(h+1-t)}\\\hspace{3.7cm} \times\Big(\prod_{k=2}^{p-1}m_{\alpha^{-j-kp^{t}}}(x)\Big)^{\mathbb{N}_{p}(h-t)}\Big)\Biggr)\times\displaystyle\prod_{j\in\widehat{\Gamma_{(h)}}}m_{\alpha^{-j-p^{h}}}(x),
  \text{ for } h\ge 2.
\end{cases}
&&
\end{flalign*}
    \begin{proof}
        With the aid of Lemma $\ref{L15}$, the sequence $s^{\infty}$ defined in Eq. $(\ref{Eq3})$ can be expanded as follows 
        \begin{align}\label{Eq19}
            s_{t}& =\operatorname{Tr}_{p}^{p^m}\left(F_{6}(\alpha^t+1)\right) \nonumber  \\%=\sum_{i=0}^{p^m-2} c_{i}(\alpha^t)^{i};\text{ for all $t\geq 0$,}
            &= \Biggl(m + (hp-h+1)\operatorname{Tr}_{p}^{p^m}\left(\alpha^t\right)+\sum_{j\in\widehat{\Gamma_{(1)}}\setminus\{1\}}h\beta(j)\operatorname{Tr}_{p}^{p^m}\left(\left(\alpha^t\right)^{j}\right)+ \nonumber \\
         &\quad \sum_{j\in\widehat{\Gamma_{(1)}}}\left(h\beta(j)\operatorname{Tr}_{p}^{p^m}\left(\left(\alpha^t\right)^{j+p}\right)+\sum_{k=2}^{p-1}(h-1)\beta(j)\operatorname{Tr}_{p}^{p^m}\left(\left(\alpha^t\right)^{j+kp}\right)\right)+\cdots \nonumber \\
         &\quad + \sum_{j\in\widehat{\Gamma_{(h-2)}}}\left(3\beta(j)\operatorname{Tr}_{p}^{p^m}\left(\left(\alpha^t\right)^{j+p^{h-2}}\right)+\sum_{k=2}^{p-1}2\beta(j)\operatorname{Tr}_{p}^{p^m}\left(\left(\alpha^t\right)^{j+kp^{h-2}}\right)\right)+ \nonumber \\ &\quad \sum_{j\in\widehat{\Gamma_{(h-1)}}}\left(2\beta(j)\operatorname{Tr}_{p}^{p^m}\left(\left(\alpha^t\right)^{j+p^{h-1}}\right)+\sum_{k=2}^{p-1}\beta(j)\operatorname{Tr}_{p}^{p^m}\left(\left(\alpha^t\right)^{j+kp^{h-1}}\right)\right) + \nonumber \\ &\quad \sum_{j\in\widehat{\Gamma_{(h)}}}\beta(j)\operatorname{Tr}_{p}^{p^m}\left(\left(\alpha^t\right)^{j+p^{h}}\right)\Biggr)\textnormal{ (mod $p$);}\text{ for all }t\geq 0.  &
        \end{align}
         From Lemma $\ref{L16}$, it can be verified that there is no overlap between any two terms on the right-hand side of Eq. $(\ref{Eq19})$. By analyzing the sizes of the $p$-cyclotomic cosets from Lemma $\ref{L3}$ and $\ref{L17}$, we determine the exact number of terms in which each trace-term on the right-hand side of Eq. $(\ref{Eq19})$ would break as $\operatorname{Tr}_{p}^{p^m}(\alpha^t)=\frac{m}{|C_{t}|}\sum_{j=0}^{|C_{t}|-1}(\alpha^t)^{p^j}$ for every $0\leq t\leq p^m-2$.  
         \vskip 1pt
         The desired conclusions on the linear span $\mathcal{L}_{s}$ and the minimal polynomial $\mathfrak{M}_{s}(x)$ of $s^{\infty}$ then follow from Lemma \ref{L6} and Eq. $(\ref{Eq19})$.
    \end{proof}
\end{theorem}
\begin{theorem}\label{Th9}
    Let the code $\mathcal{C}_{F_{6}}$ be defined by the sequence $s^{\infty}$ through the monomial $F_{6}(x)=x^{2\cdot p^{h}-1}$, $h=\frac{m}{2}$ over $\mathbb{F}_{p^m}$. Then $\mathcal{C}_{F_{6}}$ has parameters $[p^m-1,p^m-1-\mathcal{L}_{s},d]$ over $\mathbb{F}_{p}$ with the generator polynomial $\mathfrak{M}_{s}(x)$, where $\mathcal{L}_{s}$ and $\mathfrak{M}_{s}(x)$ are given in Theorem $\ref{Th5}$. In addition, 
    \begin{flalign*}
     \hspace{15pt}   d &\geq \begin{cases}
            2p^{h}+2;\text{ if }p\nmid m\text{ and }1<h<p, \\
            2p^{h}+1;\text{ if }p\mid m\text{ and }1<h<p, \\
            p^{h-1}+1;\text{ if }h=p, \\
  %          p;\text{ if }h=1, \\
            p^{p-1}+1;\text{ if }h>p\text{ and }p\nmid(h-1), \\
            p^{p-1};\text{ if }h>p\text{ and }p\mid(h-1). \\
        \end{cases} &
    \end{flalign*}
    \begin{proof}
        The dimension of the code $\mathcal{C}_{F_{6}}$ follows from the linear span $\mathcal{L}_{s}$ determined in Theorem $\ref{Th5}$. We now determine the lower bounds on the minimum weight of the code $\mathcal{C}_{F_{6}}$.
        \par When $p\nmid m$ and $1<h<p$, one can note that $\mathbb{N}_{p}(m)=1$, $\mathbb{N}_{p}(hp-h+1)=\mathbb{N}_{p}(h-1)$, and $\mathbb{N}_{p}(t)=1$ hold for all $t=1,2,\cdots,h$. Then it can be verified that $\alpha^{-j}$ for all $j\in\{0\}\cup \widehat{D_{h}(p)}\cup\{p^h\}\cup\left(p^{h}+\widehat{D_{h}(p)}\right)\cup\{2p^{h}\}$ are the zeros of the generator polynomial $\mathfrak{M}_{s}(x)$ of $\mathcal{C}_{F_{6}}$ in Theorem $\ref{Th5}$. Hence, the BCH bound gives $d\geq 2p^h+2$. In the case of $p\mid m$ and $1<h<p$, the lower bounds on the minimum weight $d$ can be achieved in a similar manner.
        \par When $h=p$, it can be noted that $\mathbb{N}_{p}(h)=0$ and $\mathbb{N}_{p}(t)=1$ for all $t=1,2,\cdots,h-1$. Since $\mathbb{N}_{p}(hp-h+1)=\mathbb{N}_{p}(1)=1$, $m_{\alpha^{-1}}(x)$ is a factor of $\mathfrak{M}_{s}(x)$. Then, one can verify that $\alpha^{-j}$ for all $j\in\{p^{h},p^{h}+1,p^{h}+2,\cdots,p^{h}+p^{h-1}-1\}$ are the zeros of the generator polynomial $\mathfrak{M}_{s}(x)$ of $\mathcal{C}_{F_{6}}$ in Theorem $\ref{Th5}$. Hence, the BCH bound gives $d\geq p^{h-1}+1$. 
        \par When $h>p$ and $p\nmid(h-1)$, one can note that $\mathbb{N}_{p}(hp-h+1)=\mathbb{N}_{p}(h-1)=1$. As $p\geq 3$, $\mathbb{N}_{p}(h+1-t)=0$ occurs only if $t=h+1-p,h+1-2p,\cdots, h+1-\lfloor\frac{h}{p}\rfloor p$, while $t$ varies through $1,2,\cdots,h-1$. It is not difficult to verify that $\alpha^{-j}$ for all $j\in\{p^{h},p^{h}+1,\cdots, p^{h}+p^{p-1}-1\}$ are the zeros of the generator polynomial $\mathfrak{M}_{s}(x)$ of $\mathcal{C}_{F_{6}}$ in Theorem $\ref{Th5}$. Hence, the BCH bound gives $d\geq p^{p-1}+1$. In the case of $h>p$ and $p\mid(h-1)$, the lower bounds on the minimum weight $d$ can
        be achieved in a similar manner.
    \end{proof}
\end{theorem}
\begin{example}
    Let $p=3$ and $m=2$; then $h=\frac{m}{2}=1$. Let $\alpha$ be a root of the primitive polynomial $x^2 + 2x + 2$ over $\mathbb{F}_{3}$. The generator polynomial of the cyclic code $\mathcal{C}_{F_{6}}$ is $\mathfrak{M}_{s}(x)=x^6 + 2x^5 + 2x^4 + 2x^2 + x + 1$. Then $\mathcal{C}_{F_{6}}$ is an optimal $[8, 2, 6]$ ternary cyclic code, and its dual $\mathcal{C}_{F_{6}}^{\perp}$ is an optimal ternary $[8,6,2]$ cyclic code. Both the optimal ternary
    linear codes with parameters $[8, 2, 6]$ and $[8,6,2]$ in the Database $\textnormal{\cite{Database}}$ are not cyclic. 
\end{example}
\vspace{1em}
\begin{example}
    Let $p=3$ and $m=4$; then $h=\frac{m}{2}=2$ and $\mathbb{N}_{p}(m)=\mathbb{N}_{p}(hp-h+1)=1$. Let $\alpha$ be a root of the primitive polynomial $x^4 + 2x^3 + 2$ over $\mathbb{F}_{3}$. The generator polynomial of the cyclic code $\mathcal{C}_{F_{6}}$ is $\mathfrak{M}_{s}(x)=x^{47} + x^{46} +  x^{45} + x^{44} + 2x^{40} + x^{39} + 2x^{36} + x^{35} + 2x^{33} + x^{32} + 2x^{30} + x^{26}  + x^{25} + x^{23} + x^{19} + 2x^{18} + x^{17} + 2x^{16} + x^{15} + x^{14} + 2x^{13} +
    2x^{12} + 2x^{11} + 2x^9 + 2x^7 + 2x^6 + 2x^3 + 2x^2 + 2x + 1$. Then $\mathcal{C}_{F_{6}}$ is a ternary $[80, 33, 21]$ cyclic code. The dual code $\mathcal{C}_{F_{6}}^{\perp}$ is an almost optimal ternary $[80, 47, 13]$ cyclic code according to the Database $\textnormal{\cite{Database}}$. 
\end{example}
\vspace{1em}
\begin{example}
    Let $p=5$ and $m=2$; then $h=\frac{m}{2}=1$. Let $\alpha$ be a root of the primitive polynomial $x^2 + 4x + 2$ over $\mathbb{F}_{5}$. The generator polynomial of the cyclic code $\mathcal{C}_{F_{6}}$ is $\mathfrak{M}_{s}(x)=x^{14} + 4x^{13} +  4x^{11} + x^{10} + 3x^8 + x^7 + 4x^6 + 3x^5 + 2x^4 + 2x^3 + 4x^2 + 2x + 4$. Then $\mathcal{C}_{F_{6}}$ is a $[24, 10, 8]$ quinary cyclic code, and its dual $\mathcal{C}_{F_{6}}^{\perp}$ is a quinary $[24, 14, 6]$ cyclic code. 
\end{example}
\vspace{1em}
\begin{remark}
    In this subsection, we utilize the monomial $F_{6}(x)=x^{2\cdot p^{m/2}-1}$, where $p\geq 3$ over $\mathbb{F}_{p^m}$, to study the $p$-ary cyclic codes $\mathcal{C}_{F_{6}}$. Consequently, Theorem $\ref{Th9}$ demonstrates that the lower bounds on the minimum distance $d$ of the cyclic code $\mathcal{C}_{F_{6}}$ are significantly better than the square root of its code length in certain cases. When $m\geq 6$ and $p\geq 3$, the parameters of $\mathcal{C}_{F_{6}}$ become very large. Due to the huge amount of computation required, using a Magma program to verify the minimum distance of $\mathcal{C}_{F_{6}}$ is quite challenging. % considering the huge amount of computation, it is difficult for us to use a Magma program to verify the minimum distance of $\mathcal{C}_{F}$.
\end{remark}

\subsection{Cyclic codes from $x^{p^{2h}-p^{h}+1}$ over $\mathbb{F}_{p^m}$, where $p\geq 3$ and $h=\frac{m-1}{2}$}
In this subsection, we study the $p$-ary cyclic codes $\mathcal{C}_{F_{7}}$ defined by the sequence $s^{\infty}$ of $(\ref{Eq3})$ through the monomial $F_{7}(x)=x^{p^{2h}-p^{h}+1}$ over $\mathbb{F}_{p^m}$, where $p$ is an odd prime, $m$ is an odd integer, and $h=\frac{m-1}{2}$. The power function $F_{7}$ over $\mathbb{F}_{p^m}$ is $(p+1)$-uniform \cite[Theorem $2$]{Uniform7}.  %Throughout this section, We aim to investigate the cyclic code $\mathcal{C}_{F}$ under the additional restrictions on $h$ as follows: 
%\begin{flalign*}
 % \hspace{15pt}  1\leq h &\leq \begin{cases}
 %       \frac{m-4}{2},\text{ if }m\equiv 0\pmod{4}; \\
 %       \frac{m-1}{2},\text{ if }m\equiv 1\pmod{4}; \\
 %       \frac{m-2}{2},\text{ if }m\equiv 2\pmod{4}; \\
 %       \frac{m-3}{2},\text{ if }m\equiv 3\pmod{4}.
 %   \end{cases} &
%\end{flalign*}
%For any integer $j$, with $0\leq j\leq p^{h}-1$, if the $p$-adic expansion of $j$ is $j=j_{0}+j_{1}p+\cdots+j_{h-1}p^{h-1}$, where $j_{0},j_{1},\cdots,j_{h-1}\in\{0,1,\cdots,p-1\}$, then $j$ can be uniquely expressed as a vector $(j_{0},j_{1},\cdots,j_{h-1})$. Define a map $\beta(\cdot):\{0,1,2,\cdots,p^{h}-1\}\rightarrow\{1,2,\cdots,p-1\}$ as follows
%\begin{equation}
%    \beta(j) = \prod_{s=0}^{h-1}\binom{p-1}{j_s} \pmod{p}
%\end{equation}
%Below, we mention some important observations that will be useful in expanding the sequence $s^{\infty}$. 
%\vspace{1mm}
%\begin{enumerate}
%    \item $\operatorname{Im}(\beta)$ never takes the value zero due to the fact that $p\nmid\binom{p-1}{k}$ for any integer $k\in\{0,1,\cdots,p-1\}$.
%    \item From Fermat’s Little Theorem, we know that $\beta(j)^{p}\equiv\beta(j)\pmod{p}$ for any $j\in\{0,1,2,\cdots,p^{h}-1\}$.
%\end{enumerate}

\vspace{1mm}
 Note that for any $x\in\mathbb{F}_{p^m}$, we have
\begin{flalign}\label{Eq20}
   \hspace{15pt} \operatorname{Tr}_{p}^{p^m}\left(F_{7}(x+1)\right) &= \operatorname{Tr}_{p}^{p^m}\left((x+1)(x+1)^{p^{2h}-p^h}\right) \nonumber \\
   &= \operatorname{Tr}_{p}^{p^m}\left((x+1)(x^{p^h}+1)^{p^{h}-1}\right) \nonumber \\ 
   &= \operatorname{Tr}_{p}^{p^m}\left[(x+1)\left(1+\sum_{j\in \widehat{D_{h}(p)}}\beta(j)(x^{p^h})^{j}\right)\right] \nonumber  \\
   &= \operatorname{Tr}_{p}^{p^m}\left(1+x+\sum_{j\in\widehat{D_{h}(p)}}\beta(j)x^{j}+\sum_{j\in\widehat{D_{h}(p)}}\beta(j)x^{1+j\cdot p^{h}}\right) \nonumber \\
   &= \operatorname{Tr}_{p}^{p^m}\left(1+\sum_{j\in\widehat{D_{h}(p)}}\beta(j)x^{j}\right) + \operatorname{Tr}_{p}^{p^m}\left(x+\sum_{j\in\widehat{D_{h}(p)}}\beta(j)x^{j+p^{m-h}}\right);
\end{flalign} 
where, $\widehat{D_{h}(p)}$ and $\beta(j)$ are defined as in Eq. $(\ref{Eq10})$ and $(\ref{Eq16})$, respectively.

\vskip 1 pt
\vspace{1mm}
We need to prove some important lemmas before presenting the main results of this section.
%Observe that when $m$ is even, $h=\frac{m}{2}$, and $p^{2h}-p^{h}+1=p^{m}-p^{m/2}+1\equiv p^{m/2}\cdot(2p^{m/2}-1)\text{ (mod $n$)}$. This implies $C_{p^{m}-p^{m/2}+1}=C_{2\cdot p^{m/2}-1}$. This case is already discussed in section $\ref{sec8}$, so we mainly focus on the case where $m$ is odd and $h=\frac{m+1}{2}$.

%The sequence $s^{\infty}$ defined in $(\ref{Eq3})$ through the monomial $F(x)=x^{p^{2h}-p^{h}+1}$ over $\mathbb{F}_{p^m}$ is given by 
%\begin{align}
 %   s_{t} &= m\text{ (mod $p$)} + \operatorname{Tr}_{p}^{p^m}\left(\alpha^{t}+ \sum_{j\in\widehat{D_{h}(p)}}\beta(j)\left(\alpha^{t}\right)^{j} + \sum_{j\in\widehat{D_{h}(p)}}\beta(j)\left(\alpha^{t}\right)^{j+p^{m-h}}\right);\text{ for all }t\geq 0.
%\end{align}
%We define the following two sets for convenience:
%\begin{equation*}
%    A=\{1,2,3,\cdots p^{h}-1\} \text{ and } B=p^{m-h}+A=\{i+p^{m-h}:i\in A\}.
%\end{equation*}
\vspace{1mm}
\begin{lemma}\label{L18}
    Let $m\geq 3$ be odd and $h=\frac{m-1}{2}$. Suppose $\widehat{\Gamma_{(t)}}=\{j\in\mathbb{N}:1\leq j\leq p^{t}-1\text{ and }p\nmid j\}$, where $t$ is a fixed positive integer. Then,
    \begin{enumerate}
        \item[$1.$] $|C_{j+p^{h+1}}|=m$ for any $j\in\widehat{\Gamma_{(h)}}$. 
        \item[$2.$] For $j\in\widehat{\Gamma_{(h)}}$, $1+jp^{h}$ is the coset leader of $C_{j+p^{h+1}}$ when $j\in\widehat{\Gamma_{(1)}}$ and $j+p^{h+1}$ is the coset leader of $C_{j+p^{h+1}}$ when $j\in\widehat{\Gamma_{(h)}}\setminus\widehat{\Gamma_{(1)}}$. 
        \item[$3.$] $C_{j_{1}+p^{h+1}}\cap C_{j_{2}+p^{h+1}}=\emptyset$ for any pair of distinct integers $j_{1}$ and $j_{2}$ in $\widehat{\Gamma_{(h)}}$. 
    \end{enumerate}
    \begin{proof}
        Let $j\in\widehat{\Gamma_{(h)}}$, then $(j+p^{h+1})\cdot p^{\ell}<p^m-1$ for any $0\leq \ell\leq h-1$. That means $|C_{j+p^{h+1}}|\geq\frac{m-1}{2}$. Since $\operatorname{gcd}(m,h)=1$ and the size of the $p$-cyclotomic coset $C_{j+p^{h+1}}$ must divide $m$, we conclude that $C_{j+p^{h+1}}$ must be of size $m$. This proves the first statement.
        \par Let $j\in\widehat{\Gamma_{(h)}}$ and $\operatorname{wt}_{p}(j)=r\geq 1$. Then $j+p^{h+1}=a_{0}+a_{1}p^{j_{1}}+\cdots+a_{r-1}p^{j_{r-1}}+p^{h+1}$ for some $a_{0},a_{1},\cdots,a_{r-1}\in\{1,2,\cdots,p-1\}$ and $1\leq j_{1}<j_{2}<\cdots<j_{r-1}\leq h-1$. According to Lemma $\ref{L4}$, the coset leader of  $C_{j+p^{h+1}}$ is not divisible by $p$, which means that the coset leader of  $C_{j+p^{h+1}}$ must be one of  $(j+p^{h+1})\cdot p^{m-j_{t}}\text{ (mod $v$)}$ for some $t\in\{1,2,\cdots,r-1\}$, or $1+j\cdot p^{h}$, or $j+p^{h+1}$ itself. However, since $h+2\leq m-j_{r-1}<\cdots<m-j_{1}\leq m-1$, it is not difficult to check that $j+p^{h+1}<(j+p^{h+1})\cdot p^{m-j_{t}}\text{ (mod $v$)}$ for each $t\in\{1,2,\cdots,r-1\}$. Clearly, for $j\in\widehat{\Gamma_{(h)}}$, $1+j\cdot p^h<j+p^{h+1}$ if $j\in\widehat{\Gamma_{(1)}}=\{1,2,\cdots p-1\}$ and $1+j\cdot p^h\geq j+p^{h+1}$ otherwise. Hence, the second statement is proved.
        \par If possible, let $C_{j_{1}+p^{h+1}}= C_{j_{2}+p^{h+1}}$ for some pair of distinct $j_{1}$ and $j_{2}$ in $\widehat{\Gamma_{(h)}}$. Then they must have a common coset leader. Note that $j_{1}+p^{h+1}= j_{2}+p^{h+1}$ and $1+j_{1}\cdot p^{h}= 1+j_{2}\cdot p^{h}$ both imply $j_{1}=j_{2}$. Without any loss of generality, we assume that $j_{1}\in\widehat{\Gamma_{(1)}}$ and $j_{2}\in\widehat{\Gamma_{(h)}}\setminus\widehat{\Gamma_{(1)}}$, then the coset leader of $C_{j_{1}+p^{h+1}}$ and $C_{j_{2}+p^{h+1}}$ must be equal. But $1+j_{1}\cdot p^h<j_{2}+p^{h+1}$ for all $(j_{1},j_{2})\in\widehat{\Gamma_{(1)}}\times(\widehat{\Gamma_{(h)}}\setminus\widehat{\Gamma_{(1)}})$, which leads to a contradiction. This proves the third statement. 
        %$\qquad$  Note that $j_{1}+p^{h+1}\not\equiv j_{2}+p^{h+1}\text{ (mod $n$)}$ unless $j_{1}=j_{2}$. Suppose $C_{j_{1}+p^{h+1}}\cap C_{j_{2}+p^{h+1}}\neq\emptyset$ for some distinct $j_{1}$ and $j_{2}$ in $\widehat{\Gamma_{(h)}}$. Then $(j_{1}+p^{h+1})\equiv(j_{2}+p^{h+1})\cdot p^{\ell}\text{ (mod $n$)}$ for some $j_{2}\in\widehat{\Gamma_{(h)}}$ and $h\leq\ell\leq m-1$. Therefore, $j_{1}$ must be of the form $a_{0}+\sum_{}+p^{h+1-\ell}$
    \end{proof}
\end{lemma}
%\begin{lemma}
 %   Let $m\geq 3$ be odd and $h=\frac{m+1}{2}$. Then, $|C_{j+p^{m-h}}|=m$ for any $j\in\{1,2,3,\cdots p^{h}-1\}$.
 %   \begin{proof}
  %    If $j\in\{1,2,3,\cdots,p^{\frac{m+1}{2}}-1\}$ and $p\mid j$. Then there is an integer $j_{1}$, $p\nmid j_{1}$ such that $j=p^{k}j_{1}$, where $k\in\{1,2,\cdots,\frac{m-1}{2}\}$. Since $C_{j+p^{\frac{m-1}{2}}}=C_{j_{1}+p^{\frac{m-1}{2}-k}}$ and $j_{1}+p^{\frac{m-1}{2}-k}<p^{\frac{m+1}{2}}$. From Lemma $\ref{L3}$, we get $|C_{j+p^{\frac{m-1}{2}}}|=m$.
  %    \par If $j\in\{1,2,3,\cdots,p^{\frac{m+1}{2}}-1\}$ and $p\nmid j$. Note that $(j+p^{\frac{m-1}{2}})\cdot p^{\ell}<p^m-1$ for any $0\leq \ell\leq\frac{m-1}{2}-1$. That means $|C_{j+p^{\frac{m-1}{2}}}|\geq\frac{m-1}{2}$. Since the size of the $p$-cyclotomic coset $C_{j+p^{\frac{m-1}{2}}}$ is divisible by $m$, we conclude that $C_{j+p^{\frac{m-1}{2}}}$ must be of size $m$. Hence, the proof.
 %   \end{proof}
%\end{lemma}
\begin{lemma}\label{L19}
    Let $m\geq 3$ be odd and $h=\frac{m-1}{2}$. Suppose $\widehat{\Gamma_{(t)}}=\{j\in\mathbb{N}:1\leq j\leq p^{t}-1\text{ and }p\nmid j\}$, where $t$ is a fixed positive integer. Then, for any $i\in\widehat{\Gamma_{(h)}}$ and $j\in\widehat{\Gamma_{(h+1)}}$, we have 
    \begin{align*}
        C_{i+p^{h+1}}\cap C_{j} &= \begin{cases}
            C_{i+p^{h+1}};\text{ if }i\in\widehat{\Gamma_{(1)}}, \\
            \emptyset;\text{ otherwise.}
        \end{cases}
    \end{align*}
    \begin{proof}
        First, we claim that if $C_{i+p^{h+1}}\cap C_{j}\neq\emptyset$ then $\operatorname{wt}_{p}(i)=1$ must hold. Assume that $i\in\widehat{\Gamma_{(h)}}$ is such that $\operatorname{wt}_{p}(i)=r\geq 2$; then $i+p^{h+1}=a_{0}+a_{1}p^{j_{1}}+a_{2}p^{j_{2}}+\cdots+a_{r-1}p^{j_{r-1}}+p^{h+1}$ for some $a_{0},a_{1},\cdots,a_{r-1}\in\{1,2,\cdots,p-1\}$ and $1\leq j_{1}<j_{2}<\cdots<j_{r-1}\leq h-1$. Note that $j<i+p^{h+1}$ holds for every $i\in\widehat{\Gamma_{(h)}}$ and $j\in\widehat{\Gamma_{(h+1)}}$. According to Lemma $\ref{L7}$, $j$ is the coset leader of $C_{j}$ for every $j\in\widehat{\Gamma_{(h+1)}}$. That means $j$ is not divisible by $p$ and is the smallest integer in $C_{j}$. Now, $i+p^{h+1}\in C_{j}$ for some $j\in\widehat{\Gamma_{(h+1)}}$ will occur only if
        \begin{equation*}
            (i+p^{h+1})\cdot p^{m-h-1}\text{ (mod $v$)}=1+a_{0}\cdot p^{\frac{m-1}{2}}+a_{1}\cdot p^{j_{1}+\frac{m-1}{2}}+\cdots+a_{r-1}\cdot p^{j_{r-1}+\frac{m-1}{2}}\in\widehat{\Gamma_{(h+1)}}.
        \end{equation*}
         Since $h+1\leq j_{1}+\frac{m-1}{2}<j_{2}+\frac{m-1}{2}<\cdots<j_{r-1}+\frac{m-1}{2}\leq m-2$ and $\operatorname{wt}_{p}(i)=r\geq 2$, we conclude that $(i+p^{h+1})\cdot p^{m-h-1}\text{ (mod $v$)}(>p^{h+1})$ is not in $\widehat{\Gamma_{(h+1)}}$, and hence $C_{i+p^{h+1}}\cap C_{j}=\emptyset$. Therefore, $C_{i+p^{h+1}}\cap C_{j}\neq\emptyset$ must imply $\operatorname{wt}_{p}(i)=1$. As $i\in\widehat{\Gamma_{(h)}}$, $\operatorname{wt}_{p}(i)=1$ would imply $i\in\widehat{\Gamma_{(1)}}=\{1,2,\cdots,p-1\}$. Hence, the proof.
    \end{proof}
\end{lemma}

\begin{lemma}\label{L20}
    Let $m\geq 3$ be odd and $h=\frac{m-1}{2}$. Suppose $\widehat{D_{t}(p)}$ and its corresponding subset $\widehat{\Gamma_{(t)}}$ are defined as in Eq. $(\ref{Eq10})$, where $t$ is any fixed positive integer. Then, for any $x\in\mathbb{F}_{p^m}$, we have 
    \begin{flalign*}
\hspace{15pt} \operatorname{Tr}_{p}^{p^m}\left((x+1)^{p^{2h}-p^{h}+1}\right)
&= \Biggl(m
 + (hp-h+1)\operatorname{Tr}_{p}^{p^m}\left(x\right)
 + \sum_{j\in\widehat{\Gamma_{(1)}}\setminus\{1\}}
   h\beta(j)\operatorname{Tr}_{p}^{p^m}\left(x^{j}\right) \\
&\quad
 + \sum_{j\in\widehat{\Gamma_{(2)}}\setminus\widehat{\Gamma_{(1)}}}
   (h-1)\beta(j)\operatorname{Tr}_{p}^{p^m}\left(x^{j}\right)
 + \Biggl(\sum_{j\in\widehat{\Gamma_{(2)}}\setminus\widehat{\Gamma_{(1)}}}
   (h-2)\beta(j)\operatorname{Tr}_{p}^{p^m}\left(x^{j+p^{2}}\right) \\
&\quad
 + \sum_{j\in\widehat{\Gamma_{(1)}}}
   (h-1)\beta(j)\operatorname{Tr}_{p}^{p^{m}}\left(x^{j+p^{2}}\right)
 + \sum_{j\in\widehat{\Gamma_{(2)}}}\sum_{k=2}^{p-1}
   (h-2)\beta(j)\operatorname{Tr}_{p}^{p^{m}}\left(x^{j+kp^{2}}\right)\Biggr)
 + \cdots \\
&\quad
 + \Biggl(\sum_{j\in\widehat{\Gamma_{(h-1)}}\setminus\widehat{\Gamma_{(h-2)}}}
   \beta(j)\operatorname{Tr}_{p}^{p^m}\left(x^{j+p^{h-1}}\right)
 + \sum_{j\in\widehat{\Gamma_{(h-2)}}}
   2\beta(j)\operatorname{Tr}_{p}^{p^{m}}\left(x^{j+p^{h-1}}\right) \\
&\quad
 + \sum_{j\in\widehat{\Gamma_{(h-1)}}}\sum_{k=2}^{p-1}
   \beta(j)\operatorname{Tr}_{p}^{p^{m}}\left(x^{j+kp^{h-1}}\right)\Biggr)
 + 2(p-1)\operatorname{Tr}_{p}^{p^m}\left(x^{1+p^{h}}\right) \\
&\quad
 + \sum_{j\in\widehat{\Gamma_{(h-1)}}\setminus\{1\}}
   \beta(j)\operatorname{Tr}_{p}^{p^m}\left(x^{j+p^{h}}\right)
 + \sum_{j\in\widehat{\Gamma_{(h)}}\setminus\{1\}}
   \beta(j)\operatorname{Tr}_{p}^{p^m}\left(x^{j+p^{h+1}}\right)\Biggr)
 \pmod{p}, &
\end{flalign*}
    where the map $\beta(\cdot)$ is defined as in Eq. $(\ref{Eq16})$.
    \begin{proof}
        Since $h=\frac{m-1}{2}$, for any $x\in\mathbb{F}_{p^m}$, Eq. $(\ref{Eq20})$ gives
        \begin{flalign}\label{Eq21}
            \hspace{15pt} \operatorname{Tr}_{p}^{p^m}\left((x+1)^{p^{2h}-p^{h}+1}\right) &= \operatorname{Tr}_{p}^{p^m}\left(1+\sum_{j\in\widehat{D_{h}(p)}}\beta(j)x^{j}\right) + \operatorname{Tr}_{p}^{p^m}\left(x+\sum_{j\in\widehat{D_{h}(p)}}\beta(j)x^{j+p^{h+1}}\right) 
        \end{flalign}
        Note that $\beta(j)^{p}\equiv\beta(j)\text{ (mod $p$)}$ for any $j\in\widehat{D_{h}(p)}$ and $\operatorname{Tr}_{p}^{p^m}(x^{p})=\operatorname{Tr}_{p}^{p^m}(x)$ for all $x\in\mathbb{F}_{p^m}$. Using the fact $\widehat{D_{h}(p)}=\widehat{\Gamma_{(h)}}\cup p\widehat{\Gamma_{(h-1)}}\cup\cdots\cup p^{h-1}\widehat{\Gamma_{(1)}}$, we can write 
        \begin{flalign}\label{Eq22}
            \operatorname{Tr}_{p}^{p^m}\left(x+\sum_{j\in\widehat{D_{h}(p)}}\beta(j)x^{j+p^{h+1}}\right) &= \operatorname{Tr}_{p}^{p^m}\Biggl(x+\sum_{j\in\widehat{\Gamma_{(h)}}}\beta(j)x^{j+p^{h+1}}+\sum_{j\in\widehat{\Gamma_{(h-1)}}}\beta(j)x^{j+p^{h}}+ \cdots \nonumber \\
&\quad + \sum_{j\in\widehat{\Gamma_{(2)}}}\beta(j)x^{j+p^{3}} +\sum_{j\in\widehat{\Gamma_{(1)}}}\beta(j)x^{j+p^{2}}\Biggr) &
        \end{flalign}
        For a positive integer $t=3,4,\cdots,h$, using the fact $\widehat{\Gamma_{(t)}}\setminus\widehat{\Gamma_{(t-1)}}=\bigcup_{k=1}^{p-1}\{j+kp^{t-1}:j\in\widehat{\Gamma_{(t-1)}}\}$, and with the help of Lemma $\ref{L14}$, Eq. $(\ref{Eq21})$ and $(\ref{Eq22})$, we obtain
        \begin{align}\label{Eq24}
         \operatorname{Tr}_{p}^{p^m}\left((x+1)^{p^{2h}-p^{h}+1}\right)   &= \Biggl(
    m
    + \sum_{j\in\widehat{\Gamma_{(1)}}}
        h\,\beta(j)\operatorname{Tr}_{p}^{p^m}\left(x^{j}\right) 
    + \sum_{j\in\widehat{\Gamma_{(2)}}\setminus\widehat{\Gamma_{(1)}}}
        (h-1)\,\beta(j)\operatorname{Tr}_{p}^{p^m}\left(x^{j}\right)
    + \cdots \nonumber \\
&\quad
    + \sum_{j\in\widehat{\Gamma_{(h-1)}}\setminus\widehat{\Gamma_{(h-2)}}}
        2\,\beta(j)\operatorname{Tr}_{p}^{p^m}\left(x^{j}\right) 
    + \sum_{j\in\widehat{\Gamma_{(h)}}\setminus\widehat{\Gamma_{(h-1)}}}
        \beta(j)\operatorname{Tr}_{p}^{p^m}\left(x^{j}\right)
\Biggr)\text{ (mod $p$)} \nonumber \\ + &\quad
\Biggl(\operatorname{Tr}_{p}^{p^m}\left(x\right)+\sum_{j\in\widehat{\Gamma_{(h)}}}\beta(j)\operatorname{Tr}_{p}^{p^m}\left(x^{j+p^{h+1}}\right)+\sum_{j\in\widehat{\Gamma_{(h-1)}}}\beta(j)\operatorname{Tr}_{p}^{p^m}\left(x^{j+p^{h}}\right)+\cdots \nonumber \\ &\quad +\sum_{j\in\widehat{\Gamma_{(2)}}}\beta(j)\operatorname{Tr}_{p}^{p^m}\left(x^{j+p^{3}}\right)+\sum_{j\in\widehat{\Gamma_{(1)}}}\beta(j)\operatorname{Tr}_{p}^{p^m}\left(x^{j+p^2}\right) \Biggr)  
        \end{align}
        Note that $\beta(1)=\binom{p-1}{1}=p-1$, so the coefficient of $\operatorname{Tr}_{p}^{p^m}(x)$ on the right-hand side of Eq. $(\ref{Eq24})$ becomes $(hp-h+1)$. From Lemma $\ref{L19}$, it is clear that $\operatorname{Tr}_{p}^{p^m}(x^{1+p^{h+1}})=\operatorname{Tr}_{p}^{p^m}(x^{1+p^h})$, and for every $t=3,4,\cdots,h$, we can write
          \[
\begin{aligned}
&\sum_{j\in\widehat{\Gamma_{(t)}}\setminus\widehat{\Gamma_{(t-1)}}}
   (h-t+1)\beta(j)\operatorname{Tr}_{p}^{p^m}\left(x^{j}\right)
 + \sum_{j\in\widehat{\Gamma_{(t-2)}}}
   \beta(j)\operatorname{Tr}_{p}^{p^{m}}\left(x^{j+p^{t-1}}\right) \\
   &\qquad = \sum_{j\in\widehat{\Gamma_{(t-1)}}}\sum_{k=1}^{p-1}(h-t+1)\beta(j)\operatorname{Tr}_{p}^{p^{m}}\left(x^{j+kp^{t-1}}\right) + \sum_{j\in\widehat{\Gamma_{(t-2)}}}
   \beta(j)\operatorname{Tr}_{p}^{p^{m}}\left(x^{j+p^{t-1}}\right) 
\\ & \qquad =
  \sum_{j\in\widehat{\Gamma_{(t-1)}}\setminus\widehat{\Gamma_{(t-2)}}}
   (h-t+1)\beta(j)\operatorname{Tr}_{p}^{p^m}\left(x^{j+p^{t-1}}\right)
 +\sum_{j\in\widehat{\Gamma_{(t-2)}}}(h-t+2)\beta(j)\operatorname{Tr}_{p}^{p^{m}}\left(x^{j+p^{t-1}}\right)
 \\ &\qquad  
 + \sum_{j\in\widehat{\Gamma_{(t-1)}}}\sum_{k=2}^{p-1}
     (h-t+1)\beta(j)\operatorname{Tr}_{p}^{p^{m}}\left(x^{j+kp^{t-1}}\right). &
\end{aligned}
\]
Hence, the conclusion follows by adding the terms on the right-hand side of Eq. $(\ref{Eq24})$ under modulo $p$.
    \end{proof}
\end{lemma}

\begin{theorem}\label{Th6}
    Let $s^{\infty}$ be defined as in Eq. $(\ref{Eq3})$ through the monomial $F_{7}(x)=x^{p^{2h}-p^{h}+1}$, $h=\frac{m-1}{2}$ over $\mathbb{F}_{p^m}$, and $\widehat{\Gamma_{(t)}}$ be defined as in Eq. $(\ref{Eq10})$. Then the linear span $\mathcal{L}_{s}$ and the minimal polynomial $\mathfrak{M}_{s}(x)$ of $s^{\infty}$ are given by
    \begin{flalign*}
        \mathcal{L}_{s} &=\begin{cases}
            \mathbb{N}_{p}(3)+6p-9;\text{ for }h=1, \\[6pt]
           % \vspace{1 em}
           \mathbb{N}_{p}(m)+\Biggl(\mathbb{N}_{p}(hp-h+1)+(p-2)\mathbb{N}_{p}(h)+(p-1)^{2}\mathbb{N}_{p}(h-1)+\\ \displaystyle\sum_{t=2}^{h-1}(p-1)p^{t-2}\Big(\mathbb{N}_{p}(h+1-t)+(2p^2-4p+1)\mathbb{N}_{p}(h-t)\Big)+(p-1)(p^{h-1}+p^{h-2})-1\Biggr)\cdot m ;\text{ for }h\geq 2
        \end{cases}
    \end{flalign*}
    and 
    \begin{flalign*}
 \hspace{15pt}       \mathfrak{M}_{s}(x) &= \begin{cases}
            (x-1)^{\mathbb{N}_{p}(3)}\displaystyle\prod_{j\in\widehat{\Gamma_{(1)}}\setminus\{1\}}m_{\alpha^{-j}}(x)\prod_{j\in\widehat{\Gamma_{(1)}}}m_{\alpha^{-j-p^{2}}}(x);  \hspace{3cm}\text{ for }h=1,  \\[10pt]
            \\
        (x-1)^{\mathbb{N}_{p}(m)}
 \left(m_{\alpha^{-1}}(x)\right)^{\mathbb{N}_{p}(hp-h+1)}
 \displaystyle\prod_{\substack{j\in\widehat{\Gamma_{(1)}}\setminus\{1\} \\ \mathbb{N}_{p}(h)=1}}
   m_{\alpha^{-j}}(x) \prod_{\substack{j\in\widehat{\Gamma_{(2)}}\setminus\widehat{\Gamma_{(1)}} \\ \mathbb{N}_{p}(h-1)=1}}
   m_{\alpha^{-j}}(x)\times \\
 \displaystyle\prod_{t=2}^{h-1} \Biggl(\prod_{j\in\widehat{\Gamma_{(t-1)}}}\big(m_{\alpha^{-j-p^{t}}}(x)\big)^{\mathbb{N}_{p}(h+1-t)} \displaystyle\prod_{j\in\widehat{\Gamma_{(t)}}\setminus\widehat{\Gamma_{(t-1)}}}\big(m_{\alpha^{-j-p^{t}}}(x)\big)^{\mathbb{N}_{p}(h-t)}\times \\ \Big(\prod_{k=2}^{p-1}\prod_{j\in\widehat{\Gamma_{(t)}}}m_{\alpha^{-j-kp^{t}}}(x)\Big)^{\mathbb{N}_{p}(h-t)}\Biggr)  m_{\alpha^{-1-p^h}}(x)\displaystyle\prod_{j\in\widehat{\Gamma_{(h-1)}}\setminus\{1\}}m_{\alpha^{-j-p^{h}}}(x) \times \\   \prod_{j\in\widehat{\Gamma_{(h)}}\setminus\{1\}}m_{\alpha^{-j-p^{h+1}}}(x); \hspace{6.3cm}\text{ for }h\geq 2.    
        \end{cases}   &
    \end{flalign*}
    \begin{proof}
        With the aid of Lemma $\ref{L20}$, the sequence $s^{\infty}$ defined in Eq. $(\ref{Eq3})$ can be expanded as follows
        \begin{flalign}\label{Eq25}
     \hspace{15pt}      s_{t} &= \operatorname{Tr}_{p}^{p^m}\left(F_{7}(\alpha^{t}+1)\right) \nonumber \\
            &= \Biggl(m
 + (hp-h+1)\operatorname{Tr}_{p}^{p^m}\left(\alpha^{t}\right)
 + \sum_{j\in\widehat{\Gamma_{(1)}}\setminus\{1\}}
   h\beta(j)\operatorname{Tr}_{p}^{p^m}\left(\left(\alpha^{t}\right)^{j}\right) 
 + \sum_{j\in\widehat{\Gamma_{(2)}}\setminus\widehat{\Gamma_{(1)}}}
   (h-1)\beta(j)\operatorname{Tr}_{p}^{p^m}\left(\left(\alpha^{t}\right)^{j}\right)  \nonumber \\
   &\quad
   + \Biggl(\sum_{j\in\widehat{\Gamma_{(2)}}\setminus\widehat{\Gamma_{(1)}}}
   (h-2)\beta(j)\operatorname{Tr}_{p}^{p^m}\left(\left(\alpha^{t}\right)^{j+p^{2}}\right) 
 + \sum_{j\in\widehat{\Gamma_{(1)}}}
   (h-1)\beta(j)\operatorname{Tr}_{p}^{p^{m}}\left(\left(\alpha^{t}\right)^{j+p^{2}}\right)
 + \nonumber \\
&\quad \sum_{j\in\widehat{\Gamma_{(2)}}}\sum_{k=2}^{p-1}
   (h-2)\beta(j)\operatorname{Tr}_{p}^{p^{m}}\left(\left(\alpha^{t}\right)^{j+kp^{2}}\right)\Biggr)
 + \cdots 
 + \Biggl(\sum_{j\in\widehat{\Gamma_{(h-1)}}\setminus\widehat{\Gamma_{(h-2)}}}
   \beta(j)\operatorname{Tr}_{p}^{p^m}\left(\left(\alpha^{t}\right)^{j+p^{h-1}}\right) \nonumber \\
&\quad
 + \sum_{j\in\widehat{\Gamma_{(h-2)}}}
   2\beta(j)\operatorname{Tr}_{p}^{p^{m}}\left(\left(\alpha^{t}\right)^{j+p^{h-1}}\right) 
 + \sum_{j\in\widehat{\Gamma_{(h-1)}}}\sum_{k=2}^{p-1}
   \beta(j)\operatorname{Tr}_{p}^{p^{m}}\left(\left(\alpha^{t}\right)^{j+kp^{h-1}}\right)\Biggr)
 + 2(p-1)\operatorname{Tr}_{p}^{p^m}\left(\left(\alpha^{t}\right)^{1+p^{h}}\right) \nonumber \\
&\quad
 + \sum_{j\in\widehat{\Gamma_{(h-1)}}\setminus\{1\}}
   \beta(j)\operatorname{Tr}_{p}^{p^m}\left(\left(\alpha^{t}\right)^{j+p^{h}}\right)
 + \sum_{j\in\widehat{\Gamma_{(h)}}\setminus\{1\}}
   \beta(j)\operatorname{Tr}_{p}^{p^m}\left(\left(\alpha^{t}\right)^{j+p^{h+1}}\right)\Biggr)
 \pmod{p};\text{ for all }t\geq 0. &
        \end{flalign}
        From Lemma $\ref{L18}$$(3)$ and $\ref{L19}$, it can be verified that there is no overlap between any two terms on the right-hand
side of Eq. $(\ref{Eq25})$. By analyzing the sizes of the $p$-cyclotomic cosets from Lemma $\ref{L3}$ and $\ref{L18}$$(1)$, we determine the exact number of terms in which each trace-term on the right-hand side of Eq. $(\ref{Eq25})$ would break as $\operatorname{Tr}_{p}^{p^m}(\alpha^t)=\frac{m}{|C_{t}|}\sum_{j=0}^{|C_{t}|-1}(\alpha^t)^{p^j}$ for every $0\leq t\leq p^m-2$.
\vskip 1pt
The desired conclusions on the linear span $\mathcal{L}_{s}$ and the minimal polynomial $\mathfrak{M}_{s}(x)$ of $s^{\infty}$ then follow
from Lemma $\ref{L6}$ and Eq. $(\ref{Eq25})$.
    \end{proof}
\end{theorem}

\begin{theorem}\label{Th11}
    Let the code $\mathcal{C}_{F_{7}}$ be defined by the sequence $s^{\infty}$ through the monomial $F_{7}(x)=x^{p^{2h}-p^{h}+1}$, $h=\frac{m-1}{2}$ over $\mathbb{F}_{p^m}$. Then $\mathcal{C}_{F_{7}}$ has parameters $[p^m-1,p^m-1-\mathcal{L}_{s},d]$ over $\mathbb{F}_{p}$ with the generator polynomial $\mathfrak{M}_{s}(x)$, where $\mathcal{L}_{s}$ and $\mathfrak{M}_{s}(x)$ are given in Theorem $\ref{Th6}$. In addition, 
    \begin{flalign*}
     \hspace{15pt}   d &\geq \begin{cases}
            p^{h}+p^{h-1};\text{ if }p\nmid m \text{ and }1<h<p, \\
            p^{h}+p^{h-1}-1;\text{ if }p\mid m \text{ and }1<h<p, \\
            p^{h}+2;\text{ if }h=p, \\
            p^{p}+1;\text{ if }h>p \text{ and }p\nmid(h-1), \\
            p^{p};\text{ if }h>p \text{ and }p\mid(h-1).
        \end{cases} &
    \end{flalign*}
    \begin{proof}
        The dimension of the code $\mathcal{C}_{F_{7}}$ follows from the linear span $\mathcal{L}_{s}$ determined in Theorem $\ref{Th6}$. We now determine the lower bounds on the minimum weight of the code $\mathcal{C}_{F_{7}}$. 
        \par When $p\nmid m$ and $1<h<p$, one can note that $\mathbb{N}_{p}(m)=1$, $\mathbb{N}_{p}(hp-h+1)=\mathbb{N}_{p}(h-1)$, and $\mathbb{N}_{p}(t)=1$ hold for all $t=1,2,\cdots, h$. Then it can be verified that $\alpha^{-j}$ for all $j\in\{0\}\cup \widehat{D_{h}(p)}\cup\{p^{h}\}\cup\{j+p^{h}:j\in\widehat{\Gamma_{(h-1)}}\}$ are the zeros of the generator polynomial $\mathfrak{M}_{s}(x)$ of $\mathcal{C}_{F_{7}}$ in Theorem $\ref{Th6}$. Hence, the BCH bound gives $d\geq p^{h}+p^{h-1}$. In the case of $p\mid m$ and $1<h<p$, the lower bounds on the minimum weight $d$ can
        be achieved in a similar manner.
        \par When $h=p$, it can be noted that $\mathbb{N}_{p}(h)=0$ and $\mathbb{N}_{p}(t)=1$ for all $t=1,2,\cdots,h-1$. Since $\mathbb{N}_{p}(hp-h+1)=\mathbb{N}_{p}(h-1)=1$ and $m_{\alpha^{-j-p^{h+1}}}(x)=\prod_{s\in C_{j+p^{h+1}}}(1-\alpha^{s}x)$, for all $j\in\widehat{D_{h}(p)}$, one can then verify that $\alpha^{-j}$, where $j\in\{p^{h+1},p^{h+1}+1,\cdots,p^{h+1}+p^{h}-1,p^{h+1}+p^{h}\}$ are the zeros of the generator polynomial $\mathfrak{M}_{s}(x)$ of $\mathcal{C}_{F_{7}}$ in Theorem $\ref{Th6}$. Hence, the BCH bound gives $d\geq p^{h}+2$.
        
        \par When $h>p$ and $p\nmid(h-1)$, one can note that $\mathbb{N}_{p}(hp-h+1)=\mathbb{N}_{p}(h-1)=1$. As $p\geq 3$, $\mathbb{N}_{p}(h+1-t)=0$ occurs only if $t=h+1-p,h+1-2p,\cdots,h+1-\lfloor\frac{h}{p}\rfloor p$, while $t$ varies through $2,3,\cdots,h-1$. It is not difficult to verify that $\alpha^{-j}$ for all $j\in\{p^{h+1},p^{h+1}+1,\cdots,p^{h+1}+p^{p}-1\}$ are the zeros of the generator polynomial $\mathfrak{M}_{s}(x)$ of $\mathcal{C}_{F_{7}}$ in Theorem $\ref{Th6}$. Hence, the BCH bound gives $d\geq p^{p}+1$. In the case of $h>p$ and $p\mid (h-1)$, the lower bounds on the minimum weight $d$ can be achieved in a similar manner.
    \end{proof}
\end{theorem}

\begin{example}
    Let $p=3$ and $m=3$; then $h=\frac{m-1}{2}=1$ and $\mathbb{N}_{p}(m)=0$. Let $\alpha$ be a root of the primitive polynomial $x^3 + 2x + 1$ over $\mathbb{F}_{3}$. The generator polynomial of the cyclic code $\mathcal{C}_{F_{7}}$ is $\mathfrak{M}_{s}(x)=x^9 + x^8 + x^6 + 1$. Then $\mathcal{C}_{F_{7}}$ is a $[26, 17, 4]$ ternary cyclic code, and its dual $\mathcal{C}_{F_{7}}^{\perp}$ is a ternary $[26, 9, 9]$ cyclic code. %The optimal ternary linear code with parameter $[26, 16, 6]$ in the Database $\textnormal{\cite{Database}}$ is not cyclic.
\end{example}
\vspace{1em}
\begin{example}
     Let $p=5$ and $m=3$; then $h=\frac{m-1}{2}=1$ and $\mathbb{N}_{p}(m)=1$. Let $\alpha$ be a root of the primitive polynomial $x^3 + 3x + 3$ over $\mathbb{F}_{5}$. The generator polynomial of the cyclic code $\mathcal{C}_{F_{7}}$ is $\mathfrak{M}_{s}(x)=x^{22} + 2x^{21} +
    x^{19} + 3x^{17} + 4x^{16} + 3x^{15} + 2x^{14} + 2x^{13} + 3x^{12} + 3x^{10} + 2x^9 + 4x^8 + 4x^7 + 4x^6 + 3x^5 + 4x^4 + 3x^3 + x^2 + 4x + 2$. Then $\mathcal{C}_{F_{7}}$ is a $[124,102,8]$ quinary cyclic code, and its dual $\mathcal{C}_{F_{7}}^{\perp}$ is a quinary $[124,22,62]$ cyclic code.
\end{example}
\vspace{1em}
\begin{example}
    Let $p=3$ and $m=5$; then $h=\frac{m-1}{2}=2$ and $\mathbb{N}_{p}(m)=\mathbb{N}_{p}(hp-h+1)=1$. Let $\alpha$ be a root of the primitive polynomial $x^5 + 2x + 1$ over $\mathbb{F}_{3}$. The generator polynomial of the cyclic code $\mathcal{C}_{F_{7}}$ is $\mathfrak{M}_{s}(x)=x^{66} + x^{64} +
    2x^{63} + 2x^{62} + x^{61} + 2x^{58} + 2x^{56} + x^{55} + x^{53} + 2x^{52} + 2x^{49} +
    2x^{48} + x^{47} + x^{45} + x^{44} + x^{43} + x^{42} + x^{40} + 2x^{39} + 2x^{38} + x^{34} +
    x^{33} + 2x^{32} + 2x^{31} + x^{30} + x^{29} + x^{28} + 2x^{27} + 2x^{26} + x^{25} +
    2x^{24} + x^{21} + 2x^{20} + x^{19} + x^{17} + x^{15} + x^{14} + x^{13} + x^{11} + 2x^8 +
    x^6 + x^5 + x^4 + 2x^3 + 2x^2 + x + 2$. Then $\mathcal{C}_{F_{7}}$ is a $[242,176,d]$, where $10\leq d\leq 23$, ternary cyclic code, and its dual $\mathcal{C}_{F_{7}}^{\perp}$ is a ternary $[242,66,d^{\perp}]$, where $33\leq d^{\perp}\leq 80$, cyclic code.
\end{example}
\vspace{1em}
\begin{remark}
    In \textnormal{\cite[Section $3.5$]{Bose}}, Bose, Parampalli, and Singh utilized the monomial $F_{7}(x)=x^{2^{2h}-2^{h}+1}$ over $\mathbb{F}_{2^m}$ with $h=\frac{m-1}{2}$ and investigated the cyclic code $\mathcal{C}_{F_{7}}$ to determine its generator polynomial and the bounds of its minimum weight. In this subsection, we chose a more general scenario by employing the monomial $F_{7}(x)=x^{p^{2h}-p^{h}+1}$ over $\mathbb{F}_{p^m}$ for $p\geq 3$ and $h=\frac{m-1}{2}$, and studied the cyclic code $\mathcal{C}_{F_{7}}$. Consequently, Theorem $\ref{Th11}$ demonstrates that the lower bound on the minimum distance of $\mathcal{C}_{F_{7}}$ is significantly better than the lower bound determined for the binary case in Theorem $5$ of \textnormal{\cite{Bose}}.  
\end{remark}
%\vspace{1em}
\section{Concluding remarks}\label{sec4}
Fascinated by the work of Ding \cite{P3}, this paper investigates $q$-ary cyclic codes by employing suitable power functions with known differential uniformity over odd characteristic finite fields. We constructed some infinite families of non-binary cyclic codes of length $q^m-1$ with dimensions larger than $(q^m-1)/2$ and minimum distance better than the square-root of the code length. Some of the constructed families of codes possess a higher probability of being optimal or near-optimal under certain conditions. We determine the exact minimum distance of some of these codes. In addition to these findings, we also partially solve an open problem raised by Ding \cite{P3} regarding the determination of the dimension and the generator polynomial of a ternary cyclic code (discussed in Remark \ref{Remark1}). Readers interested in working on this topic are invited to explore more suitable power functions or permutation polynomials over finite fields of odd characteristic or to introduce new strategies for determining the linear span of sequences in constructing infinite families of cyclic codes.  
\par Some families of $q$-ary cyclic codes presented in this paper are closely related to the primitive BCH codes, they could be used in constructing quantum error correcting codes \cite{P-1,Quantum2}. 



\end{document}